\documentclass[final-notforauthors]{fundam}

\usepackage{amssymb}  %
\usepackage{mathtools}
\usepackage{environ}

\iffalse
\newcommand{\pg}[1]{\marginpar{\textcolor{orange}{pg}}\textcolor{orange}{#1}}
\newcommand{\pv}[1]{\marginpar{\textcolor{blue}{pv}}\textcolor{blue}{#1}}
\newcommand{\eg}[1]{\marginpar{\textcolor{magenta}{eg}}\textcolor{magenta}{#1}}
\else
\newcommand{\pg}[1]{{}}
\newcommand{\pv}[1]{{}}
\newcommand{\eg}[1]{{}}
\fi

\newcommand{\tree}[2]{\tikz[style={
    level distance=#1pt,
    sibling distance=0pt,
    frontier/.style={distance from root=#2pt},
    font=\small}]}

\newcommand{\tmove}[1]{\tikz{\node at (0,0) {};
\node at (0,#1pt) {\(\ra_{\cA}\)};}}

\newcommand{\tmoveT}[1]{\tikz{\node at (0,0) {};
\node at (0,#1pt) {\(\ra_{\TTA}\)};}}

\newcommand{\tdot}[1]{\tikz{\node at (0,0) {};
\node at (0,#1pt) {\(\cdot\)};}}

\newcommand{\tuple}[1]{\langle {#1}\rangle}
\newcommand{\lang}[1]{{\mathcal{L}(#1)}}

\newcommand{\subs}[2]{{\ensuremath{#1[\![#2]\!]}}}
\newcommand{\concat}{{\cdot}}

\newcommand{\AND}{\ensuremath{\land}}
\newcommand{\OR}{\ensuremath{\lor}}
\newcommand{\wt}{\widetilde}
\newcommand{\cA}{\mathcal{A}}

\newcommand{\cC}{\mathcal{C}}
\newcommand{\cD}{\mathcal{D}}
\newcommand{\cF}[1]{\mathsf{Min}^{\mathsf{#1}}}
\newcommand{\cG}[1]{\mathsf{Det}^{\mathsf{#1}}}
\newcommand{\cJ}[1]{\mathsf{TDet}^{\mathsf{#1}}}
\newcommand{\cK}[1]{\mathsf{TMin}^{\mathsf{#1}}}

\newcommand{\len}[1]{\vert{#1}\vert}

\newcommand{\rank}[1]{\langle{#1}\rangle}
\newcommand{\cH}{\mathsf{H}}
\newcommand{\cI}{\mathcal{I}}
\newcommand{\cM}{\mathcal{M}}
\newcommand{\cL}{\mathcal{L}}
\newcommand{\cN}{\mathcal{N}}

\newcommand{\cT}{\mathcal{T}}

\newcommand{\TA}{\mathcal{A}}
\newcommand{\BTA}{\mathcal{A}^{\mathsf{B}}}
\newcommand{\TTA}{\mathcal{A}^{\mathsf{T}}}

\newcommand{\cX}{
\mathchoice
{\scalebox{0.85}{\raisebox{-0.3pt}{\(\square\)}}}
{\scalebox{0.85}{\raisebox{-0.3pt}{\(\square\)}}}
{\scalebox{0.7}{\raisebox{-0.3pt}{\(\square\)}}}
{\scalebox{0.5}{\raisebox{-0.3pt}{\(\square\)}}}
}

\newcommand{\rd}{\sim^{\mathsf{d}}}
\newcommand{\ru}{\sim^{\mathsf{u}}}
\newcommand{\ruA}{\ru_{\cA}}
\newcommand{\rdA}{\rd_{\cA}}
\newcommand{\ruAT}{\ru_{{\TTA}}}
\newcommand{\rdAT}{\rd_{{\TTA}}}
\newcommand{\rdAB}{\rd_{{\BTA}}}
\newcommand{\ruAB}{\ru_{{\BTA}}}
\newcommand{\dashdownarrowShort}{\raisebox{2.0ex}{\rotatebox{-90}{\scalebox{0.60}{$\leadsto$}}}}
\newcommand{\simp}{\sim_{\dashdownarrowShort}}

\newcommand{\true}{\textbf{T}}
\newcommand{\false}{\textbf{F}}

\newcommand{\db}{\mathsf{D}}
\newcommand{\cdb}{\mathsf{cD}}

\newcommand{\udiff}{\stackrel{\rm\scriptscriptstyle{\fontfamily{ptm}\selectfont def}}{\Leftrightarrow}}
\newcommand{\udiffg}{\stackrel{\rm\scriptscriptstyle{\fontfamily{ptm}\selectfont def}}{=}}
\newcommand{\Lra}{\Leftrightarrow}
\newcommand{\Ra}{\Rightarrow}
\newcommand{\La}{\Leftarrow}
\newcommand{\ra}{\rightarrow}

\newcommand{\ua}{\uparrow}
\newcommand{\da}{\downarrow}
\newcommand{\mydots}{\hbox to 0.5em{.\hss.}}
\newcommand{\myldots}{\hbox to 0.5em{.\hss.\hss.}}
\newcommand{\range}[1]{1\mydots#1}

\DeclareMathOperator{\leaf}{{\ensuremath{\ell}}}
\DeclareMathOperator{\dom}{{dom}}
\DeclareMathOperator{\height}{h}
\DeclareMathOperator{\piv}{piv}
\DeclareMathOperator{\pre}{pre}
\DeclareMathOperator{\wpre}{wpre}
\DeclareMathOperator{\post}{post}
\DeclareMathOperator{\initials}{i}
\DeclareMathOperator{\finals}{f}

\newcommand{\eox}{\hfill{\ensuremath{\Diamond}}}

\usepackage{tikz}
\usetikzlibrary{arrows}
\usetikzlibrary{positioning}
\usepackage{tikz-cd}
   \usepackage{tikz-qtree,tikz-qtree-compat}
  \usetikzlibrary{decorations.pathmorphing}
\usepackage{newunicodechar}
\newunicodechar{ε}{\varepsilon}
\newunicodechar{α}{\alpha}
\newunicodechar{φ}{\varphi}
\newunicodechar{ψ}{\psi}
\newunicodechar{δ}{\delta}
\newunicodechar{μ}{\mu}
\newunicodechar{π}{\pi}
\newunicodechar{ρ}{\rho}
\newunicodechar{τ}{\tau}
\newunicodechar{Σ}{\Sigma}
\newunicodechar{≤}{\leq}
\newunicodechar{≥}{\geq}
\newunicodechar{≝}{\stackrel{\rm\scriptscriptstyle def}{=}}
\newunicodechar{𝒪}{\mathcal{O}}
\newunicodechar{𝒯}{\mathcal{T}}
\newunicodechar{𝒴}{\mathcal{Y}}
\newunicodechar{Δ}{\Delta}
   \usepackage{hyperref}

\begin{document}

   \setcounter{page}{1}
   \publyear{2021}
   \papernumber{2091}
   \volume{184}
   \issue{1}


\title{A Congruence-Based Perspective on Finite Tree Automata}

\author{Pierre Ganty\thanks{Address of correspondence: IMDEA Software Institute,  Madrid, Spain. \newline \newline
          \vspace*{-6mm}{\scriptsize{Received April 2021; \ revised December 2021.}}},  Elena Guti\'errez, Pedro Valero
    \\
IMDEA Software Institute\\
Madrid, Spain\\
pierre.ganty@imdea.org, elenagutiv@gmail.com, pevalme@fb.com
}

\maketitle

\runninghead{P. Ganty et al.}{A Congruence-Based Perspective on Finite Tree Automata}

\begin{abstract}
We provide new insights on the determinization and minimization of tree automata using congruences on trees.
From this perspective, we study a Brzozowski's style minimization algorithm for tree automata.
First, we prove correct this method relying on the following fact: when the automata-based and the language-based congruences coincide, determinizing the automaton yields the minimal one.
Such automata-based congruences, in the case of word automata, are defined using pre and post operators.
Now we extend these operators to tree automata, a task that is particularly challenging due to the reduced expressive power of deterministic top-down (or equivalently co-deterministic bottom-up) automata.
We leverage further our framework to offer an extension of the original result by Brzozowski for word automata.
\end{abstract}

\section{Introduction}
\label{sec:introduction}

Finite tree automata are a well-studied~\cite{tata2007,gcseg1978minimal,gcseg2015tree} automata model processing tree structures, as opposed to the classical finite-state automata, which process words, i.e., trees where every node has at most one child.
Examples of applications of tree automata include model checking~\cite{AbdullaLdR06,bouajjani2012abstract}, natural language processing~\cite{knight2009applications} and representing the nested structured of tree-based metalanguages, such as XML~\cite{hosoya2010foundations}.

There exist two classes of tree automata, which differ in the way they process the input tree: \emph{bottom-up tree automata} (BTAs) and \emph{top-down tree automata} (TTAs).
In their non-deterministic flavor, both types of tree automata have the same expressive power: they are both finite representations of the \emph{regular tree languages}.

Like word automata, tree automata (both BTAs and TTAs) can be deterministic (DBTAs and DTTAs) or non-deterministic, offering the classical trade-off, where deterministic automata are easier to reason about while non-deterministic ones are more concise.
This situation has motivated the study of techniques to reduce the number of states of deterministic automata~\cite{abdulla2007bisimulation,almeida2016reduction,hogberd2009backward} as well as methods for building deterministic automata that are \emph{minimal} in the number of states~\cite{brainerd68minimalization,Bjorklund2016Taxonomy,nivat1997minimal}.
For both word and tree automata the minimal deterministic automaton is unique (up to isomorphisms).

Unlike word automata, where every regular language is definable by a deterministic automaton, there exist regular tree languages that cannot be defined by a deterministic TTA (equivalently, they cannot be defined by a co-deterministic BTA).
It turns out that DTTAs and co-deterministic BTAs (co-DBTAs) define a subclass of regular tree languages called \emph{path-closed languages}~\cite{Viragh81}.

In this paper, we address a Brzozowski's style algorithm for minimizing TAs~\cite{Bjorklund2016Taxonomy}.
That is, given a BTA or a TTA defining the language \(L\), the algorithm combines a co-determinization and a determinization operation to build the minimal DBTA or DTTA for \(L\).
In this sense, the algorithm works in the same fashion as the classical Brzozowski's algorithm~\cite{brzozowski1962canonical} (also known as the double-reversal method) for finding the minimal deterministic word automaton for a given regular language.
Note that in the tree case, the method only applies to tree automata generating path-closed tree languages since it requires the construction of a co-DBTA or a DTTA for the given language.

Brzozowski's algorithm relies on the fact that determinizing a co-deterministic word automaton returns the minimal deterministic automaton.
The reason why this method is also coined as the \emph{double-reversal} method is that it builds a co-deterministic automaton for the input language by combining a reverse, determinization and reverse operation.

\subsection{Contributions}%
\label{par:contributions}

We study Brzozowski's minimization algorithm for tree automata from the perspective of \emph{congruences} on trees.
In this sense, we build on work by Ganty et al.~\cite{ganty2019congruence} who applied congruences on words to the study of word automata minimization techniques.
In this work, we use congruences of finite index on \emph{trees} and \emph{contexts}.
The latter are trees for which exactly one node is labeled by a distinguished symbol of arity \(0\).

Concretely, we use so-called \emph{upward congruences}~\cite{kozen1993MyhillTrees,tata2007} over trees, which are equivalences on trees that behave well w.r.t. the concatenation of symbols on top of the tree; and \emph{downward} congruences~\cite{nivat1997minimal}, i.e., equivalences on contexts that behave well when concatenating contexts on the leaves.
Given a tree language \(L\) and a congruence satisfying some required properties, we show how to build deterministic and co-deterministic automata defining \(L\).

We leverage two kinds of congruences: \emph{language-based} congruences, defined in terms of a regular tree language, and \emph{automaton-based} congruences, relative to an automaton.
Hence we show how to use upward automata-based congruences to construct DBTAs and DTTAs, which are isomorphic to those obtained with the standard subset constructions, while using the upward language-based congruences results in minimal DBTAs and DTTAs.
We also show that a similar reasoning holds for downward congruences and (minimal) co-DBTAs and co-DTTAs.

While upward congruences allows us to build deterministic BTAs, we observe that downward congruences must satisfy an extra condition to guarantee that they yield BTAs that are \emph{co-deterministic}.
Our first contribution is to identify the class of congruences satisfying this condition, which we coin \emph{strongly downward congruences} (Definition~\ref{def:strong_downward_congruence}).
Roughly speaking, since not every tree language has a finite representation in the form of a co-DBTA (recall that only path-closed languages do), the so-called strongly downward congruences attempt to capture the notion of path-closedness in their definition.

Secondly, unlike the language-based upward~\cite{kozen1993MyhillTrees,tata2007} and downward congruences~\cite{nivat1997minimal}, which have been studied previously in the context of TA minimization; to the best of our knowledge, our automata-based congruences are novel.
For this purpose, \emph{we define an operator \(\post(\cdot)\) for TAs} (Definition~\ref{def:postpre}), in the same fashion to the existing one for word automata, and we use it to define the automata-based upward congruence (Definition~\ref{def:automataEquivalences}).
Analogously, \emph{we introduce an operator \(\pre(\cdot)\)} (Definition~\ref{def:postpre}) which allows us to define its downward counterpart.
While the definition of \(\post(\cdot)\) is a straightforward generalization of the word case, that is not the case of \(\pre(\cdot)\).
One more time, our \(\pre(\cdot)\) operator aims to capture the notion of path-closedness that is imposed by our goal of providing \emph{co-deterministic} BTA constructions.

Then, by chaining together the right constructions we obtain a simple \emph{proof of correctness of a Brzozowski's style algorithm} (Corollary~\ref{cor:RDRD}) producing the minimal DBTA for a given tree language.
This algorithm was first proposed by Björklund and Cleophas~\cite{Bjorklund2016Taxonomy}.
However, they did not include a proof, which, far from being conceptually new, could have resulted in several lines of technical details.
We believe that our proof does bring new insights in form of new non-trivial notions adapted from the classical word automata.

Finally, \emph{we generalize Brzozowski's algorithm} similarly to what Brzozowski and Tamm~\cite{brzozowski2014theory} have done for the case of automata over words.
More precisely, we give a sufficient and necessary condition on BTAs such that their determinization produces the minimal DBTA (Theorem~\ref{theorem:minimalifreverseatomic}).
This condition lifts the limitation to path-closed languages all the way up to regular tree languages.

In the main part of the document we only consider BTAs.
In Appendix~\ref{appendix:TTA} we adapt our results to TTAs enabled by a “reversal” operation on tree automata.

\section{Preliminaries}\label{sec:preliminaries}
We write \(\mathbb{N}\) for the natural numbers and \(\mathbb{N}_{+}\) for \(\mathbb{N}\setminus \{ 0 \}\).
A tree domain is a finite set of sequences over \(\mathbb{N}_+\) describing a tree structure.
Formally, a \emph{tree domain} \(D\) is a finite non-empty set \(D \subseteq (\mathbb{N}_{+})^*\) such that for each \(v \in (\mathbb{N}_{+})^*\), \(n \in \mathbb{N}_{+}\):
\begin{enumerate}
\renewcommand\labelenumi{\theenumi}
\renewcommand{\theenumi}{(\roman{enumi})}
\item if \(v\concat n \in D\) then \(v\in D\), i.e., \(D\) is prefix-closed, and
\item if \(v\concat n \in D\) and \(n>1\) then \(v\concat (n{-}1) \in D\).\label{def:TreeDomainChildrenSeq}
\end{enumerate}
The elements of \(D\) are called \emph{nodes} and every tree domain contains a node \(\varepsilon\) called the \emph{root}.
For example, \(D = \{\varepsilon, 1, 2, 3, 3\concat1, 3 \concat 2\}\) is a tree domain, while \(D' = \{\varepsilon, 1,2,3, 1\concat2, 1\concat3\}\) is not.
Note that \(D'\) does not satisfy condition~\ref{def:TreeDomainChildrenSeq} since \(1\concat1 \notin D'\).

Given an \emph{alphabet}, i.e., a finite non-empty set of \emph{symbols}, a \emph{tree} is a total function that maps nodes onto symbols.
Formally, given an alphabet \(A\) and a tree domain \(D \subseteq (\mathbb{N}_{+})^*\), an \(A\)-labeled \emph{tree} is a function \(t\colon D \to A\).
We use \(\dom(t)\) to denote the tree domain of a tree \(t\), and \(t(v)\) to denote the label of a node \(v \in \dom(t)\).

The rank of a node is the number of its children.
Formally, given a tree \(t\), the \emph{rank} of \(v \in\dom(t)\), denoted \(\rank{v}_t\), is \(\len{ \{ k\in\mathbb{N}_+ \mid  v\concat k \in \dom(t)\}}\).
Nodes of rank \(0\) are \emph{leaves}, and \(\leaf(t)\) is the set of leaves of \(t\).
We say that a tree is \emph{monadic} if{}f \(\forall v \in \dom(t) \colon \rank{v}_t \leq 1\).
Note that, in the context of word automata, monadic trees over an alphabet \(A\) can be interpreted as words over \(A\).

The \emph{height} of a tree \(t\) is defined as \(\height(t)\udiffg 1 {+} \max \{\len{v} \mid v \in \dom(t) \}\), where \(\len{v}\) denotes the \emph{length} of node \(v\) when interpreted as a sequence in \((\mathbb{N}_{+})^*\).

Given a tree \(t\) with root \(v \in \dom(t)\) the \emph{subtree} \(t'\) rooted at \(v\) is such that \(\dom(t') = \{u \in (\mathbb{N}_{+})^* \mid v\concat u \in \dom(t)\}\) and \(t'(u) = t(vu)\), for every \(u\in\dom(t')\).
Let \(t(v) = f\) for some \(v \in \dom(t)\) and let \(t_i\), with \(i = \range{\rank{v}_t}\), be the subtree of \(t\) at node \(v \concat i\).
Then we denote the subtree \(t'\) rooted at node \(v\) as \(f[t_1, \ldots, t_{\rank{v}_t}]\).
Given a symbol \(a\), we often write \(a\), instead of \(a[\,]\), to describe the tree \(t = a[\,]\).
For instance, we write \(f[a,b]\) instead of \(f[a[\,],b[\,]]\).

\begin{example}
\label{ex:running}
We describe our running example for the next definitions.
Let \(\tilde{t}\colon D \to A\) be the tree shown below, defined as \(\tilde{t} \udiffg  f[a,b,g[a,c]]\).
In this case, \(A = \{a,b,c,g,f\}\) and \(\dom(\tilde{t}) = D =  \{\varepsilon,1,2,3,3\concat 1,3\concat 2\}\) with \(\tilde{t}(\varepsilon) = f\), \(\tilde{t}(1) = \tilde{t}(3\concat 1) = a\), \(\tilde{t}(2) = b\), \(\tilde{t}(3) = g\) and \(\tilde{t}(3\concat2) = c\).
Thus, \(\rank{1}_{\tilde{t}} = \rank{2}_{\tilde{t}} = \rank{3\concat 1}_{\tilde{t}} = \rank{3\concat2}_{\tilde{t}} = 0\), and \(\leaf(\tilde{t}) = \{1, 2, 3\concat 1, 3\concat2\}\).
On the other hand, \(\rank{\varepsilon}_{\tilde{t}} =3\) and \(\rank{3}_{\tilde{t}} = 2\).
Finally, the height of \(\tilde{t}\) is \(\height(\tilde{t}) = 3\).
Figure~\ref{fig:tree} shows a depiction of the tree \(\tilde{t}\).
\end{example}

\begin{figure}[!ht]
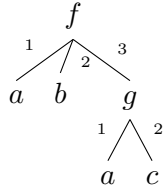

\noindent\centering\tikz[style={
    level distance=30pt,
    sibling distance=4pt,
    frontier/.style={distance from root=60pt}}]{
\Tree
[.{\(f\)}
  \edge node[auto=right,pos=.5] {\tiny{$1$}};
  [.{\(a\)} ]
  \edge node[auto=left,pos=.2] {\tiny{$2$}};
  [.{\(b\)} ]
  \edge node[auto=left,pos=.6] {\tiny{$3$}};
  [.{\(g\)}
  \edge node[auto=right,pos=.6] {\tiny{$1$}};
  [.{\(a\)} ]
  \edge node[auto=left,pos=.6] {\tiny{$2$}};
  [.{\(c\)} ]]
]
}\vspace*{-1mm}
\caption{Depiction of the tree \(\tilde{t}\) from Example~\ref{ex:running}.}\label{fig:tree}%
\end{figure}

Next we introduce \emph{ranked trees} building upon \emph{ranked alphabets}.
A \emph{ranked alphabet} \(A\) is an alphabet partitioned into pairwise disjoint subsets \(A = \bigcup_{k \in \mathbb{N}} A_k\) where \(A_k\) are the symbols of rank \(k\).
Given \(f \in A\), the unique index \(k\) such that \(f\in A_k\) is the \emph{rank} of \(f\) and we denote it by \(\rank{f}\).
The rank of a ranked alphabet \(A\) is the greatest index \(k\) such \(A_k \neq \varnothing\).

Given a ranked alphabet \(A\), let \(\cT_{A}\) denote the set of all \(A\)-labeled \emph{ranked trees} such that, in every tree \(t\in\cT_{A}\), the rank of every node \(v \in \dom(t)\) coincides with the rank of \(t(v)\), i.e., \(\forall v \in \dom(t) \colon  \rank{v}_t = \rank{t(v)}\).
For an unranked alphabet \(A\), let \(\cT_A\) denote the set of all \(A\)-labeled (\emph{unranked}) trees.
In Example~\ref{ex:running}, \(\tilde{t} \in \cT_A\) is a ranked tree with \(A = A_0 \cup A_1 \cup A_2 \cup A_3\) where \(A_0 = \{a,b,c\}\), \(A_1=\varnothing\), \(A_2 = \{g\}\) and \(A_3 = \{f\}\).

\medskip
A \emph{tree language} over an alphabet \(A\) is a subset \(L \subseteq \cT_{A}\).
Define the \emph{path language} of a tree \(t \in \cT_A\), denoted by \(\pi(t)\), as the subset of \(A(\mathbb{N}_{+}A)^*\) given by:
\[\pi(t) \udiffg %
  \begin{cases}
    \{f\} & \text{if  \(t = f[\,]\) }\\
    \bigcup_{i = 1}^{\rank{f}} \{f\concat i\concat w \mid w \in \pi(t_i)\} & \text{if \(t = f[t_1,\ldots,t_{\rank{f}}]\)}
  \end{cases} \enspace .
\]
The \emph{path language of a tree language} \(L\) is defined as \(\pi(L) \udiffg  \bigcup_{t \in L} \pi(t)\).
A tree language \(L \subseteq \cT_A\) is \emph{path-closed} if{}f \(\{t\in \cT_A \mid \pi(t) \subseteq \pi(L)\} = L\).
For example, the path language of tree \(\tilde{t}\) is \(\pi(\tilde{t}) = \{f1a, f2b, f3g1a, f3g2c\}\).
On the other hand, the path language of the tree language \(L_1 \udiffg  \{f[a,b], f[b,a]\}\) is \(\pi(L_1) = \{f1a, f2b, f1b, f2a\}\) and it is not path-closed since \(f[a,a] \notin L_1\) and \(\pi(f[a,a]) = \{f1a, f2a\} \subseteq \pi(L_1)\); while \(L_2 \udiffg  \{f[a,b], f[b,a], f[a,a], f[b,b]\}\) satisfies that \(\pi(L_2) = \pi(L_1)\) and it is path-closed.

\subsection{Contexts and substitution}

Let \(A\) be a ranked alphabet and let \(\cX \notin A\) be a special symbol with \(\rank{\cX} = 0\).
A \emph{context} over \(A\) is a tree \(t \in \cT_{A\cup \{\cX\}}\) such that \(t(v) = \cX\) for exactly one node \(v\in\dom(t)\) that we call the \emph{pivot} of \(t\) and that we denote \(\piv(t)\).
Note that if \(A\) is a ranked alphabet, so is \(A \cup \{\cX\}\).
We define the \(\cX\)-height of a context \(x \in \cC_A\), denoted \(\height_{\cX}(x)\) as follows \(\height_{\cX}(x) \udiffg  1 + \len{\piv(x)}\).

The set of all contexts over an alphabet \(A\) is denoted by \(\cC_A\).
Since contexts are trees defined over a special alphabet, all notions defined for trees, including the ones we introduce next, apply to contexts.
For clarity, we typically use \(t,r\) to denote trees and \(x,y\) for contexts when the distinction is important.

We define a substitution operator for trees as follows.
Let \(t,t' \in \cT_A\) and let \(v \in \dom(t)\).
The \emph{tree substitution} \(\subs{t}{t'}_v\) is the result of replacing the subtree rooted at node \(v\) in \(t\) with \(t'\).
Formally, \(\subs{t}{t'}_v(u)=t(u)\) for all \(u\in\dom(t)\setminus v\concat (\mathbb{N}_+)^*\) and \(\subs{t}{t'}_v(v\concat u)=t'(u)\), for all \(u \in \dom(t')\).
We omit the subindex \(v\) from \(\subs{t}{t'}_v\) when \(t\) is a context and \(v=\piv(t)\).
For instance, recall the tree \(\tilde{t}=f[a, b, g[a,c]]\) from Example~\ref{ex:running} and let \(\tilde{r} = h[a, b, c]\).
We have \(\subs{\tilde{t}}{\cX}_2 = f[a, \cX, g[a,c]]\) and \(\subs{\tilde{t}}{\tilde{r}\,}_{31} = f[a, b, g[h[a, b, c],c]]\).
Note that \(\cX[\,]\) satisfies \(\subs{\cX}{t} = t\), for every \(t \in \cT_{A}\), and thus it is called the \emph{identity} context.
For monadic trees, tree substitution coincides with word concatenation, where contexts correspond to prefixes of words.
In particular, \(\cX[\;]\) corresponds to \(\varepsilon\), the empty word.

\begin{definition}[Upward and Downward Quotients]
Given a tree language \(L \subseteq \cT_A\), a tree \(t\in\cT_A\), we define the \emph{upward quotient} of \(L\) w.r.t. \(t\) as \(L\,t^{-1} \udiffg  \{x \in \cC_{A} \mid \subs{x}{t} \in L\}\).
Similarly, we define the \emph{downward quotient} of \(L\) w.r.t. a context \(x \in \cC_A\) as \(x^{-1}L \udiffg  \{t \in \cT_{A} \mid \subs{x}{t} \in L\}\).
\end{definition}
Observe that \(t \in x^{-1} L \Lra x \in Lt^{-1}\).
It is worth remarking that, in the monadic case, \(Lt^{-1}\) and \(x^{-1}L\)  correspond to the so-called \emph{right quotient} of \(L\) by \(t\) and \emph{left quotient} of \(L\) by \(x\), respectively, where \(t\) and \(x\) are words over \(A\).

\subsection{Bottom-up tree automata}\label{sec:bta}
\begin{definition}[Bottom-up tree automaton]
A \emph{bottom-up tree automaton} (BTA) is a tuple \(\TA = \tuple{Q, \Sigma, \delta, F}\) where \(Q\) is a finite set of \emph{states}; \(\Sigma\) is a ranked alphabet of rank \(n\); \(\delta\colon  \bigcup_{i=0}^n \Sigma_i\times Q^i \rightarrow \wp(Q)\) is the partial \emph{transition function}, where \(Q^i\) denotes the set of \(i\)-tuples of elements in \(Q\) and \(\wp(Q)\) denotes the powerset of \(Q\); and \(F \subseteq Q\) is the set of \emph{final states}.
\end{definition}
We denote a tuple \((f,q_1,\ldots,q_{\rank{f}}) \in \Sigma_{\rank{f}}\times Q^{\rank{f}}\) as \(f[q_1,\ldots,q_{\rank{f}}]\).
Note that if \(\rank{f} = 0\), we denote the singleton tuple \((f) \in \Sigma_0 \times Q^0\) as \(f[\,]\).
We extend \(\delta\) to sets \(S\) of tuples as \(\delta(S) \udiffg  \bigcup_{t \in S} \delta(t)\) and define the set of \emph{initial states} of a BTA \(\TA\) as \(\initials(\TA) \udiffg  \{q \in Q \mid \exists a \in \Sigma_0 \colon q \in \delta(a[\,])\}\).

A BTA is \emph{deterministic} (DBTA for short) if{}f
every set of states in the image of \(\delta\) is a singleton or is empty.
Similarly, a BTA is \emph{co-deterministic} (co-DBTA for short) if{}f \(F\) is a singleton and for each \(q \in Q\) and \(f \in \Sigma_n\), with \(n \geq 1\), we have: %
if \(q\in\delta(f[q_1,\ldots,q_{\rank{f}}])\) and \(q\in\delta(f[q'_1,\ldots,q'_{\rank{f}}])\) then \(q_i=q'_i\) for each \(i=\range{\rank{f}}\).

We define the \emph{move} relation on a BTA, denoted by \(\xrightarrow{}_{\TA} \in  \cT_{\Sigma \cup Q} \times \cT_{\Sigma \cup Q}\) as follows.
Let \(t = \subs{x}{f[t_1,\ldots,t_{\rank{f}}]}\) and \(t' = \subs{x}{q[t_1,\ldots,t_{\rank{f}}]}\) for some \(x \in \cC_{\Sigma\cup Q}, f \in \Sigma\) and \(t_1,\ldots,t_{\rank{f}}\!\in\!\cT_{Q}\).
Then, \(t \ra_{\TA} t' \udiff q \in \delta(f[t_1(\varepsilon),\ldots,t_{\rank{f}}(\varepsilon)])\).
We use \(\xrightarrow{}^*_{\TA}\) to denote the reflexive and transitive closure of \(\xrightarrow{}_{\TA}\).
The \emph{language} defined by \(\TA\) is \(\lang{\TA} \udiffg  \{t \in \cT_{\Sigma} \mid \exists t' \in \cT_{Q} \colon  t'(\varepsilon) \in F \land t \to^*_{\TA} t'\}\).
Intuitively, a run of a BTA on a tree \(t\in\cT_\Sigma\) relabels the nodes of the \(t\) starting with its leaves and makes its way up towards the root of \(t\) as prescribed by the move relation.
The run accepts when the root of \(t\) is labelled with an accepting state of the BTA.

\begin{remark}
\label{remark:path-closed-languages}
BTAs and DBTAs define the class of \emph{regular tree languages}, while co-DBTAs (which are equivalent to deterministic top-down automata, in the sense that they accept the same class of tree languages, as we shall see in Appendix~\ref{sec:appendix-reverseTAs}) define the subclass of path-closed tree languages~\cite{Viragh81}.
It is decidable whether the language of a BTA is path-closed~\cite{tata2007}.
\end{remark}
\begin{example}\label{example:BTA}
Let \(\TA = \tuple{Q,\Sigma_0\cup \Sigma_2,\delta,F}\) be a BTA with \(Q = \{q_0,q_1\}\), \(\Sigma_0 =\{\true,\false\}\), \(\Sigma_2= \{\AND,\OR\}\), \(\{q_0\}= \delta(\AND[q_0,q_1]) = \delta(\AND[q_0,q_0]) = \delta(\AND[q_1,q_0]) = \delta(\OR[q_0,q_0]) = \delta(\false[\,])\), \(\{q_1\} = \delta(\AND[q_1,q_1]) = \delta(\OR[q_1,q_0]) = \delta(\OR[q_0,q_1]) = \delta(\OR[q_1,q_1]) = \delta(\true[\,])\) and \(F = \{q_1\}\).

\medskip
Note that \(\lang{\TA}\) is defined as the set of all trees of the form \(t \in \cT_{\Sigma_0 \cup \Sigma_2}\) which yield to propositional formulas, over the binary connectives \(\land\) and \(\lor\) and the constants \(\true\) (true) and \(\false\) (false), that evaluate to \(\true\).
For instance, the following is a sequence of moves from a tree \(t \in \cT_{\Sigma_0 \cup \Sigma_2}\) such that \(t\in \lang{\TA}\).

\centering
\tree{15}{30}{
\Tree
[.{\AND}
  [.{\OR} {\true} {\false} ]
  [.{\true} ]
]
}%
\tmove{15}\hspace{-3pt} %
\tikz{\node at (0,0) {};
\node at (0,15pt) {\(\ldots\)};}%
\tmove{15}\hspace{-7pt}%
\tree{15}{30}{
\Tree
[.{\AND}
  [.{\OR} {\(q_1\)} {\(q_0\)} ]
  [.{\(q_1\)} ]
]
}%
\tmove{15}\hspace{-7pt}\vspace{7pt}%
\tree{15}{30}{
\Tree
[.{\AND}
  [.{\(q_1\)} {\(q_1\)} {\(q_0\)} ]
  [.{\(q_1\)} ]
]
} %
\hspace{-5pt}\tmove{15}\hspace{-7pt}%
\tree{15}{30}{
\Tree
[.{\(q_1\)}
  [.{\(q_1\)} {\(q_1\)} {\(q_0\)} ]
  [.{\(q_1\)} ]
]
}\tdot{15}

Note that  \(\subs{\AND[\cX,q_1]}{\OR[q_1,q_0]} \rightarrow_{\TA} \subs{\AND[\cX,q_1]}{q_1[q_1,q_0]}\) as \(q_1 \in \delta(\OR[q_1,q_0])\).
Observe that \(\TA\) is deterministic but not co-deterministic, since \(\{q_0\}= \delta(\AND[q_0,q_1]) = \delta(\AND[q_1,q_0])\), for instance.
Finally, the language \(\lang{\cA}\) is not path-closed since \(\lor[\false,\false] \notin \lang{\TA}\) but \(\pi(\lor[\false,\false]) = \{\lor1\false, \lor2\false\} \subseteq \pi(\lang{\TA})\).\eox
\end{example}

For each \(q \in Q\) and \(S \subseteq Q\), define the \emph{upward} and \emph{downward language} of \(q\) w.r.t. \(S\), denoted by \(\mathcal{L}^{\TA}_{\ua}(q,S)\) and \(\mathcal{L}^{\TA}_{\da}(q,S)\), respectively, as follows:
\begin{equation*}%
\begin{aligned}
  \mathcal{L}^{\TA}_{\ua}(q,S) &\udiffg  \{ c \in \cC_{\Sigma} \mid \exists t' \in \cT_Q \colon  \subs{c}{q} \to^*_{\TA} t', t'(\varepsilon) \in S \}\\
  \mathcal{L}^{\TA}_{\da}(q,S) &\udiffg  \{ t \in \cT_{\Sigma} \mid \exists t' \in \cT_Q \colon  t \to^*_{\TA} t', t'(\varepsilon) = q, \ell(t') \subseteq S \}
  \enspace.
\end{aligned}
\end{equation*}
We will simplify the notation and write \(\mathcal{L}^{\TA}_{\ua}(q)\) when \(S = F\), and \(\mathcal{L}^{\TA}_{\da}(q)\) when \(S = \initials(\TA)\).
Also, we will drop the superscript \(\TA\) when the BTA \(\TA\) is clear from the context.
Note that, in the monadic case,  \(\mathcal{L}^{\TA}_{\ua}(q)\) corresponds to the so-called \emph{right language} of state \(q\), i.e., the set of words that can be read from \(q\) to a final state of the corresponding word automaton; while \(\mathcal{L}^{\TA}_{\da}(q)\) corresponds to the \emph{left language} of \(q\), i.e., the set of words that can be read from an initial state of the corresponding word automaton to state \(q\).
When generalizing to trees we have that \(\mathcal{L}_{\ua}(q)\) is the set of contexts such that the result of being processed by \(\TA\) (starting from state \(q\) instead of \(\cX\)) have their root labelled with a final/accepting state.
Similarly, \(\mathcal{L}_{\da}(q,S)\) is the set of trees such that their processing by \(\TA\) return trees with \(q\) as root and leaves labelled by \(S\).
Finally, it is easy to see that \(\lang{\TA} = \bigcup_{q \in F} \mathcal{L}_{\da}(q) \) for every BTA \(\TA\).

Given a ranked alphabet \(\Sigma\), \(f \in \Sigma\) and \(T_1,\ldots,T_{\rank{f}} \subseteq \cT_\Sigma\), let \(f[T_1,\ldots,T_{\rank{f}}] \udiffg  \{f[t_1,\ldots,t_{\rank{f}}] \mid t_i \in T_i, i=\range{\rank{f}}\}\).
This notation allows us to give an inductive characterization of downward languages of a BTA (which will be useful in the proofs of Lemmas~\ref{HBPreservesL} and~\ref{HTPreservesL}).

\begin{lemma}
\label{DownwardLanguages}
Let \(\cA = \tuple{Q,\Sigma,\delta,F}\) be a BTA.
\mbox{For every \(q \in Q\):}
\[
\mathcal{L}_{\da}(q) = \{f[\mathcal{L}_{\da}(q_1),\ldots,\mathcal{L}_{\da}(q_{\rank{f}})] \mid \exists f \!\in\! \Sigma \colon  q \!\in\! \delta(f[q_1,\ldots,q_{\rank{f}}])\} \enspace .
\]
\end{lemma}
\begin{proof}
Let \(\cA = \tuple{Q,\Sigma,\delta,F}\) be a BTA and let \(f \in \Sigma\) and \(q,q_1,\ldots,q_{\rank{f}} \in Q\) be such that \(q \in \delta(f[q_1,\ldots,q_{\rank{f}}])\).
Then,
\begin{align*}
  \forall i \in \{\range{\rank{f}}\} \colon  t_i \in \mathcal{L}_{\da}(q_i) & \Lra \\
  \forall i \in \{\range{\rank{f}}\}  \colon  \exists t_i' \in \cT_Q \colon  t_i'(\varepsilon) = q_i,\,t_i \to^*_{\cA} t'_i & \Lra\\
\exists t' \in \cT_Q \colon  t'(\varepsilon) = q,\,f[t_1,\ldots,t_{\rank{f}}] \to^*_{\cA} t' & \Lra \\
f[t_1,\ldots,t_{\rank{f}}] \in \mathcal{L}_{\da}(q) \enspace .
\end{align*}
Note that the first and last double-implication hold by definition of downward language of a state w.r.t. \(\initials(\TA)\).
The second implication holds since, by H.I., \(q \in \delta(f[q_1,\ldots,q_{\rank{f}}])\).
\end{proof}

\begin{example}
In Example~\ref{example:BTA}, \(\mathcal{L}_{\ua}(q_0) \neq \mathcal{L}_{\ua}(q_1)\).
In fact, \(\cX[\;] \in \mathcal{L}_{\ua}(q_1)\), since \(q_1 \in F\) while \(\cX[\;] \notin \mathcal{L}_{\ua}(q_0)\), since \(q_0 \notin F\).

On the other hand, \(\false[\,] \in \mathcal{L}_{\da}(q_0)\) and \(\true[\,] \in \mathcal{L}_{\da}(q_1)\).
Finally, consider the set of trees \(\land[\mathcal{L}_{\da}(q_1)$, $\mathcal{L}_{\da}(q_1)]\).
The reader may check that \(\land[\mathcal{L}_{\da}(q_1), \mathcal{L}_{\da}(q_1)] \subseteq \mathcal{L}_{\da}(q_1) \) and  \(\land[\mathcal{L}_{\da}(q_1), \mathcal{L}_{\da}(q_1)] \nsubseteq \mathcal{L}_{\da}(q_0)\).
In fact, \(\mathcal{L}_{\da}(q_0) \cap \mathcal{L}_{\da}(q_1) = \varnothing\) as \(\TA\) is deterministic.
\eox
\end{example}

A state \(q \in Q\) of a BTA is \emph{unreachable} (resp.\ \emph{empty}) if{}f its downward (resp.\ upward) language is empty.
The \emph{minimal} DBTA for a regular tree language is the DBTA with the least number of states which is unique modulo isomorphism.
Let \(\TA \equiv \TA'\) denote that two BTAs \(\TA\) and \(\TA'\) are isomorphic.

Given a BTA \(\TA = \tuple{Q, \Sigma, \delta, F}\), the \emph{bottom-up determinization} builds a DBTA $\cD \udiffg  \langle \wp(Q), \Sigma$, $\delta', F'\rangle$ where \(\delta'(f[R_1,\ldots,R_{\rank{f}}]) \udiffg  \{q \mid \exists q_1 \!\in\! R_1,\ldots,q_{\rank{f}} \!\in\! R_{\rank{f}} \colon  q \!\in\! \delta(f[q_1,\ldots,q_{\rank{f}}])\}\) for each \(f \in \Sigma\), \(R_i \in \wp(Q)\) and \(i=\range{\rank{f}}\),
and \(F' \udiffg  \{R \in \wp(Q) \mid R \cap F \neq \varnothing\}\), such that \(\lang{\cD} = \lang{\TA}\)~\cite{cleophas2008tree}.
We denote \({\cA^{\mathsf{D}}}\) the result of applying the bottom-up determinization to \(\TA\) and removing unreachable states.

Given a BTA \(\TA = \tuple{Q, \Sigma, \delta, F}\), the \emph{bottom-up} \emph{co-determinization} builds a co-DBTA \(\mathcal{E} \udiffg  \langle \wp(Q)$, $\Sigma, \delta', F' \rangle$, where \(F' \udiffg  \{F\}\) and \(\delta'\) is defined as follows.
Given \(f \in \Sigma \setminus \Sigma_0\) and \(R \in \wp(Q)\), \(R \in \delta'(f[R_1,\ldots,R_{\rank{f}}])\), where
\(R_i \udiffg  \{q_i \in Q \mid \exists q \in R, q_1,\ldots,q_{i{-}1},q_{i{+}1},\ldots,q_{\rank{f}} \in Q \colon  q \in \delta(f[q_1,\ldots,q_i, \ldots, q_{\rank{f}}])\}\).
Moreover, for every \(a \in \Sigma_0\) and \(R \in \wp(Q)\) such that \(\exists q \in R \colon q \in \delta(a[\,])\), we have that \(R \in \delta'(a[\,])\).
Note that if \(\lang{\TA}\) is a path-closed language and \(\cA\) has no unreachable states then \(\lang{\mathcal{E}} = \lang{\TA}\)~\cite{cleophas2008tree}.
We denote \(\cA^{\mathsf{cD}}\) the result of applying the bottom-up co-determinization to \(\TA\) and removing empty states.
\begin{example}
\label{example:co-BTA}
Let us give an example of the bottom-up co-determinization by removing the symbol \(\AND\) from Example~\ref{example:BTA}.
Let \(\TA = \tuple{Q,\Sigma_0\cup \Sigma_2,\delta,F}\) be a BTA with \(Q = \{q_0,q_1\}\), \(\Sigma_0 =\{\true,\false\}\), \(\Sigma_2= \{\AND\}\), \(\{q_0\}= \delta(\AND[q_0,q_1]) = \delta(\AND[q_0,q_0]) = \delta(\AND[q_1,q_0]) = \delta(\false[\,])\), \(\{q_1\} = \delta(\AND[q_1,q_1]) =  \delta(\true[\,])\) and \(F = \{q_1\}\).

Note that \(\lang{\TA}\) is defined as the set of all trees of the form \(t \in \cT_{\Sigma_{0} \cup \Sigma_2}\) which yield to propositional formulas over the binary connective \(\land\) and the constants \(\true\) and \(\false\) that evaluate to \(\true\).
It is routine to check that \(\lang{\TA}\) is path-closed, hence there exists a co-DBTA defining \(\lang{\TA}\).

Since \(\TA\) has no unreachable states, we use the bottom-up co-determinization to build a co-DBTA \(\mathcal{E}\) such that \(\lang{\mathcal{E}} = \lang{\TA}\).
We obtain \(\mathcal{E} = \tuple{\wp(Q), \Sigma_0 \cup \Sigma_2 , \delta', F'}\) where \(\delta'\) is defined as \(\{q_1\} \in \delta'(\land[\{q_1\}, \{q_1\}])\) and  \(\{q_0\}, \{q_0, q_1\} \in \delta'(\land[\{q_0, q_1\}, \{q_0, q_1\}])\).
On the other hand, \(\{q_1\}, \{q_0, q_1\} \in \delta'(\true[\,])\) and \(\{q_0\},\) \(\{q_0, q_1\} \in \delta'(\false[\,])\).
Finally, let \(F' =  \{\{q_1\}\}\).

Note that states \(\{q_0\}\) and \(\{q_0, q_1\}\) are empty.
Thus, we would remove them (and the corresponding entries of the transition function \(\delta'\)) to build \(\cA^{\mathsf{cD}}\).\eox
\end{example}

\begin{remark}
Note that if \(\lang{\TA}\) is path-closed but \(\TA\) has unreachable states then we cannot guarantee that \(\mathcal{E}\), the BTA that results from applying the bottom-up co-determinization operation, is such that \(\lang{\TA}  =  \lang{\mathcal{E}}\).

For instance, consider the BTA \(\TA  = \tuple{Q, \Sigma_0 \cup \Sigma_2, \delta, F}\) from Example~\ref{example:co-BTA} and define
 $\TA' \langle Q' $, $\Sigma_0 \cup \Sigma'_2, \delta', F' \rangle$  by setting \(Q' \udiffg  Q \cup \{q_2\}\), \(\Sigma'_2 \udiffg  \Sigma_2 \cup \{\star\}\) and \(F' \udiffg  F \cup\{ q_2\}\).
Finally, define \(\delta'\) as the union of \(\delta\) and the following entries: \(\{q_2\} = \delta'(\star[q_1, q_2]) = \delta'(\star[q_2, q_1]) = \delta'(\star[q_2, q_2])\).
Note that \(q_2\) is an unreachable state and thus, \(\lang{\TA'} = \lang{\TA}\).
However, \(\lang{\TA'} \subset \lang{\mathcal{E}'}\).
In fact, let \(\mathcal{E}' = \tuple{\wp(Q'), \Sigma_0 \cup \Sigma'_2, \delta_{\mathcal{E}'}, F_{\mathcal{E}'}}\).
Then, the reader may check that, since \(\{q_1, q_2\} \in \delta_{\mathcal{E}'}(\true[\,])\), \(\{q_1, q_2\} \in \delta_{\mathcal{E}'}(\star[\{q_1,q_2\}, \{q_1, q_2\}])\) and \(F_{\mathcal{E}'} = \{q_1, q_2\}\), the tree \(\star[\true, \true] \in \lang{\mathcal{E}'}\), while \(\star[\true, \true] \notin \lang{\TA}\).
\end{remark}

\subsection{Equivalences and congruences}

Given a set \(X\), an \emph{equivalence relation} \(\mathord{\sim} \subseteq X\times X\) induces a \emph{partition} \(P_{\sim}\) of \(X\) given by a family \(P_{\sim} = \{B_i\}_{i\in \cI}\), with \(\cI \subseteq \mathbb{N}\), of pairwise disjoint subsets of \(X\) (called \emph{blocks} or \emph{equivalence classes}) the union of which is \(X\).
The partition \(P_{\sim}\) is said to be \emph{finite} when \(\sim\) has \emph{finite index}, i.e., \(\sim\) consists of finitely many equivalence classes.

Given \(t\in X\), let \( P_{\sim}(t) \udiffg \{r \in X \mid t \sim r \}\) be the unique block containing \(t\).
This definition can be extended in a natural way to a set \(S \subseteq X\) as \(P_{\sim}(S) = \bigcup_{t\in S}{P_{\sim}(t)}\).
Unless stated otherwise, we consider equivalence relations of finite index.

Given two equivalence relations \(\sim_1, \sim_2\), we say that \(\sim_1\) is \emph{finer than or equal to} \(\sim_2\) when \(\mathord{\sim_1} \subseteq \mathord{\sim_2}\).  Sometimes we also say that \(\sim_2\) is \emph{coarser than or equal to} \(\sim_1\).
Observe that \(\mathord{\sim_1} \subseteq \mathord{\sim_2}\) is equivalent to write \(\forall t,r \in X \colon t \sim_1 r \Ra t \sim_2 r\).

Analogously to the notion of left and right congruences on words~\cite{ganty2019congruence}, we introduce \emph{upward} and \emph{downward} congruences on trees and contexts, respectively.
Intuitively, upward congruences are equivalences on trees that behave well when substituting trees into contexts, while downward congruences are equivalences on contexts that behave well when substituting pivots for trees.

\begin{definition}[Upward and Downward congruences]
Given a ranked alphabet \(\Sigma\), an equivalence \(\ru\) on \(\cT_\Sigma\) is an \emph{upward congruence} if{}f for every \(t,r \in \cT_{\Sigma}\) and context \(c \in \cC_{\Sigma}\): \(t \ru r \Ra \subs{c}{t} \ru \subs{c}{r}\).\\
Similarly, an equivalence \(\rd\) on \(\cC_\Sigma\) is a \emph{downward congruence} if{}f for every \(x,y,c \in \cC_{\Sigma} \colon  x \rd y \Ra \subs{x}{c} \rd \subs{y}{c}\).
\end{definition}

In the sequel, we follow the convention that upward congruences are denoted with a \(\mathsf{u}\) superscript while downward congruences are denoted with a \(\mathsf{d}\) superscript.
Note that similar definitions of upward~\cite{kozen1993MyhillTrees,tata2007} and downward~\cite{nivat1997minimal} congruences have been previously proposed in the literature.
In the monadic case, an upward congruence means that two congruent words remain so after prefixing the same word to their left, formally,  \(u \sim v\) implies \( wu \sim wv\) for all \(w\).
Similarly, a downward congruence in the monadic case means that two congruent words remain so after appending the same word to their right, formally, \(u \sim v\) implies \( uw \sim vw\) for all \(w\).

\begin{example}
\label{ex:congruences}
Given a tree language \(L \subseteq \cT_{\Sigma}\), the congruence defined as \(t \sim r \udiff Lt^{-1} = Lr^{-1}\), for every \(t,r \in \cT_{\Sigma}\), is an upward congruence.
Let us prove this fact by contradiction.

Suppose that \(t \sim r\), for every \(t,r \in \cT_{\Sigma}\), but \(\subs{c}{t} \not\sim \subs{c}{r}\), for some \(c \in \cC_{\Sigma}\).
In other words, \(Lt^{-1} = \{x \in \cC_{\Sigma} \mid \subs{x}{t} \in L\} = \{x \in \cC_{\Sigma} \mid \subs{x}{r} \in L\} = Lr^{-1}\), but there exists \(y \in \cC_{\Sigma}\) such that \(\subs{y}{\subs{c}{t}} \in L\) and  \(\subs{y}{\subs{c}{r}} \notin L\).
It follows that, by defining \(x = \subs{y}{c}\), we have: \(\subs{x}{t} \in L\) and \(\subs{x}{r} \notin L\), which contradicts the fact \(Lt^{-1} = Lr^{-1}\).

This upward congruence is also known as the \emph{Myhill-Nerode relation for tree languages}~\cite{kozen1993MyhillTrees}.

Similarly, the congruence defined as \(x \sim y \udiff x^{-1}L = y^{-1}L\), for every \(x,y \in \cC_{\Sigma}\), is a downward congruence.
Lemma~\ref{NerodeStronglyDownward} proves this result.\eox
\end{example}

\section{Tree automata constructions from congruences}\label{sec:AutomataConstruction}

In this section we present two tree automata constructions that are built upon a given tree language and a congruence.
Generally speaking, upward congruences yield deterministic constructions while downward congruences yield co-deterministic ones.

First we introduce the deterministic construction.
Given a regular tree language \(L\) and an upward congruence \(\ru\) that preserves the equivalence between upward quotients, i.e., \(t \ru r \Ra Lt^{-1} = Lr^{-1}\) for every \(t,r \in \cT_{\Sigma}\), the following construction yields a DBTA defining exactly \(L\).

\begin{definition}[Automata construction \(\cH^{\mathsf{u}}(\ru, L)\)]\label{def:up_TA_construction}
Let \(\ru\) be an upward congruence and let \(L \subseteq \cT_{\Sigma}\).
Define the BTA \(\cH^{\mathsf{u}}(\ru,L) \udiffg \tuple{Q, \Sigma, \delta, F}\) where \(Q = \{ P_{\ru}(t) \mid t \in \cT_{\Sigma}\}\), \(F = \{P_{\ru}(t) \mid t \in L\}\) and the transition function \(\delta\) is defined as follows:
\begin{enumerate}
\renewcommand\labelenumi{\theenumi}
\renewcommand{\theenumi}{(\roman{enumi})}
\item for every \(a \in \Sigma_0\), \(\delta(a[\,])=\{ P_{\ru}(a[\,]) \}\), and
\item for every \(f\! \in \Sigma_n\), with \(n\!\geq\!1\), and trees \(t_1,... ,t_{\rank{f}},t \!\in \cT_\Sigma\): \(P_{\ru}(t)\! \in \delta(f[P_{\ru}(t_1),... , P_{\ru}(t_{\rank{f}})]) \) if{}f \(f[P_{\ru}(t_1),... , P_{\ru}(t_{\rank{f}})]\) \(\subseteq P_{\ru}(t)\).
\end{enumerate}
\end{definition}

\begin{remark}\label{remark:DBTA}
Note that \(\cH^{\mathsf{u}}(\ru, L)\) is \emph{well-defined} since we assume that \(\ru\) has finite index and is an upward congruence.
Moreover, \(\cH^{\mathsf{u}}(\ru, L)\) is a \emph{deterministic} BTA since for every \(f \in \Sigma\) and \(t_1, \ldots, t_{\rank{f}}, t \in \cT_{\Sigma}\) there exists \emph{exactly one} block \(P_{\ru}(t) \in Q\) such that \(f[P_{\ru}(t_1), \ldots, P_{\ru}(t_{\rank{f}})]\) \(\subseteq P_{\ru}(t)\).
Finally, observe that \(\cH^{\mathsf{u}}(\ru, L)\) might contain empty states but every state is reachable.
\end{remark}

\begin{lemma}
\label{HBPreservesL}
Let \(L \subseteq \cT_{\Sigma}\) be a tree language and let \(\ru\) be an upward congruence such that \(t\ru r \Ra L\,t^{-1} = L\,r^{-1}\) for every \(t,r \in \cT_\Sigma\).
Then \(\lang{\cH^{\mathsf{u}}(\ru,L)} = L\).
\end{lemma}
\begin{proof}
To simplify the notation, we denote \(P_{\ru}\), the partition induced by \(\ru,\) simply by \(P\).
Let \(\cH = \cH^{\mathsf{u}}({\ru,}~L)\).
First, we prove that:
\begin{equation}\label{eq:ru_preserve_downward_language}
\mathcal{L}^{\cH}_{\da}(P(r)) = P(r), \text{ for each } r \in \cT_{\Sigma} \enspace .
\end{equation}

(\(\subseteq\)). We show that, for all \(t \in \cT_{\Sigma}\), \(t \in \mathcal{L}^{\cH}_{\da}(P(r)) \Ra t \in P(r)\).
We proceed by induction in the height of the tree \(t\).
\begin{itemize}
\item \textit{Base case:} Let \(t\) be of height \(0\), i.e. \(\exists a \in \Sigma_0\) such that \(t = a[\,]\).
Then,
\begin{align*}
a[\,] \in \mathcal{L}^{\cH}_{\da}(P(r)) & \Lra \quad \text{[Def. of \(\mathcal{L}^{\cH}_{\da}(P(r))\)]} \\
\exists t \in \cT_Q \colon  a[\,] \xrightarrow{}^*_{\cH} t \land t(\varepsilon) = P(r) & \Lra \\
P(r) \in \delta(a[\,]) & \Lra \quad \text{[Def.~\ref{def:up_TA_construction}]} \\
a[\,] \in P(r) \enspace .
\end{align*}

\item \textit{Inductive step:}
Now we assume by hypothesis of induction that \(t \in \mathcal{L}^{\cH}_{\da}(P(r)) \Ra t \in P(r)\) for all trees of height up to \(n\).
Let \(t\) be a tree of height \(n{+}1\), i.e. \(t = f[t_1,t_2]\) for some \(f \in \Sigma\) and \(t_1,t_2 \in \cT_\Sigma\) height up \(n\).
Note that, w.l.o.g., we assume \(\rank{f} = 2\) for the sake of clarity.
Then,
\begin{align*}
t \in \mathcal{L}^{\cH}_{\da}(P(r)) & \Lra \\
\exists t' \in \cT_Q \colon  t \xrightarrow{}^*_{\cH} t' \land t'(\varepsilon) = P(r) & \Lra \\
\exists r_1, r_2 \in \cT_\Sigma \colon  t_1 \in \mathcal{L}^{\cH}_{\da}(P(r_1)), t_2 \in \mathcal{L}^{\cH}_{\da}(P(r_2))& \text{ and }\\
P(r) \in \delta(f[P(r_1),P(r_2)]) & \enspace ,
\end{align*}
where the second double-implication holds by Lemma~\ref{DownwardLanguages}, since \(t = f[t_1,t_2]\).

By Definition~\ref{def:up_TA_construction}, the fact that \(P(r) \in \delta(f[P(r_1),P(r_2)])\) is equivalent to \(f[P(r_1),P(r_2)] \subseteq P(r)\).
Since \(t = f[t_1, t_2]\) and, by I.H., \(t_i \in P(r_i)\) with \(i \in \{1,2\}\), we conclude that \(t \in f[P(r_1),P(t_2)]\).
\end{itemize}

(\(\supseteq\)). We show that, for all \(t \in \cT_{\Sigma}\), \(t \in P(r) \Ra t \in \mathcal{L}^{\cH}_{\da}(P(r))\).
We proceed by induction in the height of the tree \(t\).
\begin{itemize}
\item \textit{Base case:} Let \(t\) be of height \(0\), i.e. \(\exists a \in \Sigma_0\) such that \(t = a[\,]\).
Then,
\begin{align*}
a[\,] \in P(r) & \Ra \quad \text{[Def.~\ref{def:up_TA_construction}]} \\
P(r) \in \delta(a[\,]) & \Ra \quad \text{[Def. of \(\mathcal{L}^{\cH}_{\da}(P(r))\)]} \\
a[\,] \in \mathcal{L}^{\cH}_{\da}(P(r)) \enspace .
\end{align*}

\item \textit{Inductive step:}
Now we assume by hypothesis of induction that \(t \in P(r) \Ra t \in \mathcal{L}^{\cH}_{\da}(P(r))\) holds for all trees \(t\) of height up to \(n\).
Let \(t\) be a tree of height \(n{+}1\), i.e. \(t = f[t_1,t_{2}]\) for some \(f \in \Sigma\) and \(t_1,t_2 \in \cT_\Sigma\) have height up to \(n\).
Note that, w.l.o.g., we assume \(\rank{f} = 2\).
Then,
\begin{align*}
t \in P(r) & \Ra \quad \text{[\(t = f[t_1,t_2]\)]} \\
f[t_1,t_2] \in P(r) & \Ra \quad \text{[\(\ru\) is an upward cong.]} \\
f[P(t_1), P(t_2)] \subseteq P(r) & \Ra \quad \text{[Def.~\ref{def:up_TA_construction}]} \\
P(r) \in \delta(f[P(t_1),P(t_2)]) & \Ra  \\
t \in \mathcal{L}^{\cH}_{\da}(P(r)) \enspace ,
\end{align*}
where the last implication holds by Lemma~\ref{DownwardLanguages} and the fact that, by I.H, \(t_i \in \mathcal{L}^{\cH}_{\da}(P(t_i))\), with \(i \in \{1,2\}\).
\end{itemize}
We conclude this proof by showing that \(\lang{\cH} = L\).
To do that, we first prove that \(P(L) = L\).

On the one hand, \(L \subseteq P(L)\) by reflexivity of \(\ru\).
Now we show that \(P(L) \subseteq L\).
Let \(t \in \cT_{\Sigma}\) such that \(t \in P(L)\).
Then, there exists \(r \in L\) with \(t \ru r\).
On the other hand, \(t \ru r\) implies \(Lt^{-1} = Lr^{-1}\).
Then, we conclude that \(r \in L \Ra \cX \in Lr^{-1} \Ra \cX\in Lt^{-1} \Ra t \in L \).

Finally,
\begin{align*}
\lang{\cH} &= \quad \text{[Def. of \(\lang{\cH}\)]} \\
\bigcup\limits_{q \in F} \mathcal{L}^{\cH}_{\da}(q) & = \quad \text{[Def.~\ref{def:up_TA_construction}]}\\
\bigcup\limits_{t \in L} \mathcal{L}^{\cH}_{\da}(P(t)) & = \quad \text{[Equation~\eqref{eq:ru_preserve_downward_language}]} \\
\bigcup\limits_{t \in L} P(t) & = \quad  \text{[\(P(L) = L\)]}\\
 L \enspace .
\end{align*}

\vspace*{-8mm}
\end{proof}
\begin{example}
\label{ex:cons1}
Let us give an example of the DBTA construction of Definition~\ref{def:up_TA_construction}.
Consider  \(L = \lang{\TA}\) where \(\TA\) is the BTA defined in Example~\ref{example:BTA}.
Namely, \(\TA = \tuple{Q,\Sigma,\delta,F}\) with \(Q = \{q_0,q_1\}\); \(\Sigma = \Sigma_0\cup \Sigma_2\) with \(\Sigma_0 =\{\true, \false\}\) and \(\Sigma_2= \{\AND,\OR\}\); \(\{q_0\}= \delta(\AND[q_0,q_1]) = \delta(\AND[q_0,q_0]) = \delta(\AND[q_1,q_0]) = \delta(\OR[q_0,q_0]) = \delta(\false[\,])\), \(\{q_1\} = \delta(\AND[q_1,q_1]) = \delta(\OR[q_1,q_0]) = \delta(\OR[q_0,q_1]) = \delta(\OR[q_1,q_1]) = \delta(\true[\,])\) and \(F = \{q_1\}\).
Thus, \(L\) is the set of all trees \(t \in \cT_{\Sigma}\) that yield to propositional formulas, over the binary connectives \(\land\) and \(\lor\) and the constants \(\true\) and \(\false\), that evaluate to \(\true\).

\medskip
On the other hand, consider the upward congruence from Example~\ref{ex:congruences}, i.e., \(t \ru r \udiff Lt^{-1} = Lr^{-1}\), for every \(t,r \in \cT_{\Sigma}\).
Note that, \(\ru\) defines two equivalence classes:
\begin{itemize}
  \item \(P_{\ru}(\true[\,]) = L\) and
  \item \(P_{\ru}(\false[\,]) = L^\complement\), where \(L^{\complement}\) is the complement of \(L\), i.e., the set \(\cT_{\Sigma} \setminus L \).
\end{itemize}
Specifically, \mbox{\(\{\lor[\false,\true], \lor[\true,\false], \lor[\true,\true], \land[\true,\true]\}\subseteq P_{\ru}(\true[\,])\)}$\,$ and \\
 $\{\land[\false,\true], \land[\true,\false], \land[\false,\false]$,  $\lor[\false, \false]\} \subseteq$ $P_{\ru}(\false[\,])$.

Thus, \(\cH^{\mathsf{u}}(\ru,L) = \tuple{Q', \Sigma_0 \cup \Sigma_2 , \delta', F'}\) where \(Q' = \{ L, L^{\complement}\}, F' = \{L\}\) %
 and \(\delta'\) is defined as \(\{L\} = \delta'(\true[\,]) = \delta'(\land([L,L])) =  \delta'(\lor([L^{\complement},L])) =  \delta'(\lor([L,L^{\complement}])) =  \delta'(\lor([L,L]))\) and \(\{L^{\complement}\} = \delta'(\false[\,]) =  \delta'(\land([L^{\complement},L])) =  \delta'(\land([L,L^{\complement}])) =  \delta'(\land([L^{\complement},L^{\complement}])) =  \delta'(\lor([L^{\complement},L^{\complement}]))\).\eox
\end{example}

Similarly, given a regular tree language \(L\) that is path-closed and a downward congruence \(\rd\),
that preserves the equivalence between downward quotients, i.e., \(x \rd y \Ra x^{-1}L = y^{-1}L\) for every \(x, y \in \cC_{\Sigma}\),
we give a tree automata construction that yields a BTA that defines exactly \(L\).
To this aim, we first recall an alternative characterization of path-closed languages due to Nivat and Podelski~\cite{nivat1997minimal}.
For simplicity, we state the characterization for symbols \(f \in \Sigma\) with \(\rank{f} \leq 2\), although it holds for alphabets of rank greater than \(2\).
Let \(L \subseteq \cT_\Sigma\) be a regular tree language, then \(L\) is \emph{path-closed} if{}f:
\begin{gather}
\forall x \in \cC_\Sigma, \forall f \in \Sigma \colon \notag\\
\subs{x}{f[t_1,t_2]},\subs{x}{f[r_1,r_2]} \in L \Ra \subs{x}{f[t'_1,t'_2]} \in L \text{ where } t'_1\in\{t_1,r_1\}, t'_2\in\{t_2,r_2\} \enspace \label{eq:pathClosed}
 .
\end{gather}
We recall that the definition of being path-closed is given at the end of Section~\ref{sec:preliminaries}.

Before turning to the construction using downward congruences we introduce a notation.
Given \(C \subseteq \cC_\Sigma\) and \(c \in \cC_\Sigma\), let
\(\subs{C}{c} \udiffg  \{\subs{c'}{c} \mid c' \in C\}\).
\begin{definition}[Automata construction \(\cH^{\mathsf{d}}(\rd, L)\)]\label{def:down_TA_construction}
Let \(\rd\) be a downward congruence and let \(L \subseteq \cT_{\Sigma}\).
Define the BTA \mbox{\(\cH^{\mathsf{d}}(\rd,L) \udiffg  \tuple{Q, \Sigma, \delta, F}\)} where \mbox{\(Q = \{ P_{\rd}(x) \mid x \in \cC_{\Sigma}, x^{-1}L \neq \varnothing\}\)}, \(F = \{P_{\rd}(\cX)\}\) and the transition function \(\delta\) is defined as follows:
\begin{enumerate}
\renewcommand\labelenumi{\theenumi}
\renewcommand{\theenumi}{(\roman{enumi})}
\item for every \(a \in \Sigma_0\), \(P_{\rd}(x) \in Q\colon P_{\rd}(x) \in \delta(a[\,])\) if{}f \(\subs{P_{\rd}(x)}{a} \subseteq L\), and
\item for every \(f \in \Sigma_n\) with \(n \geq 1\), \(P_{\rd}(x), P_{\rd}(x_i) \in Q\) with \(i=\range{\rank{f}}\): \(P_{\rd}(x)\) \linebreak \({}\in \delta(f[P_{\rd}(x_1), ... , P_{\rd}(x_{\rank{f}})]) \) if{}f
\(\exists t_1, ... , t_{\rank{f}} \in \cT_{\Sigma}\colon \subs{P_{\rd}(x)}{\subs{f[t_1, ... , t_{\rank{f}}]} {\cX}_i}\) \(\subseteq P_{\rd}(x_i)\), for every \(i=\range{\rank{f}}\).
\end{enumerate}
\end{definition}

\begin{remark}
\label{remark:no-determinism}
Note that \(\cH^{\mathsf{d}}(\rd, L)\) is \emph{well-defined} since we assume that \(\rd\) has finite index and is a downward congruence.
On the other hand, by definition, \(\cH^{\mathsf{d}}(\rd, L)\) does not contain empty states and the final set of states is a singleton.
Furthermore, \(\cH^{\mathsf{d}}(\rd, L)\) does not contain unreachable states either.
This is also a consequence of its definition, since \(Q\) only includes those blocks \(P_{\rd}(x)\) s.t. \(x^{-1} L \neq \varnothing\).
However, \(\cH^{\mathsf{d}}(\rd, L)\) is not necessarily co-deterministic since it might be the case that, for some \(f \in \Sigma_n, P_{\rd}(x), P_{\rd}(x_i), P_{\rd}(y_i) \in Q\) with \(i=\range{\rank{f}}\), there exists \(t_1, \ldots, t_{\rank{f}}, r_1, \ldots, r_{\rank{f}} \in \cT_{\Sigma}\) such that \(\subs{P_{\rd}(x)}{\subs{f[t_1, \ldots, t_{\rank{f}}]} {\cX}_i}\) \(\subseteq P_{\rd}(x_i)\), \(\subs{P_{\rd}(x)}{\subs{f[r_1, \ldots, r_{\rank{f}}]} {\cX}_i}\) \(\subseteq P_{\rd}(y_i)\) with \(x_i \not\rd y_i\).
\end{remark}

Now we prove that if \(L\), the language used to construct$\,$ \(\cH^{\mathsf{d}}(\rd,~L)\),$\,$ is path-closed then \linebreak  \(\lang{\cH^{\mathsf{d}}(\rd,L)}= L\).
Notice that if \(L\) is not path-closed then only \( L \subseteq \lang{\cH^{\mathsf{d}}(\rd, L)} \) is guaranteed (the reader may check this fact by looking at the proof of the lemma).
\begin{lemma}
\label{HTPreservesL}
Let \(L \subseteq \cT_{\Sigma}\) and \(\rd\) be a downward congruence such that \(x \rd y \Ra x^{-1}L = y^{-1}L\) for every \( x,y \in \cC_\Sigma\).
Then, \(\lang{\cH^{\mathsf{d}}(\rd,L)}\supseteq L\) holds and, if, incidentally, \(L\) is path-closed then \(\lang{\cH^{\mathsf{d}}(\rd,L)}\subseteq L\).

\end{lemma}
\begin{proof}
To simplify the notation, we denote \(P_{\rd}\), the partition induced by \(\rd\), simply by \(P\).
Let \(\cH\!=\!\cH^{\mathsf{d}}({\rd,}~L) = (Q, \Sigma, \delta, F)\).
First, we prove that for every \(P(x)\in Q\):
\begin{align}
  \mathcal{L}^{\cH}_{\da}(P(x)) &\subseteq x^{-1}L, \text{if \(L\) is path-closed, and }\label{eq:sup_downward_eq_quotient_language}\\
  \mathcal{L}^{\cH}_{\da}(P(x)) &\supseteq x^{-1}L\enspace .\label{eq:sub_downward_eq_quotient_language}
\end{align}

(\(\subseteq\)).
We show that, if \(L\) is path-closed, then for all \(t \in \cT_{\Sigma}\), \(t \in \mathcal{L}^{\cH}_{\da}(P(x)) \Ra t \in x^{-1}L\).
We proceed by induction in the height of the tree \(t\).

\begin{itemize}
\item \textit{Base case:} Let \(t\) be of height \(0\), i.e. \(\exists a \in \Sigma_0\) such that \(t = a[\,]\).
Then.
\begin{align*}
a[\,] \in \mathcal{L}^{\cH}_{\da}(P(x)) & \Ra \quad \text{[Def. of \(\mathcal{L}^{\cH}_{\da}(P(x))\)]} \\
P(x) \in \delta(a[\,]) & \Ra \quad \text{[Def. \ref{def:down_TA_construction}]} \\
\subs{P(x)}{a} \subseteq L & \Ra \quad \text{[\(P(x) \supseteq \{x\} \)]} \\
\subs{x}{a} \in L & \Ra\quad  \text{[Def. of downward quotient]}\\
a[\,] \in x^{-1}L \enspace .
\end{align*}
\item \textit{Inductive step:} Now we assume by hypothesis of induction that \(t \in \mathcal{L}^{\cH}_{\da}(P(x)) \Ra t \in x^{-1}L\) for all trees of height up to \(n\).
Let \(t\) be a tree of height \(n{+}1\), i.e. \(t = f[t_1,t_2]\) for some \(f \in \Sigma\) and  \(t_1, t_2 \in \cT_\Sigma\) have height up to \(n \geq 0\).
Note that, w.l.o.g., we assume \(\rank{f} = 2\) for the sake of clarity.
Then,
\begin{align}
t \in \mathcal{L}^{\cH}_{\da}(P(x)) & \Lra \notag\\
\exists t' \in \cT_Q \colon  t'(\varepsilon) = P(x),\, t \rightarrow^*_{\cH} t'  & \Lra \nonumber\\
\exists P(x_1),P(x_2) \in Q \colon  t_1 \in \mathcal{L}^{\cH}_{\da}(P(x_1)), t_2\in \mathcal{L}^{\cH}_{\da}(P(x_2)) &\text{ and } \nonumber\\
 P(x) \in \delta(f[P(x_1),P(x_2)]),\label{eq:proofHTPreservesL-1}
\end{align}
where the first double-implication holds by definition of \(\mathcal{L}^{\cH}_{\da}(P(x))\) and the second one holds by Lemma~\ref{DownwardLanguages}, since \(t = f[t_1,t_2]\).
By Definition~\ref{def:down_TA_construction}, Equation~\eqref{eq:proofHTPreservesL-1} is equivalent to:
\begin{align*}
\exists r_1, r_{2} \in \cT_{\Sigma} \colon  &\subs{P(x)}{\subs{f[r_1, r_2]}{\cX}_1} \subseteq P(x_1),\\
&\subs{P(x)}{\subs{f[r_1, r_2]}{\cX}_2} \subseteq P(x_2)\enspace .
\end{align*}
By definition of \(P\), we have that:
\begin{align}
\exists r_1, r_{2} \in \cT_{\Sigma} \colon  x_1 &\rd \subs{x}{\subs{f[r_1, r_2]}{\cX}_1},\notag\\
x_2 &\rd \subs{x}{\subs{f[r_1, r_2]}{\cX}_2} \label{eq:proofHTPreservesL-2} \enspace .
\end{align}
Note that, by H.I., \(t_1 \in \mathcal{L}^{\cH}_{\da}(P(x_1)), t_2\in \mathcal{L}^{\cH}_{\da}(P(x_2)) \Ra t_1 \in x_1^{-1}L, t_2 \in x_2^{-1}L\).
Moreover, \(x \rd y\) implies that \(x^{-1}L = y^{-1}L\).
Therefore, it follows from Equations~\eqref{eq:proofHTPreservesL-1} and~\eqref{eq:proofHTPreservesL-2} that \(\exists r_1, r_2 \in \cT_{\Sigma}\colon  t_1 \in \subs{x}{\subs{f[r_1, r_2]}{\cX}_1}^{-1}L\) and \(t_2 \in \subs{x}{\subs{f[r_1,r_2]}{\cX}_2}^{-1}L\), which can be rewritten as \(\exists r_1, r_2 \in \cT_{\Sigma}\colon  \subs{x}{f[t_1, r_2]} \in L\) and \(\subs{x}{f[r_1, t_2]} \in L\).
Finally, since \(L\) is path-closed and \(t=f[t_1,t_2]\), we conclude that \(\subs{x}{t} \in L\), i.e., \(t \in x^{-1}L\).
\end{itemize}

(\(\supseteq\)). We next show that, for all \(t \in \cT_{\Sigma}\), \(t \in x^{-1}L \Ra t \in \mathcal{L}^{\cH}_{\da}(P(x))\).
We proceed by induction on the height of \(t\).
\begin{itemize}
\item \textit{Base case:} Let \(t\) be of height \(0\), i.e. \(\exists a \in \Sigma_0\) such that \(t = a[\,]\).
Then,
\begin{align*}
a[\,] \in x^{-1}L & \Ra \quad \text{[Def. of downward quotient]} \\
\subs{x}{a} \in L  & \Ra \quad \text{[\(\rd\) is a downward cong.]} \\
\subs{P(x)}{a} \subseteq L  & \Ra \quad \text{[Def.~\ref{def:down_TA_construction}]} \\
P(x) \in \delta(a[\,]) & \Ra \quad \text{[Def. of \(\mathcal{L}^{\cH}_{\da}(P(x))\)]} \\
a[\,] \in \mathcal{L}^{\cH}_{\da}(P(x)) \enspace .
\end{align*}

\item \textit{Inductive step:}
Now we assume by hypothesis of induction that \(t \in x^{-1}L \Ra t \in \mathcal{L}^{\cH}_{\da}(P(x))\) holds for all trees \(t\) of height up to \(n\).
Let \(t\) be a tree of height \(n{+}1\), i.e. \(t = f[t_1,t_2]\) for some \(f \in \Sigma\) and \(t_1,t_2 \in \cT_\Sigma\) with height up to \(n\).
Note that, w.l.o.g., we assume \(\rank{f} = 2\) for the sake of clarity.
Let \(x_1 = \subs{x}{f[\cX, t_2]}\) and \(x_2 = \subs{x}{f[t_1, \cX]}\).
Since \(\rd\) is a downward congruence, we have that \(\subs{P(x)}{f[\cX, t_2]} \subseteq P(x_1)\) and \(\subs{P(x)}{f[t_1, \cX]}\subseteq P(x_2)\).
Therefore, by Definition~\ref{def:down_TA_construction}, \(P(x) \in \delta(f[P(x_1), P(x_2)])\).

On the other hand, note that \(t_1 \in {x_1}^{-1}L\) and \(t_2 \in {x_2}^{-1}L\).
By I.H., \(t_1 \in \mathcal{L}^{\cH}_{\da}(P(x_1))\) and \(t_2 \in \mathcal{L}^{\cH}_{\da}(P(x_2))\).
Relying on Lemma~\ref{DownwardLanguages}, we conclude that \(t \in \mathcal{L}^{\cH}_{\da}(P(x))\).
\end{itemize}

We finish by proving the two statements of the Lemma.
\begin{align*}
  \lang{\cH} &= \quad \text{[Def. of \(\lang{\cH}\)]} & \lang{\cH} &= \quad \text{[Def. of \(\lang{\cH}\)]}\\
  \textstyle{\bigcup_{q \in F}} \mathcal{L}^{\cH}_{\da}(q) & = \quad \text{[Def.~\ref{def:down_TA_construction}]} & \textstyle{\bigcup_{q \in F}} \mathcal{L}^{\cH}_{\da}(q) & = \quad \text{[Def.~\ref{def:down_TA_construction}]}\\
\mathcal{L}^{\cH}_{\da}(P(\cX)) & \supseteq \quad \text{[Equation~\eqref{eq:sub_downward_eq_quotient_language}]} & \mathcal{L}^{\cH}_{\da}(P(\cX)) & \subseteq \quad \text{[Equation~\eqref{eq:sup_downward_eq_quotient_language}, $L$ path-closed]}\\
(\cX)^{-1}L & = L & (\cX)^{-1}L & = L \enspace.
\end{align*}
\end{proof}

Next we introduce the notion of \emph{strongly downward congruence} which defines an extra condition on a downward congruence \(\rd\) that guarantees that \(\cH^{\mathsf{d}}(\rd, L)\) is a co-DBTA.
This condition is an ad-hoc design so as to avoid the issue raised at Remark~\ref{remark:no-determinism}.
Intuitively, strongly downward means that every congruent pair \(x\) and \(y\) of contexts remains congruent after plugging the context \(\subs{t}{\cX}_i\) in \(x\) and the context \(\subs{r}{\cX}_i\) in \(y\) where \(t\) and \(r\) are chosen such that \(\subs{x}{t},\subs{y}{r} \in L\).
Recall that given a tree \(t\in\cT\), \(\subs{t}{\cX}_i\in \cC\) denotes the context obtained by replacing the subtree rooted at \(i\) with \(\cX\).
\begin{definition}[Strongly downward congruence]\label{def:strong_downward_congruence}
Let \(\rd\) be a downward congruence and let \(L \subseteq \cT_{\Sigma}\).
We say \(\rd\) is \emph{strongly downward} w.r.t. \(L\) if{}f for every \( x,y \in\cC_\Sigma\) with \(x\rd y\), \(t \in x^{-1}L, r \in y^{-1}L \text { and } t(\varepsilon) \!=\! r(\varepsilon)\) where \(t,r \in \cT_\Sigma\), we have that \( \subs{x}{\subs{t}{\cX}_i} \rd \subs{y}{\subs{r}{\cX}_i}, \forall i\in\{\range{\rank{t(\varepsilon)}}\}\).
\end{definition}

Finally, we obtain the following lemma.
\begin{lemma}
\label{downwardIsDet}
Let \(\rd\) be a strongly downward congruence w.r.t. \(L\) such that \(x \rd y \Ra x^{-1}L = y^{-1}L\) for every \(x,y \in \cC_{\Sigma}\).
Then \(\cH^{\mathsf{d}}(\rd,L)\) is a co-DBTA.
\end{lemma}
\begin{proof}
We prove that \(\cH^{\mathsf{d}}(\rd,L) = (Q, \Sigma, \delta, F)\) is a co-DBTA by contradiction.

\medskip
If \(\cH^{\mathsf{d}}(\rd,L)\) is not co-deterministic, then either \(|F|>1\), which is not possible by Definition~\ref{def:down_TA_construction}; or \(\exists f \in \Sigma, \exists x_1,\ldots,x_{\rank{f}},y_1,\ldots,y_{\rank{f}},x \in \cC_\Sigma\) such that \(P_{\rd}(x) \in \delta(f[P_{\rd}(x_1), \ldots, P_{\rd}(x_{\rank{f}})])\) and \(P_{\rd}(x) \in \delta(f[P_{\rd}(y_1), \ldots, P_{\rd}(y_{\rank{f}})])\) and \(\exists {i_0} \in \range{\rank{f}}\) for which \(x_{i_0} \not\rd y_{i_0}\).

Then, by Definition~\ref{def:down_TA_construction}, \(\exists t_1,\ldots, t_{\rank{f}}, r_1,\ldots, r_{\rank{f}}\!\in\! \cT_{\Sigma}\ \colon \!\subs\!{P_{\rd}(x)}{\subs{f[t_1, \ldots, t_{\rank{f}}]} {\cX}_{i_0}} \subseteq P_{\rd}(x_{i_0})\), and \(\subs{P_{\rd}(x)}{\subs{f[r_1, \ldots, r_{\rank{f}}]} {\cX}_{i_0}}\) \(\subseteq P_{\rd}(y_{i_0})\) with \(x_{i_0} \not\rd y_{i_0}\).
Thus,  \(\subs{x}{\subs{f[t_1, \ldots, t_{\rank{f}}]}{\cX}_{i_0}} \rd x_{i_0}\) and \(\subs{x}{\subs{f[r_1, \ldots, r_{\rank{f}}]}{\cX}_{i_0}} \rd y_{i_0}\).

On the other hand, since \(P_{\rd}(x_{i_0})\) and \(P_{\rd}(y_{i_0})\) are states of \(\cH^{\mathsf{d}}\), we have that \(x_{i_0}^{-1}L \neq \varnothing\) and \(y_{i_0}^{-1}L \neq \varnothing\).
Since  \(\subs{x}{\subs{f[t_1, \ldots, t_{\rank{f}}]}{\cX}_{i_0}} \rd x_{i_0}\) and \(\subs{x}{\subs{f[r_1, \ldots, r_{\rank{f}}]}{\cX}_{i_0}} \rd y_{i_0}\), we also have that \(\subs{x}{\subs{f[t_1, \ldots, t_{\rank{f}}]}{\cX}_{i_0}}^{-1}L \neq \varnothing\) and \(\subs{x}{\subs{f[r_1, \ldots, r_{\rank{f}}]}{\cX}_{i_0}}^{-1}L \neq \varnothing\).
Thus, for some \(t_{i_0}, r_{i_0} \in\cT_{\Sigma}\), \(t = f[t_1, \ldots,t_{i_0}, \dots, t_{\rank{f}}]\in x^{-1}L\) and \(r = f[r_1, \ldots,r_{i_0}, \dots, r_{\rank{f}}]\in x^{-1}L\).

Finally, since \(\rd\) is a strongly downward congruence w.r.t. \(L\), \(x \rd x\) (trivially), \(t,r \in x^{-1}L\) and \(t(\varepsilon) = r(\varepsilon)\),  we have that \(\subs{x}{\subs{f[t_1, \ldots, t_{\rank{f}}]}{\cX}_{i_0}}\rd \subs{x}{\subs{f[r_1, \ldots, r_{\rank{f}}]}{\cX}_{i_0}}\).
By transitivity of \(\mathord{\rd}\), we have that \(x_{i_0} \rd y_{i_0}\), which is a contradiction.

Thus, we conclude that \(\cH^{\mathsf{d}}(\rd,L)\) is co-deterministic.
\end{proof}

\begin{example}
\label{ex:cons2}
Now we give an example of the construction of Definition~\ref{def:down_TA_construction}. %
Consider the BTA \(\TA\) defined in Example~\ref{example:co-BTA}, i.e., \(\TA = \tuple{Q,\Sigma,\delta,F}\) with \(Q = \{q_0,q_1\}\); \(\Sigma = \Sigma_0\cup \Sigma_2\) with \(\Sigma_0 =\{\true,\false\}\) and \(\Sigma_2= \{\AND\}\); \(\{q_0\}= \delta(\AND[q_0,q_1]) = \delta(\AND[q_0,q_0]) = \delta(\AND[q_1,q_0]) = \delta(\false[\,])\), \(\{q_1\} = \delta(\AND[q_1,q_1]) =  \delta(\true[\,])\) and \(F = \{q_1\}\).
Recall that \(L = \lang{\TA}\) is the set of all trees \(t \in \cT_{\Sigma}\) that yield to propositional formulas, over the binary connective \(\land\) and the constants \(\true\) and \(\false\), that evaluate to \(\true\).
Hence, \(\lang{\TA}\) is path-closed.

\medskip
On the other hand, consider the downward congruence from Example~\ref{ex:congruences}, i.e., \(x \rd y \udiff x^{-1}L = y^{-1}L\), for every \(x,y \in \cC_{\Sigma}\).
As we shall see in Lemma~\ref{NerodeStronglyDownward}, since \(L\)  is path-closed, \(\rd\) is a strongly downward congruence w.r.t. \(L\).
It turns out that \(\rd\) defines two equivalence classes:
\begin{itemize}
\itemsep=0.9pt
\item \(P_{\rd}(\land[\true,\cX])\), with \((\land[\true, \cX])^{-1}L = L \), and
\item \(P_{\rd}(\land[\false,\cX])\), with \((\land[\false, \cX])^{-1}L = \varnothing\).
\end{itemize}
Intuitively,  \(P_{\rd}(\land[\true,\cX])\) is composed of those contexts \(x\) such that \(\subs{x}{\true[\,]}\) corresponds to a propositional formula that evaluates to \(\true\).
For instance, \(\{\land[\cX, \true], \cX[\,]\}\subseteq P_{\rd}(\land[\true,\cX]) \).
On the other hand, \(P_{\rd}(\land[\false,\cX])\) is composed of those such that \(\subs{x}{\true[\,]}\) is a formula that evaluates to \(\false\).
For example, \(\land[\cX, \false]\in P_{\rd}(\land[\false,\cX])\).

Thus, \mbox{\(\cH^{\mathsf{d}}(\ru,L) = \tuple{Q', \Sigma, \delta', F'}\) with \(Q' = \{ P_{\rd}(\land[\true,\cX])\}\)} and \(F' = \{P_{\rd}(\land[\true,\cX])\}\) since \linebreak \(P_{\rd}(\cX) = P_{\rd}(\land[\true,\cX]) \).
Note that \(P_{\rd}(\land[\false,\cX]) \notin Q'\) since \((\land[\false,\cX])^{-1}L = \varnothing\).
In fact, including \(P_{\rd}(\land[\false,\cX])\) in \(Q'\) would result in an unreachable state since:
\begin{enumerate}
\renewcommand\labelenumi{\theenumi}
\renewcommand{\theenumi}{(\roman{enumi})}
\item for every \(a \in \{\false,\true\}\), \(\subs{P_{\rd}(\land[\false,\cX])}{a} \cap L = \varnothing\), and
\item (trivially) \(\nexists t \in {(\land[\false,\cX])}^{-1}L\) \(\colon t(\varepsilon) = \land\) and \(\subs{P_{\rd}(\land[\false,\cX])}{\subs{t}{\cX}_i} \subseteq P_{\rd}(x_i)\), with \(x_i \in \cC_{\Sigma}\) and \(i \in \{1,2\}\).
\end{enumerate}
Finally, \(P_{\rd}(\land[1,\cX]) \in \delta' (\land[P_{\rd}(\land[\true,\cX]), P_{\rd}(\land[\true,\cX])])\) (set \(t = \land[\true,\true]\) in Definition~\ref{def:down_TA_construction}).\eox
\end{example}

It is worth noting that the DBTA \(\cH^{\mathsf{u}}\) built in Example~\ref{ex:cons1} and the co-DBTA \(\cH^{\mathsf{d}}\) from Example~\ref{ex:cons2} are indeed the \emph{minimal} DBTA and co-DBTA for their corresponding languages, respectively.
This is due to the fact that we used the Myhill-Nerode relation for tree languages (and its counterpart downward congruence).
We will recall the main properties of these congruences and introduce their approximation using BTAs in the following section.

\section{Language-based equivalences and their approximation using BTAs}\label{sec:CongruencesandBTA}

In this section, we instantiate the automata constructions from the previous section using two classes of congruences: \emph{language-based} congruences, whose definition relies on a regular tree language; and \emph{BTA-based} congruences, whose definition relies on a BTA.

\begin{definition}[Language-based equivalences]\label{def:languageEquivalences}
Let \(L\subseteq \cT_{\Sigma}\) be a tree language and let \(t, r \in \cT_{\Sigma}\) and \(x,y \in \cC_{\Sigma}\).
The \emph{upward and downward language-based equivalences} are, respectively:
\begin{align*}
t \ru_L r &\udiff L\,t^{-1} = L\,r^{-1}\\
x \rd_L y &\udiff x^{-1}L = y^{-1}L\enspace.
\end{align*}
\end{definition}

These equivalence relations are always congruences.
The upward language-based equivalence is also known as the Myhill-Nerode relation for tree languages.
As shown by Kozen~\cite{kozen1993MyhillTrees}, given a tree language \(L \subseteq \cT_{\Sigma}\), the Myhill-Nerode relation is an upward congruence of finite index if{}f \(L\) is regular.
Note that this congruence is the coarsest upward congruence enabling Lemma~\ref{HBPreservesL} (set \(\mathord{\ru}\) to be \(\mathord{\ru_L}\)).
Similarly, Nivat and Podelski~\cite{nivat1997minimal} showed that the downward language-based equivalence has finite index if{}f \(L\) is regular and it is the coarsest downward congruence that enables Lemma~\ref{HTPreservesL}.
Next we prove that \(\rd_L\) is a strongly downward congruence w.r.t. \(L\) when \(L\) is path-closed, hence
it is the coarsest downward congruence enabling Lemma~\ref{downwardIsDet} (set \(\rd\) to \(\rd_L\)).

\begin{lemma}
\label{NerodeStronglyDownward}
Let \(L \subseteq \cT_\Sigma\) be a path-closed language.
Then, \(\rd_L\) is a strongly downward congruence w.r.t. \(L\).
\end{lemma}
\begin{proof}
First, we show that \(\rd_L\) is a downward congruence by contradiction.
Let \(x,y,c \in \cC_\Sigma\) be such that \(x \rd_L y\) and \(\subs{x}{c} \not\rd_L \subs{y}{c}\).
Then, w.l.o.g, \(\exists t \in \cT_\Sigma\) such that \(\subs{x}{\subs{c}{t}} \in L\) and \(\subs{y}{\subs{c}{t}} \notin L\), hence \(\subs{c}{t} \in x^{-1}L\) while \(\subs{c}{t} \notin y^{-1}L\), which contradicts the fact that \(x \rd_L y\).

Now we show that \(\rd_L\) is strongly downward w.r.t. \(L\).
Consider \(x, y \in \cC_\Sigma\) s.t. \(x \rd_L y\) and \(t, r \in \cT_{\Sigma}\) with \(t \in x^{-1}L\), \(r \in y^{-1}L\) and \(t(\varepsilon)= r(\varepsilon)\).
Assume that \(\rd_L\) is not a strongly downward congruence, i.e. there exists \(i_0 \in \range{\rank{f}}\) such that \(\subs{x}{\subs{t}{\cX}_{i_0}} \not\rd_L \subs{y}{\subs{r}{\cX}_{i_0}}\).
Then, there exists \(s \in \cT_\Sigma\) such that \(\subs{x}{\subs{t}{s}_{i_0}} \in L\), while \(\subs{y}{\subs{r}{s}_{i_0}} \notin L\).

On the one hand, since \(r \in y^{-1}L\), there exists \(s' \in \cT_{\Sigma}\) such that \(\subs{y}{\subs{r}{s'}_{i_0}} \in L\).
Since \(x \rd_L y\), i.e., \(x^{-1}L = y^{-1}L \), we have that \(\subs{x}{\subs{r}{s'}_{i_0}} \in L\).
Therefore, since \(\subs{x}{\subs{t}{s}_{i_0}} \in L\), \(\subs{x}{\subs{r}{s'}_{i_0}} \in L\) and  \(L\) is path-closed, by definition (see \eqref{eq:pathClosed} in Section~\ref{sec:AutomataConstruction}), we have that \(\subs{x}{\subs{r}{s}_{i_0}} \in L\).
Since \(x \rd_L y\), it follows that \(\subs{y}{\subs{r}{s}_{i_0}} \in L\) as well, which is a  contradiction.

We conclude that \(\rd_L\) is a strongly downward congruence w.r.t. \(L\).
\end{proof}

Now, we propose congruences based on the states of a given BTA.
These BTA-based congruences are finer than (or equal to) the corresponding language-based ones and are thus said to \emph{approximate} the language-based congruences.

\medskip
In order to define the downward BTA-based congruences, we first introduce the notion of \emph{root-to-pivot equivalence} between two contexts w.r.t. a BTA \(\TA = \tuple{Q, \Sigma, \delta, F}\) and a subset of states \(S\subseteq Q\).
Intuitively, two contexts \(x,y \in \cC_{\Sigma}\) are root-to-pivot equivalent w.r.t. \(S\) if{}f
\begin{enumerate}
\renewcommand\labelenumi{\theenumi}
\renewcommand{\theenumi}{(\roman{enumi})}
\itemsep=0.85pt
\item their pivots (the nodes with the label \(\cX\)) coincide and so do the paths from their roots to their pivots, and
\item either \(x\) and \(y\) belong to the upward language of some state (w.r.t. \(S\)), or none of them does.
\end{enumerate}
\begin{definition}[Root-to-pivot equivalence]
\label{def:root-pivot}
Let \({\TA} = \tuple{Q, \Sigma, \delta, F}\) be a BTA and \(S \subseteq Q\).
We say that \(x, y \in \cC_{\Sigma}\) are \emph{root-to-pivot equivalent} w.r.t. \(S\), denoted by \(x \simp^S y\), if{}f
\begin{enumerate}
\renewcommand\labelenumi{\theenumi}
\renewcommand{\theenumi}{(\roman{enumi})}
\itemsep=0.85pt
\item \(\piv(x)=\piv(y)\) and \(\forall p,p'\in (\mathbb{N}_+)^*\) if \(p\concat p' = \piv(x) \) then \(x(p) = y(p) \), and
\item \(x \in \mathcal{L}_{\ua}(q,S)\) for some \(q\in Q\) if and only if \(y \in \mathcal{L}_{\ua}(q',S)\) for some \(q' \in Q\).
\end{enumerate}
The root-to-pivot equivalence need not have finite index.
\end{definition}

We will simply write \(x \simp y\) when \(S = F\).   Let us fix intuitions with the following example.
\eject
\begin{example}
\label{ex:root-to-pivot}
Consider the BTA \(\TA\) given in Example~\ref{example:BTA}, i.e., \(\TA = \tuple{Q,\Sigma_0\cup \Sigma_2,\delta,F}\) with \(Q = \{q_0,q_1\}\), \(\Sigma_0 =\{\true,\false\}\), \(\Sigma_2= \{\AND,\OR\}\), \(\{q_0\}= \delta(\AND[q_0,q_1]) = \delta(\AND[q_0,q_0]) = \delta(\AND[q_1,q_0]) = \delta(\OR[q_0,q_0]) = \delta(\false[\,])\), \(\{q_1\} = \delta(\AND[q_1,q_1]) = \delta(\OR[q_1,q_0]) = \delta(\OR[q_0,q_1]) = \delta(\OR[q_1,q_1]) = \delta(\true[\,])\) and \(F = \{q_1\}\).
Thus, \(\lang{\TA}\) is defined as the set of all trees of the form \(t \in \cT_{\Sigma_0 \cup \Sigma_2}\) which yield to propositional formulas, over the binary connectives \(\land\) and \(\lor\) and the constants \(\true\) and \(\false\), that evaluate to \(\true\).

Let \(S = F\), then we have that \(x = \lor[\lor[\true,\false], \cX]\) and \(y = \lor[\lor[\true,\true], \cX]]\) are root-to-pivot equivalent w.r.t. \(S\), i.e., \(x \simp y\).
In fact, \(\piv(x) = \piv(y) = 2\) and for all \(p,p'\in (\mathbb{N}_+)^*\) if \(p\concat p' = \piv(x) \) then \(x(p) = y(p) \).
Moreover, \(q_1\in Q\) is s.t. \(x \in \mathcal{L}_{\ua}(q_1,S)\) and \(y \in \mathcal{L}_{\ua}(q_1,S)\).

On the other hand, we have that \(x' =\land[\land[\true,\true], \cX]\) and \(y' = \land[\land[\true,\false], \cX]\) are \emph{not} root-to-pivot equivalent w.r.t. \(S\).
Note that, despite of having that \(\piv(x') = \piv(y') = 2\) and for all \(p,p'\in (\mathbb{N}_+)^*\) if \(p\concat p' = \piv(x) \) then \(x(p) = y(p) \), \(\nexists q \in Q \colon y' \in \mathcal{L}_{\ua}(q, S)\), while \(x' \in \mathcal{L}_{\ua}(q_1,S) \).
In fact, \((y')^{-1}L = \varnothing\) while \((x')^{-1}L  \neq \varnothing\).
Intuitively, in the absence of unreachable states, the second condition in the definition of root-to-pivot equivalence prevents from defining two equivalent contexts when one of them defines an empty downward quotient and the other does not.
\eox
\end{example}
Now we define the \(\post(\cdot)\) and \(\pre(\cdot)\) operators for trees.
Given a BTA \(\TA\), a tree \(t \in \cT_{\Sigma}\) and a set \(S \subseteq Q\), \(\post^{\TA}_t(S)\) is the set of states appearing at the root after \(\TA\) is done processing the tree \(t\) provided that this processing labelled the leaves of \(t\) with states in \(S\).
On the other hand, given a context \(x \in \cC_{\Sigma}\) and \(S \subseteq Q\), \(\pre^{\TA}_x(S)\) contains the states \(q\) such that the result of \(\TA\) processing \(\subs{x}{q}\) has its root labelled with a state in \(S\).
Furthermore, \(\pre^{\TA}_x(S)\) contains the states \(q\) obtained as above this time starting from \(\subs{y}{q}\) where \(y\in \cC_\Sigma\) is root-to-pivot equivalent to \(x\).

\begin{definition}[Post and pre operators]\label{def:postpre}
Let \({\TA} = \tuple{Q, \Sigma, \delta, F}\) be a BTA and let \(t \in \cT_\Sigma\), \(x \in \cC_\Sigma\) and \(S \subseteq Q\).
Define: %
\begin{align*}
\post^{\TA}_t(S) &\udiffg  \{q \in Q \mid t \in \mathcal{L}^{\TA}_{\da}(q,S)\}\\
\pre^{\TA}_x(S) & \udiffg  \{q \in Q \mid x \in P_{\simp^S}(\mathcal{L}^{\TA}_{\ua}(q, S))\} \enspace . \vspace*{-2mm}
\end{align*}
\end{definition}

As usual, we will omit the superscript \(\TA\) when it is clear from the context.
Note that, for every \(x \in \cC_{\Sigma}, S \subseteq Q \colon \pre_x(S) = \varnothing \Lra \nexists q \in Q \colon x \in \mathcal{L}_{\ua}(q,S)\).

Let us illustrate these definitions with an example.
\begin{example}
As in the previous example, consider the BTA \(\TA\) from Example~\ref{example:BTA}.
Given the tree \(t = \land[\lor[\true,\false], \land[\true,\true]]\), we have that \(\post_{t}(\{q_0,q_1\}) = \{q_1\}\), while \(\post_{t}(\{q_0\}) = \varnothing\).

On the other hand, given the context \(x = \lor[\land[\true,\false], \cX]\), we have that \(x \in P_{\simp}(\mathcal{L}_{\ua}(q_1))\) since \(x \in \mathcal{L}_{\ua}(q_1)\).
Thus, \(q_1 \in \pre_x(F)\).
Moreover, \(x \in P_{\simp}(\mathcal{L}_{\ua}(q_0))\) since \( y = \lor[\land[\true,\true], \cX]\) is s.t. \(y \simp x\) and \(y \in \mathcal{L}_{\ua}(q_0)\).
Hence, \(\pre_x(F) = \{q_0, q_1\}\).

\medskip
Finally, given the context \(x' = \land[\lor[\false,\false], \cX]\),  \(\pre_{x'}(F) = \varnothing\) since \(\nexists q \in Q \colon x' \in \mathcal{L}_{\ua}(q)\).
\eox
\end{example}

Note that the notion of root-to-pivot equivalence in the definition of \(\pre^{\TA}_x(S)\) allows us to prove that the downward BTA-based equivalence (see Definition~\ref{def:automataEquivalences}) is
\begin{enumerate}
\renewcommand\labelenumi{\theenumi}
\renewcommand{\theenumi}{(\roman{enumi})}
\itemsep=0.9pt
\item strongly downward w.r.t. \(\lang{\TA}\) (see Lemma~\ref{automataCongruences}\ref{lemma:automataCongruences:downAT}) and
\eject
\item mimics the construction of each set \(R_i\), from the definition of  bottom-up co-determinization given in Section~\ref{sec:preliminaries}, i.e., we produce the same set of states as that of \(\TA^{\mathsf{cD}}\) by constructing the sets \(\pre^{\TA}_x(F)\), with  \(x \in \cC_\Sigma\).
\end{enumerate}
The reader is referred to Example~\ref{example:new-pre} for fixing these intuitions.

Our next step is to establish basic properties of the \(\post(\cdot)\) and \(\pre(\cdot)\) operator, the upward and downward languages they induce and their relationships with upward and downward quotients.
Before to establish those properties, we need three technical lemmas.

\begin{lemma}\label{lemma:simpImpliesQuotientEqual}
Let \(\cA = \tuple{Q,\Sigma,\delta,F}\) be a BTA without  unreachable states and such that \(L = \lang{\TA}\) is a path-closed language.
Then, \( x \simp y  \Ra x^{-1}L = y^{-1}L\), for every \(x,y \in \cC_{\Sigma}\).
\end{lemma}
\begin{proof}
Recall that \(\simp\) shortly denotes \(\simp^F\).
By definition, \mbox{\(x \simp y\)} if{}f:
\begin{enumerate}
\renewcommand\labelenumi{\theenumi}
\renewcommand{\theenumi}{(\roman{enumi})}
\item \(\piv(x)=\piv(y)\) and \(\forall p,p'\in (\mathbb{N}_+)^*\) if \(p\concat p' = \piv(x) \) then \(x(p) = y(p) \), and
\item \(\exists q \in Q \colon x \in \mathcal{L}_{\ua}(q) \Lra \exists q' \in Q \colon y \in \mathcal{L}_{\ua}(q')\).
\end{enumerate}

Note that if \(x \simp y\) then the \(\cX\)-height of \(x\) and \(y\) coincide.
The proof goes by induction on the \(\cX\)-height of \(x\) and \(y\).
\begin{itemize}
\itemsep=0.9pt
\item \emph{Base case:} Let \(\height_{\cX}(x) = \height_{\cX}(y) =1\).
Then \(x = y = \cX[\,]\) and \(x^{-1}L = y^{-1}L = L\).

\item \emph{Inductive step:} Assume that the hypothesis holds for all contexts up to \(\cX\)-height \(n\).
Let \(x,y \in \cC_\Sigma\)  be such that \(x \simp y\) and \(\height_{\cX}(x) = \height_{\cX}(y) = n{+}1\).

First, assume that \(\nexists q \in Q \colon x \in \mathcal{L}_{\ua}(q)\) and \(\nexists q' \in Q \colon y \in \mathcal{L}_{\ua}(q')\).
Then, since \(\TA\) has no unreachable states we have that \(x^{-1}L = y^{-1}L = \varnothing\).

On the other hand, assume that \(\exists q \in Q \colon x \in \mathcal{L}_{\ua}(q)\) and \(\exists q' \in Q \colon y \in \mathcal{L}_{\ua}(q')\).
Since \(\TA\) has no unreachable states, we have that \(\exists t,r \in \cT_\Sigma \colon \subs{x}{t},\subs{y}{r} \in L\).

It is easy to check that \(\exists x'\!,x''\!,y'\!,y''\! \in \cC_\Sigma \colon \height_{\cX}(x') = \height_{\cX}(y') = n{-}1\), \( \height_{\cX}(x'')=\height_{\cX}(y'') = 2\), \(\subs{x'}{x''} = x \) and \(\subs{y'}{y''}=y\).
On the other hand, since \(\exists q \in Q \colon x \in \mathcal{L}_{\ua}(q)\) and \(\exists q' \in Q \colon y \in \mathcal{L}_{\ua}(q')\) then \(\exists \tilde{q} \in Q \colon x' \in \mathcal{L}_{\ua}(\tilde{q})\) and \(\exists \tilde{q}' \in Q \colon y' \in \mathcal{L}_{\ua}(\tilde{q}')\).
Thus, we have that \(x' \simp y'\).

W.l.o.g. and for the sake of clarity, let us consider \( f \in \Sigma\) with \(\rank{f} = 2\) such that \(x'' = f[\cX, t_2]\) and \(y'' = f[\cX,r_2]\) with \(t_2,r_2 \in \cT_\Sigma\).
Then:
\begin{align*}
\subs{x}{t} \in L, \subs{y}{r} \in L & \Lra  \\
\subs{x'}{\subs{x''}{t}} \in L, \subs{y'}{\subs{y''}{r}} \in L & \Ra^{\dagger} \\
\subs{y'}{\subs{x''}{t}} \in L, \subs{y'}{\subs{y''}{r}}  \in L & \Lra \\
\subs{y'}{f[t,t_2]} \in L, \subs{y'}{f[r,r_2]} \in L
& \Ra^{\dagger\dagger} \\
\subs{y'}{f[t,r_2]} \in L & \Lra  \\
\subs{y'}{\subs{y''}{t}} \in L & \Lra \\
\subs{y}{t} \in L \enspace .
\end{align*}

Note that implication \(\dagger\) holds since, by H.I., \((y')^{-1}L = (x')^{-1}L\).
On the other hand, implication \(\dagger\dagger\) holds since \(L\) is path-closed.

We conclude that \(x^{-1}L \subseteq y^{-1}L\).
The proof for the reverse inclusion is symmetric.
Therefore, \(x^{-1}L = y^{-1}L\).
\end{itemize}

\vspace*{-7mm}
\end{proof}

\begin{lemma}
\label{lemma:subs-pre}
Let \(\cA = \tuple{Q,\Sigma,\delta,F}\) be a BTA without unreachable states and such that \(L = \lang{\TA}\) is path-closed then, for every \(x,y \in \cC_{\Sigma}\):
\[
\exists q \in Q \colon \subs{x}{y} \in \mathcal{L}_{\ua}(q) \text{ if{}f }
\exists q_1 \in Q \colon y \in \mathcal{L}_{\ua}(q_1,\pre_x(F)) \enspace .
\]
\end{lemma}
\begin{proof}
Recall that \(\mathcal{L}_{\ua}(q)\) denotes \(\mathcal{L}_{\ua}(q, F)\).

\smallskip
(\(\Ra\)). We assume that \(\exists q \in Q \colon \subs{x}{y} \in \mathcal{L}_{\ua}(q)\), i.e., \(\exists q \in Q, \exists t, s \in \cT_Q \colon \subs{y}{q} \to^*_{\TA} t\) and \(\subs{x}{ t(\varepsilon)} \to^*_{\TA} s, s(\varepsilon) \in F\).
By setting \(q_1 = q\), we have that \(\subs{y}{q_1} \to^*_{\TA} t\), where \(t(\varepsilon) \in \pre_x(F)\), since \(\exists z \in \cC_{\Sigma}\colon z \simp^F x\) s.t. \(z \in \mathcal{L}_{\ua}(t(\varepsilon))\) by setting \(z = x\).
Thus, \(\exists q_1 \in Q\colon y \in \mathcal{L}_{\ua}(q_1, \pre_x(F))\).

\medskip
(\(\La\)). Assume now that \(\exists q_1 \in Q \colon y \in \mathcal{L}_{\ua}(q_1,\pre_x(F))\), i.e., \(\exists t \in \cT_{Q} \colon \subs{y}{q_1} \to^*_{\TA} t, t(\varepsilon) \in \pre_x(F)\).
We will show that \(\exists q \in Q \colon \subs{y}{q} \to^*_{\TA} t'\) and \(\subs{x}{t'(\varepsilon)} \to^*_{\TA} s', s'(\varepsilon) \in F\).

By hypothesis, \(t(\varepsilon) \in \pre_x(F)\).
In the case that \(x \in \mathcal{L}_{\ua}(t(\varepsilon))\), then indeed, \(\exists q \in Q \colon \subs{y}{q} \to^*_{\TA} t'\) and \(\subs{x}{t'(\varepsilon)} \to^*_{\TA} s', s'(\varepsilon) \in F\), by setting \(q = q_1\) and \( t'(\varepsilon) =  t(\varepsilon)\).
Now, let us assume that \(t(\varepsilon) \in \pre_x(F)\) but \(x \notin \mathcal{L}_{\ua}(t(\varepsilon))\).
By definition of \(t(\varepsilon) \in \pre_x(F)\), we have that \(\exists \tilde{x} \in \cC_{\Sigma}\colon \tilde{x} \simp^F x\) s.t. \(\tilde{x} \in \mathcal{L}_{\ua}(t(\varepsilon))\), where \(\tilde{x} \neq x\).
Since \(\TA\) has no unreachable states, \(\subs{y}{q_1} \to^*_{\TA} t\) and \(\tilde{x} \in \mathcal{L}_{\ua}(t(\varepsilon))\), we have that \(\exists \tilde{t} \in \cT_{\Sigma} \colon \subs{y}{\tilde{t}} \in \tilde{x}^{-1}L\).
Since \(\tilde{x} \simp^F x\) and \(L\) is path-closed, by Lemma~\ref{lemma:simpImpliesQuotientEqual}, we have that \(\tilde{x}^{-1}L = x^{-1}L\).
Thus, \(\subs{y}{\tilde{t}} \in x^{-1}L\).
We conclude that, since \(\TA\) has no unreachable states, \(\exists q \in Q \colon \subs{x}{y} \in \mathcal{L}_{\ua}(q)\).
\end{proof}

\begin{lemma}
\label{lemma:simp_associativity}
Let \(\cA = \tuple{Q,\Sigma,\delta,F}\) be a BTA without unreachable states and such that \(\lang{\TA}\) is path-closed.
Then, for every \(x_1, x_2,y_1,y_2 \in \cC_{\Sigma}\) with \(\height_{\cX}(x_1) = \height_{\cX}(x_2)\), \(\subs{x_1}{y_1} \in \mathcal{L}_{\ua}(q)\) and \(\subs{x_2}{y_2} \in \mathcal{L}_{\ua}(q')\) for some \(q,q' \in Q\):
\begin{equation*}
\subs{x_1}{y_1} \simp \subs{x_2}{y_2} \Lra x_1 \simp x_2 \text{ and } y_1 \simp^{\pre_{x_1}(F)} y_2 \enspace .
\end{equation*}
\end{lemma}
\begin{proof}
Recall  that \(x, y\) are \emph{root-to-pivot equivalent} w.r.t. \(S \subseteq Q\), denoted by \(x \simp^S y\), if{}f:
\begin{enumerate}
\renewcommand\labelenumi{\theenumi}
\renewcommand{\theenumi}{(\roman{enumi})}
\item \(\piv(x)=\piv(y)\) and \(\forall p,p'\in (\mathbb{N}_+)^*\) if \(p\concat p' = \piv(x) \) then \(x(p) = y(p) \), and
\item \(\exists q \in Q \colon x \in \mathcal{L}_{\ua}(q,S) \Lra \exists q' \in Q \colon y \in \mathcal{L}_{\ua}(q',S)\).
\end{enumerate}
The reader can easily check that the double implication in the statement: \(\subs{x_1}{y_1} \simp \subs{x_2}{y_2} \Lra x_1 \simp x_2 \text{ and } y_1 \simp^{\pre_{x_1}(F)} y_2 \) holds w.r.t. condition \((i)\). In particular, the left-to-right implication holds as \(\height_{\cX}(x_1) = \height_{\cX}(x_2)\).

\medskip
We now focus on proving that the double implication holds w.r.t. condition \((ii)\) in the above definition.
Recall that \(\mathcal{L}_{\ua}(q)\) denotes \(\mathcal{L}_{\ua}(q, F)\).\\
(\(\Ra\)).
By hypothesis, \(\exists q \in Q\colon \subs{x_1}{y_1} \in \mathcal{L}_{\ua}(q)\) and  \(\exists q' \in Q\colon \subs{x_2}{y_2} \in \mathcal{L}_{\ua}(q')\).
Then, clearly, \(\exists q_1 \in Q\colon x_1 \in \mathcal{L}_{\ua}(q_1)\) and \(\exists q_2 \in Q\colon x_2 \in \mathcal{L}_{\ua}(q_2)\).
Therefore, \(x_1 \simp x_2\).
Also, by hypothesis and Lemma~\ref{lemma:subs-pre}, \(\exists q'_1 \in Q\colon y_1 \in \mathcal{L}_{\ua}(q'_1, \pre_{x_1}(F))\) and \(\exists q'_2 \in Q\colon y_2 \in \mathcal{L}_{\ua}(q'_2, \pre_{x_2}(F))\).
Since \(x_1 \simp x_2\) implies that \(\pre_{x_1}(F) = \pre_{x_2}(F)\), we have that \(y_1 \simp^{\pre_{x_1}(F)} y_2\). \smallskip\\
(\(\La\)). Finally we show the right-to-left implication:
\begin{align*}
\exists q \in Q\colon \subs{x_1}{y_1} \in \mathcal{L}_{\ua}(q) &\Lra \\
\exists q_1 \in Q \colon y_1 \in \mathcal{L}_{\ua}(q_1, \pre_{x_1}(F)) &\Lra\\
\exists q_2 \in Q \colon y_2 \in \mathcal{L}_{\ua}(q_2, \pre_{x_2}(F)) &\Lra \\
\exists q' \in Q\colon \subs{x_2}{y_2} \in \mathcal{L}_{\ua}(q') \enspace.
\end{align*}
Note that the first and last double-implications hold by Lemma~\ref{lemma:subs-pre}, while the second one holds by hypothesis.
Finally, we conclude that \(\subs{x_1}{y_1} \simp \subs{x_2}{y_2}\).
\end{proof}
We now can turn back to the lemmas on the \(\post(\cdot)\) and \(\pre(\cdot)\) operators, the upward and downward languages they induce and their relationship to upward and downward quotient.

\begin{lemma}
\label{relStateLanguage}
Let \(\TA = (Q, \Sigma, \delta, F)\) be a BTA with \(L = \lang{\TA}\) and let \(t \in \cT_{\Sigma}\) and \(x \in \cC_{\Sigma}\).
Then the following hold:
\begin{enumerate}
\renewcommand\labelenumi{\theenumi}
\renewcommand{\theenumi}{(\alph{enumi})}
\item \(\bigcup_{q \in \post_{t}(\initials(\TA))} \mathcal{L}_{\ua}(q) = L\,t^{-1}\), and \label{lemma:rel-state-language-up}
\item If \(\TA\) has no unreachable states and \(L\) is path-closed, then \(\bigcup_{q \in \pre_{x}(F)} \mathcal{L}_{\da}(q) = x^{-1}L \enspace\). \label{lemma:rel-state-language-dw}
\end{enumerate}
\end{lemma}
\begin{proof}
  We start by recalling that  \(\initials(\TA)=\{q \in Q \mid \exists a \in \Sigma_0 \colon q \in \delta(a[\,])\}\) as defined in Section~\ref{sec:bta}.

\begin{enumerate}
\renewcommand\labelenumi{\theenumi}
\renewcommand{\theenumi}{(\alph{enumi})}
\item  \(\bigcup_{q \in \post_{t}(\initials(\TA))} \mathcal{L}_{\ua}(q) = Lt^{-1}\).

Recall that \(\mathcal{L}_{\ua}(q)\) simply denotes \(\mathcal{L}_{\ua}(q,F)\).
To simplify further the notation, denote \(\initials(\TA)\) as \(I\).
\begin{align*}
L\,t^{-1} & =  \\
\{x \in \cC_\Sigma \mid \subs{x}{t} \in L\} & =  \\
\{x \in \cC_\Sigma \mid \exists t' \!\in\! \cT_Q \colon  \subs{x}{t} \to^*_{{\TA}} t',\, t'(\varepsilon) \in F,\, \leaf(t') \subseteq I\} & =^{\dagger} \\
\{x \in \cC_\Sigma \mid \exists t' \!\in\! \cT_Q \colon  \subs{x}{q} \to^*_{{\TA}} t',\, t'(\varepsilon) \in F,\, q \in \post_t (I)\}& = \\
\bigcup_{q \in \post_{t}(I)} \mathcal{L}_{\ua}(q) \enspace .
 \end{align*}
Note that, from the statement to the left of the equality \(\dagger\) we have that  \(t \in \mathcal{L}_{\da}(q, I)\) for some \(q \in Q\), which is equivalent to the fact \(q \in \post_t(I)\).

\item If \(\TA\) has no unreachable states and \(L\) is path-closed, then \(\bigcup_{q \in \pre_{x}(F)} \mathcal{L}_{\da}(q) = x^{-1}L \enspace\).\smallskip

Recall that \(\mathcal{L}_{\da}(q)\) denotes \(\mathcal{L}_{\da}(q,\initials(\TA))\), \(\mathcal{L}_{\ua}(q)\) denotes \(\mathcal{L}_{\ua}(q, F)\) and \(\simp\)  denotes  \(\simp^F\).
We have that:
\begin{align*}
x^{-1}L & = \\
\{t \in \cT_\Sigma \mid \subs{x}{t} \in L\} & =^{\dagger} \\
\{t \in \cT_\Sigma \mid \exists y \in \cC_\Sigma \colon x \simp y, \subs{y}{t} \in L\} & = \\
\{t \in \cT_{\Sigma} \mid \exists t' \in \cT_Q  \colon
t \to_{{\TA}}^* t', t'(\varepsilon) = q, x\in P_{\simp}(\mathcal{L}_{\ua}(q))\} &= \\
\{t \in \cT_{\Sigma} \mid  \exists t' \in \cT_Q \colon t \to_{{\TA}}^* t', t'(\varepsilon) = q,  q \in \pre_x(F)\} & = \\
\bigcup_{q \in \pre_{x}(F)} \mathcal{L}_{\da}(q) \enspace ,
\end{align*}
where equality \(\dagger\) holds by Lemma~\ref{lemma:simpImpliesQuotientEqual}.
Namely, since \(L\) is path-closed and \(\TA\) has no unreachable states, then for every \(x,y \in \cC_{\Sigma}\) s.t. \(x \simp y\) we have that \(x^{-1}L = y^{-1}L\).
\end{enumerate}

\vspace*{-6mm}
\end{proof}

\begin{lemma}
\label{prepost}
Let \({\TA} = \tuple{Q,\Sigma,\delta,F}\) be a BTA. %
Let \(S \subseteq Q\), \(f \in \Sigma\), \(t_1,\ldots,t_{\rank{f}} \in \cT_\Sigma\) and let \(x,y \in \cC_\Sigma\).
Then the following hold:
\begin{enumerate}
\renewcommand\labelenumi{\theenumi}
\renewcommand{\theenumi}{(\alph{enumi})}
\itemsep=0.9pt
\item \(\post_{f[t_1,\ldots,t_{\rank{f}}]}(S) =
  \delta( \{ f[\post_{t_1}(S),\ldots,\post_{t_{\rank{f}}}(S)] \} )\).\label{lemma:prepost:prop:tree}
\item If \(\TA\) has no unreachable states and \(\lang{\TA}\) is path-closed then \(\pre_{\subs{x}{y}}(F)  = \pre_y(\pre_x(F))\).\label{lemma:prepost:prop:ctx}
\end{enumerate}
\end{lemma}

\begin{proof}\hfill

\vspace*{-6mm}
\begin{enumerate}
\renewcommand\labelenumi{\theenumi}
\renewcommand{\theenumi}{(\alph{enumi})}
\item \(\post_{f[t_1,\ldots,t_{\rank{f}}]}(S) =
  \delta( \{ f[\post_{t_1}(S),\ldots,\post_{t_{\rank{f}}}(S)] \} )\). \vspace*{-2mm}
  \begin{align*}
\post_{f[t_1,\ldots, t_{\rank{f}}]}(S) & =  \\
\{q  \mid f[t_1,\ldots, t_{\rank{f}}] \in \mathcal{L}_{\da}(q,S)\} & = \\
\{q \mid \forall i \in \range{\rank{f}}, \exists t_i' \in \cT_{Q} \colon  t_i \to^*_{\cA} t'_i,  t_i'(\varepsilon) =  q_i,&\\
 \leaf(t'_i) \subseteq S, q \in \delta(f[q_1,\ldots, q_{\rank{f}}])\}&=\\
\{q \mid \forall i \in \range{\rank{f}}, \exists q_i \in \post_{t_i}(S) \colon q\in \delta(f[q_1,\ldots, q_{\rank{f}}])\} &= \\
\delta( \{ f[\post_{t_1}(S),\ldots,\post_{t_{\rank{f}}}(S)] \}  ) \enspace .
\end{align*}
Note that the first three equalities hold by definition of \(\post(\cdot)\), by definition of downward language of \(q\) w.r.t. \(S\) and by definition of the transition function \(\delta\) respectively.
Finally, the last equality holds by definition of \(\delta\) extended to sets.

\item If \(\TA\) has no unreachable states and \(\lang{\TA}\) is path-closed then \(\pre_{\subs{x}{y}}(F)  = \pre_y(\pre_x(F))\).

First, note that \(\mathcal{L}_{\ua}(q)\) denotes \(\mathcal{L}_{\ua}(q,F)\), and \(\simp\) denotes \(\simp^F\).
Also recall that, by definition of root-to-pivot equivalence, \(x \simp y\) if{}f:
\begin{enumerate}
\renewcommand\labelenumi{\theenumi}
\renewcommand{\theenumi}{(\roman{enumi})}
\item \(\piv(x)=\piv(y)\) and \(\forall p,p'\in (\mathbb{N}_+)^*\) if \(p\concat p' = \piv(x) \) then \(x(p) = y(p) \), and
\item \(\exists q \in Q \colon x \in \mathcal{L}_{\ua}(q) \Lra \exists q' \in Q \colon y \in \mathcal{L}_{\ua}(q')\).
\end{enumerate}
(\(\Ra\)). We first prove that \(\forall q \in Q \colon q \in \pre_{\subs{x}{y}}(F) \Ra q \in \pre_x(\pre_y(F))\).
By definition we have that: \vspace*{-2mm}
\begin{align}
q \in \pre_{\subs{x}{y}}(F) &\Lra \quad \notag\\
\subs{x}{y} \in P_{\simp^F}(\mathcal{L}_{\ua}(q)) &\Lra \quad \notag\\
\exists z \in \cC_\Sigma \text{ with } z \simp \subs{x}{y} \colon z \in \mathcal{L}_{\ua}(q) & \enspace .\quad \label{eq:proof_pre_1}
\end{align}

Note that for every \(z \in \cC_{\Sigma}\) and for every \(i \in \range{\height_{\cX}(z)}\), there exists \(\tilde{x},\tilde{y} \in \cC_{\Sigma}\) such that \(z = \subs{\tilde{x}}{\tilde{y}}\) and \(\height_{\cX}(\tilde{x}) = i\).
Using this, we rewrite statement~\eqref{eq:proof_pre_1} as follows: \vspace*{-2mm}
\begin{align*}
q \in \pre_{\subs{x}{y}}(F) &\Lra\\
\exists \tilde{x},\tilde{y} \in \cC_\Sigma, \subs{\tilde{x}}{\tilde{y}} \simp\subs{x}{y}, \height_{\cX}(\tilde{x}) = \height_{\cX}(x) \colon \subs{\tilde{x}}{\tilde{y}}\in \mathcal{L}_{\ua}(q) & \enspace .
\end{align*}
Using the definition of \(\mathcal{L}_{\ua}(q)\) and \(\to^*_{\cA}\):
\begin{align}
q \in \pre_{\subs{x}{y}}(F) &\Lra \notag\\
\exists \tilde{x},\tilde{y} \!\in \cC_\Sigma, \subs{\tilde{x}}{\tilde{y}} \simp \subs{x}{y}, \height_{\cX}(\tilde{x}) = \height_{\cX}(x),\exists t,r \!\in \cT_Q &\colon \notag\\
\subs{\tilde{y}}{q}\to^*_{\cA} r , \subs{\tilde{x}}{r(\varepsilon)} \to^*_{\cA} t , t(\varepsilon) \in F &\enspace . \label{eq:equivalence_together}
\end{align}

Note that we are under the conditions of Lemma~\ref{lemma:simp_associativity}.
Specifically, \(\TA\) has no unreachable states and \(\lang{\TA}\) is path-closed, by hypothesis.
Furthermore, \(\height_{\cX}(x) = \height_{\cX}(\tilde{x})\), \(\subs{\tilde{x}}{\tilde{y}} \in \mathcal{L}_{\ua}(q)\) and, since \(\subs{x}{y} \simp \subs{\tilde{x}}{\tilde{y}}\), we have that \(\exists q' \in Q\colon \subs{x}{y} \in \mathcal{L}_{\ua}(q')\).
Therefore, by Lemma~\ref{lemma:simp_associativity}, we have that \(\subs{\tilde{x}}{\tilde{y}} \simp \subs{x}{y} \Lra\) \mbox{\(\tilde{x} \simp x\)}, and \(\tilde{y} \simp^{\pre_x(F)} y\).
Thus: \vspace*{-2mm}
\begin{align}
q \in \pre_{\subs{x}{y}}(F) &\Ra \notag \\
\exists \tilde{x},\tilde{y} \in \cC_\Sigma, (\tilde{x} \simp x), (\tilde{y}\simp^{\pre_x(F)} y), \exists t,r \in \cT_Q   &\colon \notag\\
\subs{\tilde{y}}{q} \to^*_{\cA} r , \subs{\tilde{x}}{r(\varepsilon)} \to^*_{\cA} t , t(\varepsilon) \in F &\enspace \label{eq:hypotheses} .
\end{align}
Notice that \(r(\varepsilon) \in \pre_{x}(F)\) since \( \tilde{x} \simp x\) and \(\tilde{x} \in \mathcal{L}_{\ua}(r(\varepsilon))\).
Then:
\begin{align*}
\exists \tilde{x},\tilde{y} \in \cC_\Sigma, (\tilde{x} \simp x), (\tilde{y}\simp^{\pre_x(F)} y), \exists t,r \in \cT_Q   &\colon \notag\\
\subs{\tilde{y}}{q} \to^*_{\cA} r , \subs{\tilde{x}}{r(\varepsilon)} \to^*_{\cA} t , t(\varepsilon) \in F \notag &\Lra \\
\exists  \tilde{y} \in \cC_\Sigma, \tilde{y} \simp^{\pre_x(F)} y,\exists r \in \cT_Q &\colon\\
\subs{\tilde{y}}{q} \to^*_{\cA} r,  r(\varepsilon) \in \pre_{x}(F) &\Lra \notag\\
\exists  \tilde{y} \in \cC_\Sigma, \tilde{y} \simp^{\pre_x(F)} y\colon \tilde{y} \in \mathcal{L}_{\ua}(q, \pre_x(F)) &\Lra \notag\\
y\in P_{\simp^{\pre_x(F)}}(\mathcal{L}_{\ua}(q, \pre_x(F))) &\Lra \notag\\
q \in \pre_y(\pre_x(F)) \notag \enspace .
\end{align*}

(\(\La\)). Now we show that \(\forall q \in Q\colon q \in \pre_y(\pre_x(F))\Ra q \in \pre_{\subs{x}{y}}(F)\).
Observe that, from the hypothesis \(q \in \pre_y(\pre_x(F))\), we can follow the same reasoning as the one above bottom-up up to Equation~\eqref{eq:hypotheses}, as every statement is chained with a \(\Lra\) symbol.
Thus, we have that: \vspace*{-2mm}
\begin{align}
q \in \pre_{x}(\pre_y(F)) &\Lra \notag \\
\exists \tilde{x},\tilde{y} \in \cC_\Sigma, (\tilde{x} \simp x), (\tilde{y}\simp^{\pre_x(F)} y), \exists t,r \in \cT_Q   &\colon \notag\\
\subs{\tilde{y}}{q} \to^*_{\cA} r , \subs{\tilde{x}}{r(\varepsilon)} \to^*_{\cA} t , t(\varepsilon) \in F &\enspace \label{eq:hypotheses2} .
\end{align}

Now we check that we are under the hypotheses of Lemma~\ref{lemma:simp_associativity}.
First, we have that \(\height_{\cX}(x) = \height_{\cX}(\tilde{x})\) since \(\tilde{x} \simp x\).
Also, \(\exists q \in Q\colon \subs{\tilde{x}}{\tilde{y}} \in \mathcal{L}_{\ua}(q)\) (see Equation~\eqref{eq:hypotheses2}).
Finally, using that \(\tilde{x}\simp x\) (which implies that \(\exists q_1 \in Q\colon x \in \mathcal{L}_{\ua}(q_1)\)) and \(\tilde{y} \simp^{\pre_x(F)} y\) (which implies that \(\exists q_2 \in Q\colon y \in \mathcal{L}_{\ua}(q_2, \pre_x(F))\)) together with Lemma~\ref{lemma:subs-pre}, we conclude that \(\exists q' \in Q\colon \subs{x}{y} \in \mathcal{L}_{\ua}(q')\).
Thus, we are indeed under the conditions of Lemma~\ref{lemma:simp_associativity} and we have that:
\begin{align}
q \in \pre_{x}(\pre_y(F)) &\Ra \notag \\
\exists \tilde{x},\tilde{y} \!\in \cC_\Sigma, \subs{\tilde{x}}{\tilde{y}} \simp \subs{x}{y}, \height_{\cX}(\tilde{x}) = \height_{\cX}(x), \exists t,r \!\in \cT_Q &\colon \notag\\
\subs{\tilde{y}}{q}\to^*_{\cA} r , \subs{\tilde{x}}{r(\varepsilon)} \to^*_{\cA} t , t(\varepsilon) \in F &\enspace . \label{eq:hypotheses3}
\end{align}
Finally, from Equation~\eqref{eq:hypotheses3}, we can follow the same reasoning as the one given above from Equation~\eqref{eq:equivalence_together} bottom-up up to the statement \(q \in \pre_{\subs{x}{y}}(F)\), as every statement is chained with a \(\Lra\) symbol.

We conclude that \(\forall q \in Q, q \in \pre_{\subs{x}{y}}(F) \Lra q \in \pre_y(\pre_x(F))\).
\end{enumerate}
\end{proof}

We are now in position to introduce BTA-based equivalences.
\begin{definition}[BTA-based equivalences]
\label{def:automataEquivalences}
Let \({\TA} = \tuple{Q,\Sigma,\delta,F}\) be a BTA and let \(t, r \in \cT_{\Sigma}\) and \(x, y \in \cC_{\Sigma}\).
The \emph{upward and downward BTA-based equivalences} are respectively given by: %
\begin{align*}
  t \ruA r &\udiff \post_{t}(\initials({\TA})) = \post_{r}(\initials({\TA}))\\
  x \rdA y &\udiff \pre_{x}(F) = \pre_{y}(F)\enspace.
\end{align*}
\end{definition}

Note that these equivalences have finite index since tree automata have finitely many states.
Next, we show that \(\ruA\) and \(\rdA\) enable Lemmas~\ref{HBPreservesL} and~\ref{downwardIsDet} respectively.

\begin{lemma}
\label{automataCongruences}
Let \({\TA}\) be a BTA with \(L = \lang{{\TA}}\).
Then,
\begin{enumerate}
\renewcommand\labelenumi{\theenumi}
\renewcommand{\theenumi}{(\alph{enumi})}
\item \(\ruA\) is an upward congruence and \(\mathord{\ruA} \subseteq \mathord{\ru_L}\).\label{lemma:automataCongruences:upAB}
\item If \(\TA\) has no unreachable states and \(L\) is path-closed then \(\rdA\) is a strongly downward congruence w.r.t. \(L\) and \( \mathord{\rdA} \subseteq  \mathord{\rd_L}\). \label{lemma:automataCongruences:downAT}
\end{enumerate}
\end{lemma}
\begin{proof} Let \({\TA} = \tuple{Q, \Sigma, \delta, F}\) and, to simplify the notation, let \(I\) be the set \(\initials({\TA})\).
\begin{enumerate}
\renewcommand\labelenumi{\theenumi}
\renewcommand{\theenumi}{(\alph{enumi})}
\item \(\ruA\) is an upward congruence and \( \mathord{\ruA}\subseteq \mathord{\ru_L}\).

We first prove that \(\ruA\) is an upward congruence.
Let \(f \in \Sigma\) and \(t_i,r_i \in \cT_\Sigma\) for \(i \in \range{\rank{f}}\) be such that \(t_i \ruA r_i\).
By Definition~\ref{def:automataEquivalences}, we have that \(\post_{t_i}(I) = \post_{r_i}(I)\), for each \( i \in \range{\rank{f}}\).
Therefore, by Lemma~\ref{prepost}\ref{lemma:prepost:prop:tree}, we conclude that \(\post_{f[t_1,\ldots,t_{\rank{f}}]}(I) = \post_{f[r_1,\ldots,r_{\rank{f}}]}(I)\), i.e., \(f[t_1, \ldots, t_{\rank{f}}] \ruA f[r_1, \ldots, r_{\rank{f}}]\).

Next, we show that \(\forall t,r \in \cT_\Sigma \colon  t \ruA r \Ra L\,t^{-1} = L\,r^{-1}\).
Since \(t \ruA r\), i.e., \(\post_t(I) = \post_r(I)\), we have that \(\bigcup_{q \in \post_{t}(I)} \mathcal{L}_{\ua}(q) = \bigcup_{q \in \post_{r}(I)} \mathcal{L}_{\ua}(q)\).
By Lemma~\ref{relStateLanguage}\ref{lemma:rel-state-language-up}, we conclude that \(Lt^{-1} = Lr^{-1}\).

\item If \(\TA\) has no unreachable states and \(L\) is path-closed then \(\rdA\) is a strongly downward congruence w.r.t. \(L\) and \( \mathord{\rdA} \subseteq  \mathord{\rd_L}\).

First, we show that  \(\forall x,y \in \cC_{\Sigma} \colon  x \rdA y \Ra x^{-1}L = y^{-1}L \).
Since \(x \rdA y\), i.e., \(\pre_x(F) = \pre_y(F)\), we have that:
\[\bigcup_{q \in \pre_{x}(F)} \mathcal{L}_{\da}(q) = \bigcup_{q \in \pre_{y}(F)} \mathcal{L}_{\da}(q) \enspace .\]
Since \(\TA\) has no unreachable states and \(L\) is path-closed, by Lemma~\ref{relStateLanguage}\ref{lemma:rel-state-language-dw}, we conclude that \(x^{-1}L = y^{-1}L\).

Now we prove that \(\rdA\) is a downward congruence.
Let  \(x,y \in \cC_{\Sigma} \colon x\rdA y\), i.e., \(\pre_x(F) = \pre_y(F)\).
We will show that for every \(c \in \cC_{\Sigma}\), \( \subs{x}{c}\rdA \subs{y}{c}\), i.e., \(\pre_{\subs{x}{c}}(F) = \pre_{\subs{y}{c}}(F)\).
Relying on Lemma~\ref{prepost}\ref{lemma:prepost:prop:ctx} we have that, \(\pre_{\subs{x}{c}}(F) = \pre_c(\pre_x(F)) = \pre_c(\pre_y(F)) =\pre_{\subs{y}{c}}(F)\).
Thus, we conclude that \( \subs{x}{c}\rdA \subs{y}{c}\).

Finally, we show that \(\rdA\) is strongly downward w.r.t. \(L\).
Let \(x,y \in \cC_\Sigma \colon x \rdA y, t \in x^{-1}L, r \in y^{-1}L\) and \(t(\varepsilon) = r(\varepsilon)\), with \(t,r \in \cT_{\Sigma}\).
We will prove that \(\subs{x}{\subs{t}{\cX}_i} \rdA \subs{x}{\subs{r}{\cX}_i}\) and \(\subs{y}{\subs{t}{\cX}_i} \rdA \subs{y}{\subs{r}{\cX}_i}\), for every \(i \in \range{\rank{f}}\).
Note that if the previous holds then, since \(\rdA\) is a downward congruence, i.e., \(\subs{x}{\subs{t}{\cX}_i} \rdA \subs{y}{\subs{t}{\cX}_i}\) and using transitivity of \(\rdA\), we can conclude that \(\subs{x}{\subs{t}{\cX}_i} \rdA \subs{y}{\subs{r}{\cX}_i}\), for every \(i \in \range{\rank{f}}\).

For the clarity of the argument, let us assume w.l.o.g. that \(t = f[t_1, t_2]\) and \(r = f[r_1, r_2]\).
Thus, we have to prove that \(\subs{x}{f[\cX, t_2]} \rdA \subs{x}{f[\cX, r_2]}\) and \(\subs{x}{f[t_1, \cX]} \rdA \subs{x}{f[r_1, \cX]}\), as well as \(\subs{y}{f[\cX, t_2]} \rdA \subs{y}{f[\cX, r_2]}\) and \(\subs{y}{f[t_1, \cX]} \rdA \subs{y}{f[r_1, \cX]}\).
Note that, by definition,  \(\subs{x}{f[\cX, t_2]} \rdA \subs{x}{f[\cX, r_2]}\) is equivalent to \(\pre_{\subs{x}{f[\cX, t_2]}}(F) = \pre_{\subs{x}{f[\cX, r_2]}}(F)\).

\medskip
Assume \(q \in \pre_{\subs{x}{f[\cX, t_2]}}(F)\), i.e., \(\subs{x}{f[\cX, t_2]} \in P_{\simp}(\mathcal{L}_{\ua}(q))\).
We will show that \\  \(q \in \pre_{\subs{x}{f[\cX, r_2]}}(F)\).

As we have shown previously, \(x \rdA y \Ra x^{-1} L = y^{-1}L\).
Therefore, since \(t \in x^{-1}L\) and \(r \in y^{-1}L\), then \(r \in x^{-1}L\) and \(t \in y^{-1}L\).
Hence, as \(\TA\) has no unreachable states and \(r \in x^{-1}L\), we have that \(\exists {q'} \in Q, \exists t' \in \cT_{Q} \colon \subs{x}{f[q', r_2]} \to^*_\TA t', t'(\varepsilon) \in F\), i.e., \(\exists {q'} \in Q \colon \subs{x}{f[\cX, r_2]} \in \mathcal{L}_{\ua}(q')\).
Note that, by Definition~\ref{def:root-pivot} (root-to-pivot equivalence w.r.t. \(F\)), \( \subs{x}{f[\cX, t_2]} \simp \subs{x}{f[\cX, r_2]}\).
Since, by hypothesis,  \(\subs{x}{f[\cX, t_2]} \in P_{\simp}(\mathcal{L}_{\ua}(q))\), we have that \( \subs{x}{f[\cX, r_2]}\!\in \! P_{\simp}(\mathcal{L}_{\ua}(q))\).$\,$By definition of \( \pre_{\subs{x}{f[\cX\!, r_2]}}(F)\),$\,$we conclude that \(q\! \in\! \pre_{\subs{x}{f[\cX,r_2]}}(F)\) and thus, \(\pre_{\subs{x}{f[\cX, t_2]}}\!(F) \subseteq \pre_{\subs{x}{f[\cX, r_2]}}\!(F)\).

The proof of \(\pre_{\subs{x}{f[\cX, r_2]}}(F) \subseteq \pre_{\subs{x}{f[\cX, t_2]}}(F)\) goes similarly.

Thus, we have that \(\subs{x}{f[\cX, t_2]} \rdA \subs{x}{f[\cX, r_2]}\).
On the other hand, the proof of \(\subs{x}{f[t_1, \cX]} \rdA \subs{x}{f[r_1, \cX]}\) is symmetric.

Finally, note that the remainder of the proof, i.e.,  \(\subs{y}{f[\cX, t_2]} \rdA \subs{y}{f[\cX, r_2]}\) and \(\subs{y}{f[t_1, \cX]} \rdA \subs{y}{f[r_1, \cX]}\) is symmetric, replacing \(x\) by \(y\).
\end{enumerate}

\vspace*{-6mm}
\end{proof}

The next example illustrates the need of root-to-pivot equivalence to define \(\pre(\cdot)\) in order to ensure that the automata-based downward congruence is indeed strongly downward.
\begin{example}
\label{example:new-pre}
Consider the BTA \(\TA = \tuple{Q,\Sigma_0\cup \Sigma_2,\delta,F}\) with \(Q = \{q_a, q_b, q_c, q_f, q_{\bullet}\}\), \(\Sigma_0 = \{a,b,c\}\), \(\Sigma_2 = \{f\}, F = \{q_f\}\), and s.t. \(\delta\) is defined as follows: \(\{q_f\} = \delta(f[q_a, q_a]) = \delta(f[q_b, q_b]) = \delta(f[q_c, q_c]) = \delta(f[q_{\bullet}, q_{\bullet}])\); \(\{q_a, q_{\bullet}\} = \delta(a[\,])\); \(\{q_b, q_{\bullet}\} = \delta(b[\,])\), and \(\{q_c, q_{\bullet}\} = \delta(c[\,])\).
We will construct \(\cH^{\mathsf{d}}(\rd, \lang{\TA}) = \tuple{Q',\Sigma_0\cup \Sigma_2,\delta',F'}\) where \(\rd\) is a downward congruence defined as follows.
For each \(x,y \in \cC_{\Sigma}\):
\[x \rd y \udiff \wpre^{\TA}_x(F) = \wpre^{\TA}_y(F) \enspace,\]
where \(\wpre^{\TA}_x(S)  \udiffg  \{q \in Q \mid x \in \mathcal{L}_{\ua}(q,S)\}\), with \(S \subseteq Q\).
Note that the definition of \(\wpre^{\TA}_x(S)\) is similar to that of \(\text{pre}^{\TA}_x(S)\) (Definition~\ref{def:postpre}) except that we drop root-to-pivot equivalence.

\medskip
It is easy to check that \(\rd\) defines the following blocks:
\begin{itemize}
\itemsep=0.9pt
\item \(P_{\rd}(X_1)\), where \(X_1 = \{f[\cX, a], f[a, \cX]\}\), since \(\wpre_{x}(F) = \{q_a, q_{\bullet}\}\), \(\forall x \in X_1\);
\item \(P_{\rd}(X_2)\),  where \(X_2 = \{f[\cX, b], f[b, \cX]\}\), since \(\wpre_{x}(F) = \{q_b, q_{\bullet}\}\), \(\forall x \in X_2\);
\item \(P_{\rd}(X_3)\), where \(X_3 = \{f[\cX, c], f[c, \cX]\}\), since \(\wpre_{x}(F) = \{q_c, q_{\bullet}\}\), \(\forall x \in X_3\);
\item \(P_{\rd}(\cX)\) since \(\wpre_{x}(F) = \{q_f\}\); and
\item \(P_{\rd}(X_4)\), where \(X_4 = \cC_{\Sigma_0 \cup \Sigma_2} \setminus (X_1 \cup X_2 \cup X_3 \cup \{\cX\})\) since \(\wpre_{x}(F) = \varnothing\), \(\forall x \in X_4\).
\end{itemize}
By Definition~\ref{def:down_TA_construction}, \(Q' = \{P_{\rd}(X_1), P_{\rd}(X_2), P_{\rd}(X_3), P_{\rd}(\cX)\}\), \(F' = P_{\rd}(\cX)\) and \(\delta'\) is defined as follows: \{\(P_{\rd}(X_1)\} = \delta'(a[\,])\); \{\(P_{\rd}(X_2)\} = \delta'(b[\,])\); \{\(P_{\rd}(X_3)\} = \delta'(c[\,])\); and \{\(P_{\rd}(\cX)\} = \delta'(f[P_{\rd}(X_i), P_{\rd}(X_j)])\), \(\forall i,j\!\in\! \{1,2,3\}\).
For instance, note that \(P_{\rd}(\cX)\!\in\! \delta'(f[P_{\rd}(X_3), P_{\rd}(X_1)])\) since there exist \(t_1 = a[\,]\), \(t_2 = c[\,]\) such that  \(\subs{\cX}{\subs{f[a,c]}{\cX}_1} \in P_{\rd}(X_3)\) and \(\subs{\cX}{\subs{f[a,c]}{\cX}_2} \in P_{\rd}(X_1)\).

However, observe that \(\cH^{\mathsf{d}}(\rd, \lang{\TA})\) is not co-deterministic since, for instance, \(P_{\rd}(\cX)\) \({} \in \delta'(f[P_{\rd}(X_3), P_{\rd}(X_1)])\) and \(P_{\rd}(\cX) \in \delta'(f[P_{\rd}(X_1), P_{\rd}(X_3)])\) where \(P_{\rd}(X_1) \neq P_{\rd}(X_3)\).
This is because \(\rd\) is not a strongly downward congruence w.r.t. \(\lang{\TA}\).
In fact, in Definition~\ref{def:strong_downward_congruence}, let \(x = y = \cX\) and \(t = f[a,c]\), \(r = f[a,a]\) where \(\subs{x}{t}, \subs{y}{r} \in \lang{\TA}\).
Then, note that (trivially) \(x \rd y\), but \(\subs{x}{\subs{t}{\cX}_1} = f[\cX, a] \not\rd f[\cX, c] = \subs{x}{\subs{r}{\cX}_1}\).

\medskip
Now we will build \(\cH^{\mathsf{d}}(\rd_\TA, \lang{\TA}) =  \tuple{Q',\Sigma_0\cup \Sigma_2,\delta',F'}\) where \(\rd_\TA\) is the automata-based downward congruence (Definition~\ref{def:automataEquivalences}).
Note that \(\rd_{\TA}\) defines the following classes of equivalence:
\begin{itemize}
\itemsep=0.9pt
\item \(P_{\rd_{\TA}}(X)\), where \mbox{\(X = \{f[\cX, a], f[a, \cX], f[\cX, b],f[b, \cX],\)} \(f[\cX, c], f[c, \cX]\}\), as \(\pre_{x}(F) = \{q_a, q_b$, $q_c, q_{\bullet}\}\), \(\forall x \in X\);
\item \(P_{\rd_{\TA}}(\cX)\) as \(\pre_{x}(F) = \{q_f\}\); and
\item \(P_{\rd_{\TA}}(X_4)\), where \(X_4 = \cC_{\Sigma_0 \cup \Sigma_2} \setminus (X \cup \{\cX\})\).
\end{itemize}
Notice that \(X = X_1 \cup X_2 \cup X_3\).
According to Definition~\ref{def:down_TA_construction}, \(Q' = \{P_{\rd_{\TA}}(X), P_{\rd_{\TA}}(\cX)\}, F' = \{P_{\rd_{\TA}}(\cX)\}\) and \(\delta'\) is defined as follows: \{\(P_{\rd_{\TA}}(X)\} = \delta'(a[\,]) = \delta'(b[\,]) = \delta'(c[\,])\) and \{\(P_{\rd_{\TA}}(\cX)\} = \delta'(f[P_{\rd_{\TA}}(X), P_{\rd_{\TA}}(X)])\).

\medskip
Note that \(\cH^{\mathsf{d}}(\rd_\TA, \lang{\TA})\) is co-deterministic as \(\rd_\TA\) is a strongly downward congruence w.r.t. \(\lang{\TA}\).
The reader may check that \(\cH^{\mathsf{d}}(\rd_\TA, \lang{\TA})\) is isomorphic to the co-DBTA \(\TA^{\mathsf{cD}}\) that results from applying the co-determinization operation to \(\TA\) and removing empty states.
\eox
\end{example}

In the light of the Lemma~\ref{automataCongruences}, the upward BTA-based  congruences are indeed finer than (or equal to) the language-based ones, i.e., \(\mathord{\ruA} \subseteq \mathord{\ru_L}\).
So are their downward counterparts if \(\TA\) has no unreachable states and \(L(\TA)\) is path-closed.
The following theorem gives a sufficient condition for the language-based and the BTA-based congruences to coincide.

\begin{theorem}
\label{automataEqualNerode}
Let \({\TA}\) be a BTA such that \(L = \lang{\TA}\) is a path-closed language.
Then,
\begin{enumerate}
\renewcommand\labelenumi{\theenumi}
\renewcommand{\theenumi}{(\alph{enumi})}
\itemsep=0.9pt
\item If \(\TA\) is a co-DBTA with no empty states, then \(\mathord{\ruA} = \mathord{\ru_L}\).\label{theorem:automata=nerodeUP}
\item If \(\TA\) is a DBTA with no unreachable states then \mbox{\(\mathord{\rdA} = \mathord{\rd_L}\)}.\label{theorem:automata=nerodeDW}
\end{enumerate}
\end{theorem}

\begin{proof}

\vspace*{-6mm}
\begin{enumerate}
\renewcommand\labelenumi{\theenumi}
\renewcommand{\theenumi}{(\alph{enumi})}
\item If \(\TA\) is a co-DBTA with no empty states, then \( \mathord{\ruA}=\mathord{\ru_L}\).

By Lemma~\ref{automataCongruences}\ref{lemma:automataCongruences:upAB}, we have that \(\mathord{\ruA} \subseteq \mathord{\ru_L}\).

Next, we show that if \({\TA}\) is a co-DBTA without empty states, then \(\mathord{\ru_L} \subseteq \mathord{\ruA}\).
To simplify the notation, let \(I\) denote the set \(\initials(\TA)\).

First, note that since \({\TA}\) has no empty states, \(Lt^{-1} = Lr^{-1} = \varnothing\) if{}f \(\post_t(I) = \post_r(I) = \varnothing\).
Now we consider the case \(Lt^{-1} = Lr^{-1} \neq \varnothing\) and proceed by contradiction.
Assume that \(L\,t^{-1} = L\,r^{-1}\) and \(\post_t(I) \neq \post_{r}(I)\).
Recall that:
\begin{align*}
L\,t^{-1} & = \\
\{x \in \cC_\Sigma \mid \subs{x}{t} \in L\} & = \\
\{x \in \cC_\Sigma \mid \exists t' \!\in\! \cT_Q \colon  \subs{x}{t} \to^*_{{\TA}} t',\, t'(\varepsilon) \in F,\, \leaf(t') \subseteq I\} & =^{\dagger} \\
\{x \in \cC_\Sigma \mid \exists t' \!\in\! \cT_Q \colon  \subs{x}{q} \to^*_{{\TA}} t',\, t'(\varepsilon) \in F,\, q \in \post_t (I)\} &\enspace .
\end{align*}
Note that, from the statement to the left of the equality \(\dagger\) we have that  \(t \in \mathcal{L}_{\da}(q, I)\) for some \(q \in Q\), which is equivalent to the fact \(q \in \post_t(I)\).

By hypothesis, \(Lt^{-1} = Lr^{-1} \neq \varnothing\), and thus we can assume that there exists \(q'\in Q\) with \(q' \neq q\) such that \(q \in \post_t (I) \cap  (\post_r (I))^{\complement}\) and \(q' \in \post_r (I) \).
Then, there exists \(t' \in \cT_{\Sigma}: \subs{x}{q} \to^*_{\TA} t',\, t'(\varepsilon) \in F\), and there exists \(r' \in \cT_{\Sigma}: \subs{x}{q'} \to^*_{\TA} r',\, r'(\varepsilon) \in F\).
Since \(|F| = 1\), \(Lt^{-1} = Lr^{-1}\) and \(q \neq q'\), there exists necessarily \(q_0 \in Q\) and \(f \in \Sigma\) such that \(|\delta^{-1}(q_0) \cap \{f\}\times Q^{\rank{f}}| > 1\).
This is a contradiction since \(\TA\) is co-deterministic.
Therefore, we conclude that if \(Lt^{-1} = Lr^{-1}\) then \(\post_t(I)  =  \post_{r}(I)\).

\item If \(\TA\) is a DBTA with no unreachable states then \(\mathord{\rdA} = \mathord{\rd_L}\).

Since \(\TA\) has no unreachable states and \(L = \lang{\TA}\) is path-closed, by Lemma~\ref{automataCongruences}\ref{lemma:automataCongruences:downAT}, we have that \(\rdA\) is a strongly downward congruence w.r.t. \(L\) and \(\mathord{\rdA} \subseteq \mathord{\rd_L}\).

Now we show that, in particular, if \(\TA\) is a DBTA  then \(\mathord{\rdA} = \mathord{\rd_L}\).
Given \(x,y \in \cC_{\Sigma}\) with \(x^{-1}L = y^{-1}L\), we will show that \(\pre_x(F) = \pre_y(F)\).

By Lemma~\ref{relStateLanguage}\ref{lemma:rel-state-language-dw} and the fact that \(x^{-1}L = y^{-1}L\), we have that:
\begin{align}
x^{-1}L = \bigcup_{q \in \pre_{x}(F)} \mathcal{L}_{\da}(q) = \bigcup_{q \in \pre_{y}(F)} \mathcal{L}_{\da}(q) = y^{-1}L \enspace . \label{eq:proof_equality_downward}
\end{align}
First, note that \(\TA\) has no unreachable states, i.e., \(\forall q \in Q \colon \mathcal{L}_{\da}(q) \neq \varnothing\).
In particular, \(\forall q \in \pre_x(F) \colon \mathcal{L}_{\da}(q) \neq \varnothing\).

Second, we prove that if \(\TA\) is a DBTA then, for every pair of states \(q,q' \in Q\), \(\mathcal{L}_{\da}(q) = \mathcal{L}_{\da}(q')\) implies \(q = q'\).
The proof goes by contradiction.
Assume that \(q \neq q'\).
Then, since \(\mathcal{L}_{\da}(q) = \mathcal{L}_{\da}(q')\), for every \(t \in \mathcal{L}_{\da}(q)\), there exists \(t', t'' \in \cT_Q\) with \(t' \neq t''\) s.t. \(t \to^*_{\TA} t', t'(\varepsilon) = q\) and \(t \to^*_{\TA} t'', t''(\varepsilon) = q''\).
Therefore, for some \(t_0 \in \cT_{\Sigma \cup Q}\) satisfying that \(t \to^*_{\TA} t_0  \to^*_{\TA} t'\) and \(t \to^*_{\TA} t_0  \to^*_{\TA} t''\), we have that \(|\delta(t_0)| > 1\).
This is a contradiction since \(\TA\) is deterministic.
Thus, we conclude that if \(\TA\) is a DBTA then \(\mathcal{L}_{\da}(q) = \mathcal{L}_{\da}(q')\) implies \(q = q'\).

Finally, we prove that \(\pre_x(F)\!=\!\pre_y(F)\) by contradiction.
Assume w.l.o.g. that \(q \in \pre_x(F) \cap (\pre_y(F))^{\complement}\).
Thus, as we have shown before, for each \(q' \in \pre_y(F)\) with  \(q'\neq q\), we have that \(\mathcal{L}_{\da}(q) \neq \mathcal{L}_{\da}(q')\).
Therefore, equality~\eqref{eq:proof_equality_downward} does not hold, which yields to contradiction.
We conclude that, necessarily, \(\pre_x(F) = \pre_y(F)\).
\end{enumerate}

\vspace*{-7mm}
\end{proof}

Finally, the following lemma shows that the blocks of \(\ru_L\) (resp.\ \(\rd_L\)) can be described as intersections of complemented (using the symbol \(\complement\) in superscript) and uncomplemented downward (resp.\ upward) quotients of \(L\).
Similarly, the blocks of \(\ruA\) correspond to intersections of complemented and uncomplemented downward languages  of states of \({\TA}\).

For the counterpart characterization of the blocks of \(\rdA\), the result is not symmetric to that of the blocks of \(\ruA\).
This is expected as \(\rdA\) is based on our definition of \(\pre(\cdot)\), which in turn is defined not only in terms of upward languages but also uses the notion of root-to-pivot equivalence.
We will use this lemma to give the generalization of Brzozowski's method for the construction of the minimal DBTA (in Section~\ref{sec:general-RDRD}).

\begin{lemma}
\label{blockBTAasIntersection}
Let \({\TA}= \tuple{Q,\Sigma,\delta,F}\) be a BTA without unreachable states and such that \(L = \lang{{\TA}}\) and \(I = \initials(\TA)\).

\medskip
Then, for every \(t\in \cT_\Sigma\), \(x \in \cC_\Sigma\), we have that:
\begin{align*}
P_{\ru_L}(t) & = \bigcap_{\mathclap{\substack{y \in L\,t^{-1}}}} y^{-1}L \; \cap \; \bigcap_{\mathclap{\substack{y \notin L\,t^{-1}}}} (y^{-1}L)^{\complement}\\
P_{\rd_L}(x) & = \bigcap_{\mathclap{\substack{r \in x^{-1}L}}} L\,r^{-1} \; \cap\; \bigcap_{\mathclap{\substack{r \notin x^{-1}L}}} (L\,r^{-1})^{\complement}\\
P_{\ruA}(t) & = \bigcap_{\mathclap{q \in \post_t(I)}} \mathcal{L}_{\da}(q) \; \cap \; \bigcap_{\mathclap{q \notin \post_t(I)}}(\mathcal{L}_{\da}(q))^{\complement}\\
P_{\rd_{{\TA}}}(x) & = \bigcap_{\mathclap{q \in \pre_x(F)}} P_{\simp}(\mathcal{L}_{\ua}(q))\; \cap \;\bigcap_{\mathclap{q \notin \pre_x(F)}} (P_{\simp}(\mathcal{L}_{\ua}(q)))^{\complement} \enspace .
\end{align*}
\end{lemma}

\begin{proof}
Let us start with the equalities on the upward congruences \(\ru_L\) and \(\ruA\).

\vspace*{0.6mm}
For each \(r \in \cT_\Sigma\) we have that:\vspace*{-1mm}
\begin{align*}
r \in  \bigcap_{\mathclap{y \in Lt^{-1}}} y^{-1}\,L \;\; \cap\;\; \bigcap_{\mathclap{y \notin Lt^{-1}}} (y^{-1}\,L)^{\complement} &\Lra^{\dagger}\\[2pt]
\forall y \in \cC_\Sigma \colon  y \in Lt^{-1} \Lra r \in y^{-1}\,L &\Lra^{\dagger\dagger} \\[-2pt]
\forall y \in \cC_\Sigma \colon  y \in Lt^{-1} \Lra y \in Lr^{-1}  &\Lra \\[-2pt]
Lt^{-1} = Lr^{-1} & \Lra \\[-2pt]
r \in P_{\ru_L}(t) \enspace ,
\end{align*}
where double-implication \(\dagger\) holds by definition of set intersection and \(\dagger\dagger\) holds since, for every \(r \in \cT_{\Sigma}, y \in \cC_{\Sigma}\colon r \in y^{-1}L \Lra y \in Lr^{-1}\).

\vspace*{0.6mm}
On the other hand, for each \(r \in \cT_{\Sigma}\):
\begin{align*}
r \in \bigcap_{\mathclap{q \in \post_t(I)}} \mathcal{L}_{\da}(q, I) \; \cap \bigcap_{\mathclap{q \notin \post_t(I)}}(\mathcal{L}_{\da}(q, I))^{\complement} &\Lra\\[-2pt]
\forall q \in Q\colon q \in \post_t(I) \Lra  r \in \mathcal{L}_{\da}(q, I) & \Lra^{\dagger}  \\[-2pt]
\forall q \in Q\colon  q \in \post_t(I) \Lra q \in \post_r(I)  & \Lra\\[-2pt]
\post_t(I) = \post_r(I) & \Lra \\[-2pt]
r \in P_{\ru_\TA}(t) \enspace ,
\end{align*}
where double-implication \(\dagger\) holds as, for each \mbox{\(r \in \cT_\Sigma, q \in Q \colon\)}  \(r \in\mathcal{L}_{\da}(q, I) \Lra q \in \post_r(I)\).

\vspace*{0.6mm}
Now we move to the equalities on downward congruences.
First, for each \(z \in \cC_\Sigma\) we have that:
\begin{align*}
z \in  \bigcap_{\mathclap{r \in x^{-1}L}} L\,r^{-1} \;\; \cap\;\; \bigcap_{\mathclap{r \notin x^{-1}L}} (L\,r^{-1})^{\complement} & \Lra\\[-2pt]
\forall r \in \cT_\Sigma \colon  r \in x^{-1}L \Lra z \in L\,r^{-1} & \Lra^{\dagger}  \\[-2pt]
\forall r \in \cT_\Sigma \colon  r \in x^{-1}L \Lra r \in z^{-1}L  & \Lra  \\[-2pt]
x^{-1}L = z^{-1}L & \Lra \\[-2pt]
z \in P_{\rd_L}(x) \enspace ,
\end{align*}
where double-implication \(\dagger\) holds since, for each \( z \in \cC_\Sigma, r \in \cT_\Sigma \colon  r \in z^{-1}L \Lra z \in L\,r^{-1}\).

\vspace*{0.6mm}
Finally, we prove the last equality.
For each \(z \in \cC_{\Sigma}\) we have that:
\begin{align*}
z \in  \bigcap_{\mathclap{q \in \pre_x(F)}} P_{\simp}(\mathcal{L}_{\ua}(q)) \;\; \cap\;\;  \bigcap_{\mathclap{q \notin \pre_x(F)}} (P_{\simp}(\mathcal{L}_{\ua}(q)))^{\complement} & \Lra\\[-2pt]
\forall q \in Q \colon  q \in \pre_x(F) \Lra z \in P_{\simp}(\mathcal{L}_{\ua}(q)) & \Lra^{\dagger} \\[-2pt]
\forall q \in Q \colon  q \in \pre_x(F) \Lra q \in \pre_z(F)  & \Lra \\[-2pt]
\pre_x(F) = \pre_z(F) & \Lra\\[-2pt]
z \in P_{\rdA}(x) & \enspace ,
\end{align*}
where double-implication \(\dagger\) holds as, for each \mbox{\( z \in \cC_\Sigma, q \in Q \colon\)} \(z \in P_{\simp}(\mathcal{L}_{\ua}(q)) \Lra q \in \pre_z(F)\).
\end{proof}

\subsection{Determinization and minimization of BTAs using congruences}\label{sec:detandCodet}

In what follows, we will use \(\cF{\mathsf{u}},\cF{\mathsf{d}}\) and \(\cG{\mathsf{u}},\cG{\mathsf{d}}\) to denote the constructions \(\cH^{\mathsf{u}},\cH^{\mathsf{d}}\) when applied, respectively, to the language-based congruences induced by a regular tree language and the automata-based congruences induced by a BTA\@.

\begin{definition}[\(\cF{\mathsf{u}},\cF{\mathsf{d}}\) and \(\cG{\mathsf{u}},\cG{\mathsf{d}}\)]%
\label{def:FG}
Let \({\TA}\) be a BTA  with \(L = \lang{{\TA}}\).
Define:
\begin{align*}
\cF{\mathsf{u}}(L) & \udiffg  \cH^{\mathsf{u}}(\ru_{L}, L) & \cG{\mathsf{u}}({\TA}) & \udiffg  \cH^{\mathsf{u}}(\ruA, L)\\
\cF{\mathsf{d}}(L) & \udiffg  \cH^{\mathsf{d}}(\rd_{L}, L) & \cG{\mathsf{d}}({\TA}) & \udiffg \cH^{\mathsf{d}}(\rdA, L)\enspace.
\end{align*}
\end{definition}

All the above constructions yield to BTAs defining the language \(L\) (recall that, additionally, we need \(L\) to be path-closed when using the downward congruences).
Concretely, \(\cG{\mathsf{u}}(\TA)\) and \(\cG{\mathsf{d}}(\TA)\) correspond to the bottom-up determinization and co-determinization operations defined in Section~\ref{sec:preliminaries}.

On the other hand, since \(\cF{\mathsf{u}}(L)\) and \(\cF{\mathsf{d}}(L)\) are constructed upon the language-based congruences, the resulting BTAs are minimal.
More precisely, \(\cF{\mathsf{u}}(L)\) yields  the minimal DBTA for \(L\), and \(\cF{\mathsf{d}}(L)\) yields  the minimal co-DBTA, as long as \(L\) is path-closed.
Finally, since \(\cG{\mathsf{d}}(\TA)\) is a co-DBTA satisfying the conditions of Theorem~\ref{automataEqualNerode}\ref{theorem:automata=nerodeUP}, when we apply to it the determinization operation (\(\cG{\mathsf{u}}\)) we obtain the minimal DBTA for \(L\).
Note that the fact that we construct it upon the co-DBTA \(\cG{\mathsf{d}}(\TA)\) requires that \(L\) is path-closed.
A similar result holds when co-determinizing (\(\cG{\mathsf{d}}\)) the DBTA \(\cG{\mathsf{u}}(\TA)\).
All these notions are summarized by the following theorem.

\begin{theorem}
\label{theoremF}
Let \({\TA}\) be a BTA with \(L = \lang{\TA}\).
Then the following properties hold:
\begin{enumerate}
\renewcommand\labelenumi{\theenumi}
\renewcommand{\theenumi}{(\alph{enumi})}
\itemsep=0.8pt
\item \(\lang{\cF{\mathsf{u}}(L)} = L  = \lang{\cG{\mathsf{u}}({\TA})}\).
\label{theoremF:languageBTA}
\item If \(L\) is path-closed then \(\lang{\cF{\mathsf{d}}(L)} = L = \lang{\cG{\mathsf{d}}({\TA})}\).\label{theoremF:languageTTA}
\item\(\cG{\mathsf{u}}({\TA}) \equiv \TA^{\db}\). \label{theoremF:DetismorphicDBTA}
\item If \(L\) is path-closed and \({\TA}\) has no unreachable states then \(\cG{\mathsf{d}}({\TA}) \equiv \TA^{\cdb}\). \label{theoremF:DetismorphiccoDBTA}
\item \(\cF{\mathsf{u}}(L)\) is isomorphic to the minimal DBTA for \(L\).\label{theoremF:MinimalDBTA}
\item If \(L\) is path-closed then \(\cF{\mathsf{d}}(L)\) is isomorphic to the minimal co-DBTA of \(L\).\label{theoremF:MinimalcoDBTA}
\item If \(L\) is path-closed then \(\cG{\mathsf{u}}(\cG{\mathsf{d}}(\TA)) \equiv \cF{\mathsf{u}}(L)\). \label{theoremF:DoubleReversalDBTA}
\item If \(L\) is path-closed then \(\cG{\mathsf{d}}(\cG{\mathsf{u}}(\TA)) \equiv \cF{\mathsf{d}}(L)\).\label{theoremF:DoubleReversalcoDBTA}
\end{enumerate}
\end{theorem}
\begin{proof}
Let \({\TA} \udiffg  \tuple{Q,\Sigma,\delta,F}\).
\begin{enumerate}
\renewcommand\labelenumi{\theenumi}
\renewcommand{\theenumi}{(\alph{enumi})}
\itemsep=0.9pt
\item \(\lang{\cF{\mathsf{u}}(L)} = L  = \lang{\cG{\mathsf{u}}({\TA})}\).

It is well-known that \(\ru_{L}\) is an upward congruence~\cite{kozen1993MyhillTrees,tata2007} and, by Lemma~\ref{automataCongruences}\ref{lemma:automataCongruences:upAB}, so is \(\ruA\).
By definition, we have that \(t \ru_{L} r \Ra Lt^{-1} = Lr^{-1}\).
On the other hand, also by Lemma~\ref{automataCongruences}\ref{lemma:automataCongruences:upAB}, we have that \(\mathord{\ruA} \subseteq \mathord{\ru_L}\).
Therefore, by Lemma~\ref{HBPreservesL}, \(\lang{\cF{\mathsf{u}}(L)} = L  = \lang{\cG{\mathsf{u}}({\TA})}\).

\item If \(L\) is path-closed then \(\lang{\cF{d}(L)} = L =  \lang{\cG{\mathsf{d}}({\TA})}\).%

Lemmas~\ref{NerodeStronglyDownward} and~\ref{automataCongruences}\ref{lemma:automataCongruences:downAT} show that both \(\rd_{L}\) and \(\rdA\) are  downward congruences satisfying the condition of Lemma~\ref{HTPreservesL}.
Therefore, \(\lang{\cF{d}(L)} = L  = \lang{\cG{\mathsf{d}}({\TA})}\).

\item\(\cG{\mathsf{u}}({\TA}) \equiv \TA^{\db}\).

Recall that, given \({\TA} = \tuple{Q, \Sigma, \delta, F}\), \(\TA^{\db}\) denotes the DBTA that results from applying the bottom-up determinization to \({\TA}\) and removing all unreachable states.
Let \(\TA^{\db} = \tuple{Q_{d}, \Sigma, \delta_d, F_d}\) and \(\cG{\mathsf{u}}({\TA}) = \tuple{\widetilde{Q}, \Sigma, \tilde{\delta}, \widetilde{F}}\).

Let \(P\) be the partition induced by \(\ruA\) and let \(\varphi: \widetilde{Q} \rightarrow Q_{d}\) be the mapping assigning to each state \(P(t) \in \widetilde{Q}\), the set \(\post^{\TA}_t(\initials(\TA)) \in Q_{d}\) with \(t \in \cT_\Sigma\).
Observe that \(\varphi\) is well-defined since, by def.~\ref{def:automataEquivalences}, \(t \ruA r \) if{}f \(\post^{\TA}_{t}(\initials({\TA})) = \post^{\TA}_{r}(\initials({\TA}))\).
We show that \(\varphi\) is BTA isomorphism between \(\cG{\mathsf{u}}({\TA})\) and \(\TA^{\db}\).
First, we show that \(\varphi(\initials(\cG{\mathsf{u}}({\TA}))) = \initials(\TA^{\db} )\).
To simplify the notation, let us denote \(\initials(\TA)\) as \(I\). \vspace*{-2mm}
\begin{align*}
\varphi(\initials(\cG{\mathsf{u}}({\TA}))) & = \\
\varphi(\{P(t) \in \wt{Q} \mid \exists a \in \Sigma_0 \colon  P(t) \in \tilde{\delta}(a[\,])\}) & = \quad \text{[Def. of \(\varphi\)]} \\
\{\post^{\TA}_t(I) \mid \exists a \in \Sigma_0 \colon  P(t) \in \tilde{\delta}(a[\,])\} & = \quad \text{[Def.~\ref{def:up_TA_construction}]} \\
\{\post^{\TA}_t(I) \mid \exists a \in \Sigma_0 \colon  a \in P(t)\} & = \quad \text{[Def. of \(P\)]} \\
\{\post^{\TA}_t(I) \mid \exists a \in \Sigma_0 \colon  \post_a(I) = \post_{t}(I)\} \enspace .
\end{align*}
Rewriting the above equation we have that:
\begin{align*}
\varphi(\initials(\cG{\mathsf{u}}({\TA}))) & = \\
\{\post_a(I) \mid  a \in \Sigma_0\} & = \quad \text{[Def. of \(\post_{a}(I)\)]} \\
\{q \in Q \mid q \in \delta(a[\,]),  a \in \Sigma_0\} & = \quad \text{[Def. of \(\initials(\TA^{\db})\)]} \\
\initials(\TA^{\db}) \enspace .
\end{align*}

Next, we show that \(\varphi\) is surjective, i.e. \(\forall S \in Q_d, \exists t \in \cT_\Sigma \colon S = \post^{\TA}_t(I)\).
We proceed by induction on the structure of \(\TA^{\mathsf{D}}\), i.e., we set the base case to \(S = \initials(\TA^{\mathsf{D}})\) and we use the transition function \(\delta_d\) onwards to reach all \(S \in Q_d\).
Recall that \(\TA^{\mathsf{D}}\) has no unreachable states.
\begin{itemize}
\item \emph{Base case:}
Let \(S = \initials(\TA^{\mathsf{D}})\).
Then, \(S = \post_a(I)\), with \(a \in \Sigma_0\).

\item \emph{Inductive step:}
W.l.o.g., let us assume that \(f \in \Sigma\) with \(\rank{f} = 2\).
Let \(S,S_1,S_2 \in Q_d\) be such that \(S \in \delta_d(f[S_1, S_2])\).
Assume \(\exists t_1,t_2 \in \cT_\Sigma\) such that \(S_i = \post_{t_i}(I)\) with \(i \in \{1,2\}\).
Then, \vspace*{-2mm}
\begin{align*}
S & =  \\
\{q \in Q \mid \exists q_1 \in S_1, q_2 \in S_2,  q \in \delta(f[q_1,q_2])\} & =  \\
\{q \in Q \mid \exists q_1 \in \post_{t_1}(I), q_2 \in \post_{t_{2}} (I),&\\
q \in \delta(f[q_1, q_2])\} & = \\
\post_{f[t_1,t_2]}(I) \enspace ,
\end{align*}
where the first equality holds by definition of \(\TA^{\mathsf{D}}\), the second equality holds by definition of \(S_i\) with \(i \in  \{1,2\}\), and the third equality holds using Lemma~\ref{prepost}\ref{lemma:prepost:prop:tree}.

Since \(\TA^{\db}\) keeps only reachable states, it follows that \(\forall S \in Q_d, \exists t\in \cT_\Sigma \colon  \varphi(P(t)) = S\), i.e. \(\varphi\) is surjective.
\end{itemize}

It is routine to check that \(\varphi\) is injective for otherwise there exists \(t \not\ruA r\) with \(\post^{\TA}_{t}(\initials({\TA})) = \post^{\TA}_{r}(\initials({\TA}))\) which contradicts definition~\ref{def:automataEquivalences}.
We thus conclude from above (\(\varphi\) is injective and surjective) that \(\varphi\) is bijective.

Next, we show that \(\varphi(\wt{F}) = F_d\): \vspace*{-2mm}
\begin{align*}
\varphi(\wt{F}) & = \quad \text{[Def.~\ref{def:up_TA_construction}]} \\
\varphi(\{P(t) \mid t \in L \}) & = \quad \text{[Def. of \(\varphi\)]} \\
\{\post^{\TA}_t(I) \mid t \in L\} & = \quad \text{[Def. of \(\lang{\TA}\)]} \\
\{\post^{\TA}_t(I) \mid \post^{\TA}_t(I) \cap F \neq \varnothing\} & =  \\
F_d \enspace .
\end{align*}
Note that the last equality holds by definition of \(\TA^{\db}\) and using the fact that \(\varphi\) is bijective.

Finally, w.l.o.g., assume that \(f \in \Sigma\) with \(\rank{f} = 2\), and let \(P(t),P(t_1), P(t_2) \in \wt{Q}\).
By definition of  \(\TA^{\db}\) we have that:
\begin{align*}
\varphi(P(t)) \in \delta_d(f[\varphi(P(t_1)),\varphi(P(t_2))]) & \Lra \\
\post^{\TA}_t(I) = \{q' \in Q \mid \exists q'_i \in \post_{t_i}(I) \colon  q' \in \delta(f[q'_1,q'_2])\}& \enspace .
\end{align*}
Using Lemma~\ref{prepost}\ref{lemma:prepost:prop:tree}, we have that the above equation is equivalent to:
\begin{align*}
\post^{\TA}_t(I) = \post_{f[t_1,t_2]}(I) & \Lra \;\; \text{[Def. \(\ruA\)]} \\
t \ruA  f[t_1,t_2] & \Lra \;\; \text{[Def. of \(P\)]} \\
P(t) = P(f[t_1, t_2]) &\Lra^{\dagger}\text{[\(f[P(t_1),P(t_2)] \subseteq P(f[t_1, t_2])\)]}\\
f[P(t_1),P(t_2)] \subseteq P(t) & \Lra \;\; \text{[Def.~\ref{def:up_TA_construction}]} \\
P(t) \in \tilde{\delta}(f[P(t_1),P(t_2)]) \enspace .
\end{align*}
Note that double-implication \(\dagger\) holds since \(\ruA\) is an upward congruence and thus, $f[P(t_1)$, $P(t_2)] \subseteq P(f[t_1, t_2])$.

\item If \(L\) is path-closed and \({\TA}\) has no unreachable states then \(\cG{\mathsf{d}}({\TA}) \equiv (\TA)^{\mathsf{cD}}\).

Recall that, given \({\TA} = \tuple{Q, \Sigma, \delta, F}\), \(\TA^{\mathsf{cD}}=\tuple{Q_{d}, \Sigma, \delta_d, F_d}\) denotes the co-DBTA that results from applying the co-determinization operation to \({\TA}\) and removing all empty states.
Let \(\cG{\mathsf{d}}({\TA}) = \tuple{\widetilde{Q}, \Sigma, \tilde{\delta}, \widetilde{F}}\).
Let \(P\) be the partition induced by \(\rdA\) and let \(\varphi: \widetilde{Q} \rightarrow Q_{d}\) be the mapping assigning to each state \(P(x) \in \widetilde{Q}\), the set \(\pre^{\TA}_x(F) \in Q_{d}\).
Observe that \(\varphi\) is well-defined since, by def.~\ref{def:automataEquivalences}, \(x \rdA y \) if{}f \(\pre^{\TA}_{x}(F) = \pre^{\TA}_{y}(F)\).
We show that \(\varphi\) is an isomorphism between \(\cG{\mathsf{d}}({\TA})\) and \({\TA}^{\mathsf{cD}}\).
First, we show that \(\varphi(\wt{F}) = F_d\).
\begin{align*}
\varphi(\wt{F}) & = \quad \text{[Def.~\ref{def:down_TA_construction}]} \\
\varphi(\{P(\cX)\}) & = \quad \text{[Def. of \(\varphi\)]} \\
\{\pre^{\TA}_{\cX}(F)\} & = \quad \text{[Def. of \(\pre^{\TA}_{\cX}(F)\)]} \\
\{F\} & = \quad \text{[Def. of \(F_d\)]} \\
F_d \enspace . \vspace*{-2mm}
\end{align*}

Next, we show that \(\varphi\) is surjective, i.e. \(\forall S \in Q_d, \exists x \in \cC_\Sigma \text{ with } x^{-1}L \neq \varnothing \colon  S = \pre^{\TA}_x(F)\).
We proceed by induction on the structure of \(\TA^{\mathsf{cD}}\), i.e., we set the base case to \(S = F_d\) and we use the transition function \(\delta_d\) backwards to reach all \(S \in Q_d\).
Recall that \(\TA^{\mathsf{cD}}\) has no empty states and the set of states of  \(\cG{\mathsf{d}}({\TA})\), i.e., \(\widetilde{Q}\), is defined as \(\widetilde{Q} = \{P(x) \mid x \in \cC_{\Sigma}, x^{-1}L \neq \varnothing\}\).
\begin{itemize}
\itemsep=0.9pt
\item \textit{Base case:}
Let \(S = F_d\). Then \(S = \pre^{\TA}_{x}(F)\) with \(x = \cX\).

\item \textit{Inductive step:}
Now, let \(f \in \Sigma\) and let \(S,S_1,\ldots,S_{\rank{f}} \in Q_d\) be such that $S \in \delta_d(f[S_1,\ldots\,$, $S_{\rank{f}}])$.
By I.H., \(\exists x \in \cC_\Sigma \text{ with } x^{-1}L \neq \varnothing\) \(\colon  S = \pre^{\TA}_x(F)\).
Now let us show that \(\forall i \in \range{\rank{f}}, \exists z \in \cC_\Sigma \text{ with } z^{-1}L \neq \varnothing \colon  S_i = \pre^{\TA}_z(F)\).
By definition of \(\TA^{\mathsf{cD}}\) we have that, for all \(i \in \range{\rank{f}}\):
\begin{align*}
S_i =
\{q_i \in Q \mid \exists &q' \in S,\exists q_1,\ldots,q_{i{-}1},q_{i{+}1},\ldots,q_{\rank{f}} \in Q \colon\\
&q' \in \delta(f[q_1,\ldots,q_{\rank{f}}]) \}  \enspace .
\end{align*}
Since \(\TA\) has no unreachable states, the above set is equivalent to:
\begin{gather}
\{q_i \in Q \mid \exists t \in \cT_\Sigma, \exists r \in \cT_Q \colon t(\varepsilon) = f , \subs{t}{q_i}_i \to^*_{{\TA}} r, \notag\\
r(\varepsilon) \in S\} \enspace . \label{eq:def_q_i} \vspace*{-2mm}
\end{gather}

Note that \(q_i \in Q\) belongs to the set in~\eqref{eq:def_q_i} if{}f \(q_i \in \pre_{\subs{\widetilde{t}}{\cX}_i}(S)\), for every \(\widetilde{t} \in \cT_{\Sigma}\) with \(\subs{\widetilde{t}}{\cX}_i \simp^S \subs{t}{\cX}_i\).
In fact, \(\exists t \in \cT_\Sigma, \exists r \in \cT_Q \colon  t(\varepsilon) = f ,\; r(\varepsilon) \in S\) and \(\subs{t}{q_i}_i \to^*_{{\TA}} r\) implies that \(q_i \in \pre_{\subs{\widetilde{t}}{\cX}_i}(S)\) by setting \(\subs{\widetilde{t}}{\cX}_i = \subs{t}{\cX}_i\).
On the other hand, if \(q_i \in \pre_{\subs{\widetilde{t}}{\cX}_i}(S)\) with \(\subs{\widetilde{t}}{\cX}_i \simp^S \subs{t}{\cX}_i\) then, by definition of \( \pre_{\subs{\widetilde{t}}{\cX}_i}(S)\), \(\exists y \in \cC_{\Sigma}\) with \(y \simp^S  \subs{\widetilde{t}}{\cX}_i \colon y \in \mathcal{L}^{\TA}_{\ua}(q_i,S)\).
Therefore, \(\exists t \in \cT_\Sigma, \exists r \in \cT_Q \colon  t(\varepsilon) = f,\, \subs{t}{q_i}_i \to^*_{{\TA}} r\) and \(\; r(\varepsilon) \in S\) by setting \(\subs{t}{\cX}_i = y\).
\eject
Thus, the set in~\eqref{eq:def_q_i} is equivalent to:
\begin{equation*}
\{q_i \in Q \mid \exists \widetilde{t} \in \cT_\Sigma \colon \subs{\widetilde{t}}{\cX}_i \simp^S \subs{t}{\cX}_i, q_i \in \pre^{\TA}_{\subs{\widetilde{t}}{\cX}_i}(S)\} \enspace .
\end{equation*}
By induction hypothesis we have that the above set is equivalent to:
\begin{gather*}
\{q_i \in Q \mid \exists \widetilde{t}\in \cT_\Sigma,\, \exists x \in \cC_{\Sigma}\colon \subs{\widetilde{t}}{\cX}_i \simp^S \subs{t}{\cX}_i,\,  x^{-1}L \neq \varnothing,\\
\,q_i \in \pre^{\TA}_{\subs{\widetilde{t}}{\cX}_i}(\pre^{\TA}_x(F))\} \enspace .
\end{gather*}
Finally, using Lemma~\ref{prepost}\ref{lemma:prepost:prop:ctx}, the latter is equivalent to:
\begin{gather*}
\{q_i \in Q \mid \exists \widetilde{t}\in \cT_\Sigma,\, \exists x \in \cC_{\Sigma}\colon \subs{\widetilde{t}}{\cX}_i \simp^S \subs{t}{\cX}_i,\,  x^{-1}L \neq \varnothing,\\
q_i \in \pre^{\TA}_{\subs{x}{\subs{\widetilde{t}}{\cX}_i}}(F)\} \enspace .
\end{gather*}
Therefore, for all non-empty \(S_i \in Q_d\), there exists \(z (= \subs{x}{\subs{\widetilde{t}}{\cX}_i}) \in \cC_\Sigma\) with \(z^{-1}L \neq \varnothing\) such that \(S_i = \pre^{\TA}_{z}(F)\) and, since \(\TA^{\mathsf{cD}}\) has no empty states, it follows that \(\varphi\) is surjective.
\end{itemize}\smallskip
It is routine to check that \(\varphi\) is injective for otherwise there exists \(x \not\rdA y\) with \(\pre^{\TA}_{x}(F) = \pre^{\TA}_{y}(F)\) which contradicts definition~\ref{def:automataEquivalences}.
We thus conclude from above (\(\varphi\) is injective and surjective) that \(\varphi\) is bijective.

\medskip
On the other hand, \(\varphi(\initials(\cG{\mathsf{d}}(\TA))) = \initials(\TA^{\mathsf{cD}})\) since:
\begin{align*}
\varphi(\initials(\cG{\mathsf{d}}(\TA))) & = \\
\varphi(\{P(x) \in \wt{Q} \mid \exists a \in \Sigma_0 \colon  P(x) \in \widetilde{\delta}(a[\,])\}) & = \\
\{\pre^{\TA}_x(F) \mid \exists a \in \Sigma_0 \colon  P(x) \in \widetilde{\delta}(a[\,])\} & = \\
\{\pre^{\TA}_x(F) \mid \exists a \in \Sigma_0 \colon  \subs{P(x)}{a} \subseteq L\} &\enspace .
\end{align*}
By  definition of \(\lang{\TA}\) and Lemma~\ref{lemma:simpImpliesQuotientEqual} we have that the above set is equivalent to:
\begin{equation*}
\{\pre^{\TA}_x(F) \mid \exists a \in \Sigma_0, \exists q \in \pre_x(F) \colon  q \in \delta(a[\,])\} \enspace .
\end{equation*}
Finally, since \(\varphi\) is bijective we have that the latter is equivalent to:
\begin{align*}
\{q_d \in Q_d \mid  \exists a \in \Sigma_0, \exists q \in q_d \colon  q \in \delta(a[\,])\} & = \; \text{[Def. of \(\delta_d\)]} \\
\{q_d \in Q_d \mid \exists a \in \Sigma_0 \colon  q_d \in \delta_d(a[\,])\} & = \; \text{[Def. of \(\initials(\TA^{\mathsf{cD}})\)]} \\
\initials(\TA^{\mathsf{cD}}) \enspace . \vspace*{-3mm}
\end{align*}

Finally, we prove that:
\begin{gather*}
\varphi(P(x)) \in \delta_d(f[\varphi(P(x_1)),\ldots,\varphi(P(x_{\rank{f}}))]) \text{ if{}f }\\
P(x)\in \tilde{\delta}(f[P(x_1), \ldots, P(x_{\rank{f}})]) \enspace ,
\end{gather*}
where \(P(x),P(x_1),\ldots, P(x_k) \in \wt{Q}\) and \(f \in \Sigma\).

\eject

Note that, by definition of \(\TA^{\mathsf{cD}}\) and \(\varphi\),   \(\varphi(P(x)) \in \delta_d(f[\varphi(P(x_1)),\ldots,\varphi(P(x_{\rank{f}}))])\) if{}f, for each \mbox{\(i \in \range{\rank{f}}\colon\)}
\begin{gather*}
\pre^{\TA}_{x_i}(F) = \{q_i \in Q \mid \exists q \in \pre^{\TA}_x(F), \exists q_1,\ldots,q_{\rank{f}} \in Q \colon\\
q \in \delta(f[q_1,\ldots,q_{\rank{f}}])\} \enspace .
\end{gather*}
Since \(\TA\) has no unreachable states, for all \(i \in \range{\rank{f}}\), we have that:
\begin{gather*}
\pre^{\TA}_{x_i}(F) = \{q_i \in Q \mid \exists q \in \pre^{\TA}_x(F), \exists t \in \cL_{\da}(q) \colon\\
t(\varepsilon) = f, q_i \in \pre^{\TA}_{\subs{t}{\cX}_i}(\{q\})\} \enspace .
\end{gather*}

By Lemma~\ref{prepost}\ref{lemma:prepost:prop:ctx}, we have that the above equality is equivalent to:
\begin{align*}
 \exists t \in \cT_{\Sigma} \colon t(\varepsilon) = f, \pre^{\TA}_{x_i}(F) = \pre^{\TA}_{\subs{x}{\subs{t}{\cX}_i}}(F)  & \Lra\\
  \exists t \in \cT_{\Sigma}  \colon  t(\varepsilon) = f, x_i \rdA  \subs{x}{\subs{t}{\cX}_i} & \Lra \\
    \exists t \in \cT_{\Sigma}  \colon  t(\varepsilon) = f, P(x_i) =  P(\subs{x}{\subs{t}{\cX}_i}) & \Lra^{\dagger} \\
 \exists t \in \cT_{\Sigma}  \colon t(\varepsilon) = f,  \subs{P(x)}{\subs{t}{\cX}_i} \subseteq P(x_i) & \Lra \\
P(x) \in \tilde{\delta}(f[P(x_1),\ldots,P(x_k)]) \enspace .
\end{align*}
Note that double-implication \(\dagger\) holds since \(\rdA\) is a downward congruence and thus, \\ \(\subs{P(x)}{\subs{t}{\cX}_i} \subseteq P(\subs{x}{\subs{t}{\cX}_i})\).

\item \(\cF{\mathsf{u}}(L)\) is isomorphic to the minimal DBTA for \(L\).%

Let \(P\) be the partition induced by the Nerode's upward congruence \(\ru_L\).
As shown by Comon et al.~\cite{tata2007}, the minimal DBTA for \(L\) is the BTA \(\cM=\tuple{Q',\Sigma,\delta',F'}\) where \(Q' = \{P(t) \mid t\in \cT_\Sigma\}\), \(F' = \{P(t) \mid t \in L\}\) and \(\delta'(f[P(t_1),\ldots,P(t_{\rank{f}})]) = P(f[t_1,\ldots,t_{\rank{f}}])\) for every \(f \in \Sigma\) and \(t,t_1,\ldots,t_{\rank{f}} \in \cT_\Sigma\).
Hence, it is easy to check that \(\cM \equiv \cF{\mathsf{u}}(L)\).

\item If \(L\) is path-closed then \(\cF{\mathsf{d}}(L)\) is isomorphic to the minimal co-DBTA for \(L\).

Trivially, \(\mathord{\rd} \udiffg  \mathord{\rd_L}\) is the coarsest strongly downward congruence w.r.t. \(L\) that satisfies:
\begin{align}\label{eq:preserveCtx}
\forall x,y \in \cC_\Sigma \colon  x \rd y \Ra x^{-1}L = y^{-1}L \enspace .
\end{align}
Next, we show that \(\cF{\mathsf{d}}(L)\) is the minimal co-DBTA for \(L\) by contradiction.

Let \(\TA' = \tuple{Q,\Sigma,\delta,F}\) be a co-DBTA for \(L\) with strictly fewer states than \(\cF{\mathsf{d}}(L)\).
Note that we can assume that \(\TA'\) has no unreachable states, otherwise we could simply remove them and obtain a smaller equivalent co-DBTA.

First, we show that, for every co-DBTA \(\TA = \tuple{Q,\Sigma,\delta,F}\) with no unreachable states, the set \(\pre_x^{\TA}(F)\) is a singleton, for every \(x \in \cC_{\Sigma}\).
We proceed by contradiction.
Let \(F = \{q_f\}\), \(x \in \cC_{\Sigma}\) and assume that \(\exists q,q' \in Q\) with \(q \neq q'\) and such that \(q,q'\in \pre_x^{\TA}(F)\) , i.e., \(x \in P_{\simp}(\mathcal{L}_{\ua}(q))\) and  \(x \in P_{\simp}(\mathcal{L}_{\ua}(q'))\).
By definition of \(P_{\simp}\), \(\exists y \simp x \colon y \in \mathcal{L}_{\ua}(q)\) and \(\exists y' \simp x \colon y \in \mathcal{L}_{\ua}(q')\), i.e., \(\exists t, t' \in \cT_{Q} \colon \subs{y}{q} \to^*_{\TA} t, \subs{y'}{q'} \to^*_{\TA} t' \) and \(t(\varepsilon) = t'(\varepsilon) = q_f\).
Observe that \(y \simp y'\), by transitivity of \(\simp\).

Since \(q \neq q'\), \(t(\varepsilon) = t'(\varepsilon) = q_f\) and  \(y \simp y'\), there exists necessarily \({y_1}, {y'_1} \in \cC_{\Sigma \cup Q}\) with \(y_1 \simp y'_1\), \(r_i, r_i' \in \cT_Q\), \(f \in \Sigma\) and \(q_0 \in Q\) s.t.:
\begin{align*}
\subs{y}{q} \to^*_{\TA} \subs{y_1}{f[r_1, \ldots, r_{\rank{f}}]} &\to_{\TA} \subs{\tilde{x}}{q_0[r_1, \ldots, r_{\rank{f}}]}  \to^*_{\TA} t\\
&\text{and}\\
\subs{y'}{q'} \to^*_{\TA} \subs{y'_1}{f[r'_1, \ldots, r'_{\rank{f}}]} &\to_{\TA} \subs{\tilde{y}}{q_0[r'_1, \ldots, r'_{\rank{f}}]} \to^*_{\TA} t' \enspace ,
\end{align*}
where \(r_i \neq r'_i\), for some \(i \in \range{\rank{f}}\).
It follows that \(\exists q_0 \in \delta(f[r_1(\varepsilon), \ldots, r_{\rank{f}}(\varepsilon)])\) and \(q_0 \in \delta(f[r'_1(\varepsilon), \ldots, r'_{\rank{f}}(\varepsilon)])\) with \(r_i(\varepsilon) \neq r'_i(\varepsilon)\) for some \(i \in \range{\rank{f}}\), which contradicts the fact that \(\TA\) is co-deterministic.
Thus, we conclude that if \(\TA\) is a co-DBTA with no unreachable states, \(\pre_x^{\TA}(F)\) is a singleton, for every \(x \in \cC_{\Sigma}\).

Therefore, \(\rd_{\TA'}\), which by Lemma~\ref{automataCongruences}\ref{lemma:automataCongruences:downAT} is a strongly downward congruence w.r.t. \(L\) that satisfies Equation~\eqref{eq:preserveCtx}, has as many equivalence classes as states has \(\TA'\).
Since \(\TA'\) has fewer states than \(\cF{\mathsf{d}}(L)\), it follows that \(\mathord{\rd_{\TA'}} \subset \mathord{\rd_L}\), which contradicts the fact that \(\rd_L\) is the coarsest strongly downward congruence that satisfies Equation~\eqref{eq:preserveCtx}.

Therefore \(\TA'\) has, at least, as many states as \(\cF{\mathsf{d}}(L)\) and, as a consequence, \(\cF{\mathsf{d}}(L)\) is the minimal co-DBTA for \(L\).\vspace*{-1mm}

\item If \(L\) is path-closed then \(\cG{\mathsf{u}}(\cG{\mathsf{d}}(\TA)) \equiv \cF{\mathsf{u}}(L)\).

It follows from Corollary~\ref{cor:HTPreservesL} and~\ref{downwardIsDet} that \(\cG{\mathsf{d}}(\TA)\) is a co-DBTA with \(\lang{\cA} = \lang{\cG{\mathsf{d}}(\TA)}\).
Furthermore, by Remark~\ref{remark:no-determinism}, \(\cG{\mathsf{d}}(\TA)\) has no empty states.
Therefore, by Theorem~\ref{automataEqualNerode}\ref{theorem:automata=nerodeUP}, we have that \(\mathord{\ru_{\cG{\mathsf{d}}(\TA)}} = \mathord{\ru_{\lang{\cA}}}\) so \(\cG{\mathsf{u}}(\cG{\mathsf{d}}(\TA)) \equiv \cF{\mathsf{u}}(L)\).\vspace*{-1mm}

\item If \(L\) is path-closed then \(\cG{\mathsf{d}}(\cG{\mathsf{u}}(\TA)) \equiv \cF{\mathsf{d}}(L)\).

It follows by Remark~\ref{remark:DBTA} that \(\cG{\mathsf{u}}(\TA)\) is a DBTA with no unreachable states.
Furthermore, by Lemma~\ref{HBPreservesL}, \(\lang{\cA} = \lang{\cG{\mathsf{u}}(\TA)}\) .
Thus, by Theorem~\ref{automataEqualNerode}\ref{theorem:automata=nerodeDW}, we have that \( \mathord{\rd_{\cG{\mathsf{u}}(\TA)}} = \mathord{\rd_{\lang{\cA}}}\), so \(\cG{\mathsf{d}}(\cG{\mathsf{u}}(\TA)) \equiv \cF{\mathsf{d}}(L)\).
\end{enumerate}

\vspace*{-7mm}
\end{proof}

\section{A congruence-based perspective on Brzozowski's method}
\label{sec:Novel}

Brzozowski's double-reversal method~\cite{brzozowski1962canonical} is a classical algorithm for minimizing word automata.
It relies on the fact that determinizing a co-deterministic automaton \(\cN\) yields  the minimal deterministic automaton for the language of \(\cN\).
In this section, we show, as an easy consequence of Theorem~\ref{theoremF}\ref{theoremF:DetismorphicDBTA},~\ref{theoremF:DetismorphiccoDBTA}, and~\ref{theoremF:DoubleReversalDBTA}, that this algorithm can be adapted for finding the minimal DBTA for the language of a given BTA.
To the best of our knowledge this is the first proof of correctness of this method for minimizing bottom-up tree automata.
A precise statement is given next followed by a justification.

\begin{corollary}\label{thm:Brzozowski-method}
Let \(\TA\) be a BTA without unreachable states such that \(L = \lang{\TA}\) is a path-closed language.
Then, \((\TA^{\cdb})^{\db} \equiv \cF{\mathsf{u}}(L)\).
\end{corollary}

To be precise, Brzozowski's double-reversal method for word automata constructs the intermediate co-deterministic automaton by combining a reverse operation, followed by a determinization operation and then a reverse operation again.
However, the reverse construction applied to BTAs, i.e., the operation that flips the direction of the transition function and switches final states by initial states, does not yield a BTA.
On the contrary, it produces, a \emph{top-down tree automaton} (TTA).

Top-down tree automata (see~\cite{tata2007} for a detailed description of this model), as opposed to the class of BTAs, start their computations from the root, which is an \emph{initial} state of the automaton, down to the leaves.
Since TTAs can be interpreted as the \emph{reverse} of BTAs, indeed both classes of automata define the regular tree languages, all the results of Section~\ref{sec:CongruencesandBTA} have TTA-equivalents that we defer to Appendix~\ref{appendix:TTA}.
It is worth noting that, as co-DBTAs, deterministic TTAs are strictly less expressive than the general TTAs and  define the subclass of path-closed languages.
Since Brzozowski's method goes through the construction of an intermediate co-DBTA \(\TA^{\cdb}\), this algorithm is restricted to BTAs defining path-closed tree languages.

Relying on the definitions and notation we introduce in Appendix~\ref{appendix:TTA}, we conclude this section by stating Brzozowski's method using the reverse operation.
Namely, given a BTA \(\TA\), \(\TA^R\) denotes the reverse TTA of \(\TA\) and \((\TA^R)^{\db}\) denotes the result of applying the counterpart determinization operation on TTAs.
Thus, \(((\TA^R)^{\db})^R \equiv \TA^{\cdb}\) and the following holds.

\begin{corollary}
\label{cor:RDRD}
Let \(\TA\) be a BTA without unreachable states such that \(L = \lang{\TA}\) is a path-closed language.
Then, \((((\TA^R)^{\db})^R)^{\db} \equiv \cF{\mathsf{u}}(L)\).
\end{corollary}

\subsection{Generalization of Brzozowski's method}
\label{sec:general-RDRD}
Brzozowski's double-reversal methods builds a co-DBTA in order to guarantee, by Theorem~\ref{automataEqualNerode}\ref{theorem:automata=nerodeUP}, that determinizing the automaton produces the minimal DBTA.

In the word automata case, Brzozowski and Tamm~\cite{brzozowski2014theory} showed that going through a co-deterministic automata is not necessary and defined a class of automata, which strictly contains the co-deterministic ones, for which determinizing the automata yields the minimal deterministic automaton.
Next we generalize that result from word to trees which, incidentally, allows us to drop the restriction to the path-closed languages.

\begin{theorem}
\label{theorem:minimalifreverseatomic}
Let \(\TA = \tuple{Q,\Sigma,\delta,F}\) be a BTA with \(L=\lang{\TA}\).
Then \(\cG{\mathsf{u}}(\TA)\equiv \cF{\mathsf{u}}(L)\) if{}f \(\forall q \in Q \colon   P_{\ru_L}(\mathcal{L}_{\da}(q)) = \mathcal{L}_{\da}(q)\).
\end{theorem}
\begin{proof}
First, we show that if \(\cG{\mathsf{u}}(\TA)\) is the minimal DBTA for \(L\) then we have \(\forall q \in Q \colon   P_{\ru_L}(\mathcal{L}_{\da}(q)) = \mathcal{L}_{\da}(q)\).

\medskip
To simplify the notation, let \(I\) denote \(\initials(\TA)\).
Then, we have that:
\begin{align*}
P_{\ru_L}(\mathcal{L}_{\da}(q)) & =\\
\{r \in \cT_\Sigma \mid \exists t \in \mathcal{L}_{\da}(q) \colon  L\,r^{-1} = L\,t^{-1}\} & = \;\; \text{[\(\mathord{\ru_L} = \mathord{\ru_{\TA}}\)]} \\
\{r \in \cT_\Sigma \mid \exists t \in \mathcal{L}_{\da}(q) \colon  \post^{\TA}_r(I) = \post^{\TA}_t(I)\} &\enspace .
\end{align*}
Note that \(\mathord{\ru_L} = \mathord{\ru_{\TA}}\) by hypothesis, since \(\cG{\mathsf{u}}(\TA)\equiv \cF{\mathsf{u}}(L)\).
On the other hand, by definition, \(  q \in \post^{\TA}_t(I) \Ra t \in \mathcal{L}_{\da}(q)\).
Thus, we have the following set inclusion:
\begin{align*}
\{r \in \cT_\Sigma \mid \exists t \in \mathcal{L}_{\da}(q) \colon  \post^{\TA}_r(I) = \post^{\TA}_t(I)\} &\subseteq \\
\{r \in \cT_\Sigma \mid q \in \post^{\TA}_r(I)\} &=\\
\mathcal{L}_{\da}(q) &\enspace .
\end{align*}
By reflexivity of \(\ru_L, \) we have that \(\mathcal{L}_{\da}(q) \subseteq P_{\ru_L}(\mathcal{L}_{\da}(q))\), and thus we conclude \(P_{\ru_L}(\mathcal{L}_{\da}(q)) = \mathcal{L}_{\da}(q) \).

\medskip
Now, assume that \(P_{\ru_L}(\mathcal{L}_{\da}(q)) = \mathcal{L}_{\da}(q)\), for each \(q \in Q\). %
Then, for every \(t \in \cT_\Sigma\),
\begin{align}
P_{\ru_{\TA}}(t) & =^{\dagger} \notag \\
\bigcap\limits_{\substack{q \in \post^{\TA}_t(I)}} \mathcal{L}_{\da}(q) \; \cap \bigcap\limits_{\substack{q \notin \post^{\TA}_t(I)}} (\mathcal{L}_{\da}(q))^{\complement} & =^{\dagger\dagger} \notag \\
\bigcap\limits_{\substack{q \in \post^{\TA}_t(I)}} P_{\ru_L}(\mathcal{L}_{\da}(q)) \; \cap \bigcap\limits_{\substack{q \notin \post^{\TA}_t(I)}} (P_{\ru_L}(\mathcal{L}_{\da}(q)))^{\complement}\enspace . \label{eq:union_blocks_uL}
\end{align}
Note that equality \(\dagger\) holds by Lemma~\ref{blockBTAasIntersection} and \(\dagger\dagger\) holds since \(P_{\ru_L}(\mathcal{L}_{\da}(q)) = \mathcal{L}_{\da}(q)\) by hypothesis.

\medskip
It follows from~\eqref{eq:union_blocks_uL} that \(P_{\ru_{\TA}}(t)\) is a \emph{union} of blocks of \(P_{\ru_L}\), for each \(t \in \cT_{\Sigma}\).
In other words, \(\mathord{\ru_L} \subseteq \mathord{\ru_{\TA}}\).
On the other hand, by Lemma~\ref{automataCongruences}\ref{lemma:automataCongruences:upAB}, \(\mathord{\ru_{\TA}} \subseteq \mathord{\ru_L}\).
Therefore, \(P_{\ru_{\TA}}(t)\) necessarily corresponds to one single block of \(P_{\ru_L}\), namely, \(P_{\ru_L}(t)\).
Since \(P_{\ru_{\TA}}(t) = P_{\ru_L}(t)\) for each \(t \in \cT_\Sigma\), we conclude that \(\cG{\mathsf{u}}(\TA)\equiv \cF{\mathsf{u}}(L)\).
\end{proof}

Given a regular tree language \(L\), the minimal DBTA for \(L\) is \(\cF{\mathsf{u}}(L)\) (Theorem~\ref{theoremF}\ref{theoremF:MinimalDBTA}).
On the other hand, we show (see proof of Lemma~\ref{HBPreservesL}) that the states of the minimal DBTA are in one-to-one correspondence with the blocks of \(P_{\ru_L}\), i.e., for every state \(q\) of \(\cF{\mathsf{u}}(L)\), there exists a tree \(t \in \cT_\Sigma\) such that \(\cL_{\da}(q) = P_{\ru_L}(t)\) and vice-versa.
Therefore, for every \(S \subseteq \cT_\Sigma\), \(P_{\ru_L}(S) = S\)  if{}f \(S\) is a union of downward languages of states of the minimal DBTA for \(L\).
This property and Theorem~\ref{theorem:minimalifreverseatomic} allows us to give an alternative characterization of the class of automata for which the determinization operation yields the minimal DBTA.

\begin{corollary}\label{coro:cond2}
Let \(\TA\) be a BTA with \(L = \lang{\TA}\).
Then \((\TA)^{\db} \equiv \cF{\mathsf{u}}(L)\) if{}f the downward language of every state in \(\TA\) is a union of downward languages of the minimal DBTA for L.
\end{corollary}

It is worth pointing that similarly to path-closedness the conditions of Corollary~\ref{coro:cond2} and Theorem~\ref{theorem:minimalifreverseatomic} are decidable since isomorphism is decidable, \(\cG{\mathsf{u}}(\TA)\) and \((\TA)^{\db}\) are effectively computable and so is \(\cF{\mathsf{u}}(L)\) \cite{tata2007}.

Finally, we give the counterpart result of Theorem~\ref{theorem:minimalifreverseatomic} for co-DBTAs.
Namely, we show a sufficient and necessary condition on a BTA \(\TA\) that guarantees that when we co-determinize it (\(\TA^{\mathsf{cD}}\)) we obtain the minimal co-DBTA for \(\lang{\TA}\).

\begin{theorem}
\label{cominimalifreverseatomic}
Let \(\TA = \tuple{Q,\Sigma,\delta,F}\) be a BTA without unreachable states and such that \(L=\lang{\TA}\) is a path-closed language.
Then \(\cG{\mathsf{d}}(\TA) \equiv \cF{\mathsf{d}}(L)\)  if{}f \(\forall q \in Q \colon   P_{\rd_L}(\mathcal{L}_{\ua}(q)) = P_{\simp}(\mathcal{L}_{\ua}(q))\).
\end{theorem}
\begin{proof}
First, we will show that if \(\cG{\mathsf{d}}(\TA) \equiv \cF{\mathsf{d}}(\TA)\) then \(P_{\rd_L}(\mathcal{L}_{\ua}(q)) = P_{\simp}(\mathcal{L}_{\ua}(q)), \forall q \in Q\).
\begin{align*}
P_{\rd_L}(\mathcal{L}_{\ua}(q)) & = \\
\{x \in \cC_{\Sigma} \mid \exists y \in \mathcal{L}_{\ua}(q) \colon  y^{-1}L = x^{-1}L\} & = \quad \text{[\(\mathord{\rd_L} = \mathord{\rdA}\)]} \\
\{x \in \cC_{\Sigma} \mid \exists y \in \mathcal{L}_{\ua}(q) \colon  \pre^{\TA}_y(F) = \pre^{\TA}_x(F)\} &\enspace .
\end{align*}
Note that \(\mathord{\rd_L} = \mathord{\rdA}\) by hypothesis, since \(\cG{\mathsf{d}}(\TA) \equiv \cF{\mathsf{d}}(L)\).
On the other hand, since \( y \in \mathcal{L}_{\ua}(q) \Ra q \in \pre^{\TA}_y(F)\), we have the following set inclusion:
\begin{align*}
\{x \in \cC_{\Sigma} \mid \exists y \in \mathcal{L}_{\ua}(q) \colon  \pre^{\TA}_y(F) = \pre^{\TA}_x(F)\} & \subseteq \\
\{x \in \cC_{\Sigma} \mid q \in \pre^{\TA}_x(F)\} & = \\
P_{\simp}(\mathcal{L}_{\ua}(q)) \enspace .
\end{align*}
Since \(L\) is path-closed and \(\cA\) has no unreachable states, by Lemma~\ref{lemma:simpImpliesQuotientEqual}, \(x \simp y \Ra x^{-1}L = y^{-1}L\), for every \(x,y \in \cC_{\Sigma}\).
Therefore, \(P_{\simp}(\mathcal{L}_{\ua}(q)) \subseteq P_{\rd_L}(\mathcal{L}_{\ua}(q))\), hence we conclude that \(P_{\rd_L}(\mathcal{L}_{\ua}(q)) = P_{\simp}(\mathcal{L}_{\ua}(q))\).

\medskip
Now, we will prove that if \(P_{\rd_L}(\mathcal{L}_{\ua}(q)) = P_{\simp}(\mathcal{L}_{\ua}(q))\), for each \(q \in Q\), then  \(\cG{\mathsf{d}}(\TA)\equiv \cF{\mathsf{d}}(L)\).  %
For every \(x \in \cC_\Sigma\):
\begin{align}
P_{\rd_{\TA}}(x) & =^{\dagger} \notag  \\
\bigcap_{\mathclap{q \in \pre^{\TA}_x(F)}} P_{\simp}(\mathcal{L}_{\ua}(q)) \; \cap \bigcap_{\mathclap{q \notin \pre^{\TA}_x(F)}}(P_{\simp}(\mathcal{L}_{\ua}(q)))^{\complement} & =^{\dagger\dagger} \notag \\
\bigcap\limits_{\substack{q \in \pre^{\TA}_x(F)}} P_{\rd_L}(\mathcal{L}_{\ua}(q)) \;\; \cap \bigcap\limits_{\substack{q \notin \pre^{\TA}_x(F)}} (P_{\rd_L}(\mathcal{L}_{\ua}(q)))^{\complement} \label{eq:union_blocks} &\enspace .
\end{align}
Note that equality \(\dagger\) holds by Lemma~\ref{blockBTAasIntersection} and \(\dagger\dagger\) holds since \(P_{\simp}(\mathcal{L}_{\ua}(q)) = P_{\rd_L}(\mathcal{L}_{\ua}(q))\) by hypothesis.

It follows from~\eqref{eq:union_blocks} that \(P_{\rd_{\TA}}(x)\) is a \emph{union} of blocks of \(P_{\rd_L}\).
In other words, \(x~\mathord{\rd_L}~y \Ra x~\mathord{\rdA}~y\), for every \(x,y \in \cC_{\Sigma}\).
By Lemma~\ref{automataCongruences}\ref{lemma:automataCongruences:downAT}, we have that \(x~\mathord{\rdA}~y \Ra x~\mathord{\rd_L}~y \).
Thus, \(P_{\rd_{\TA}}(x)\) necessarily corresponds to one single block of \(\mathord{\rd_L}\), namely, \(P_{\rd_L}(x)\).
Since \(P_{\rd_{\TA}}(x) = P_{\rd_L}(x)\), for each \(x \in \cC_\Sigma\), we conclude that \(\cG{\mathsf{d}}(\TA)\equiv \cF{\mathsf{d}}(L)\).
\end{proof}

\begin{corollary}
\label{coro:condiii}
Let \(\TA\) be a BTA with \(L = \lang{\TA}\).
Then \((\TA)^{\cdb} \equiv \cF{\mathsf{d}}(L)\) if{}f the set of contexts root-to-pivot equivalent to the upward language of every state in \(\TA\) is a union of upward languages of the minimal co-DBTA for \(L\).
\end{corollary}

It is worth pointing that the conditions of Corollary~\ref{coro:condiii} and Theorem~\ref{cominimalifreverseatomic} are decidable since isomorphism is decidable, \(\cG{\mathsf{d}}(\TA)\) and \((\TA)^{\cdb}\) are effectively computable and so is \(\cF{\mathsf{d}}(L)\) \cite[\S~2.11]{gecseg2015tree}.

\section{Related work and conclusions}
In this paper, we build on previous work on word automata~\cite{ganty2019congruence} and present a congruence-based perspective on the determinization and minimization operations for bottom-up (and top-down, see Appendix~\ref{appendix:TTA}) tree automata.
As a consequence, we obtain the first, to the best of our knowledge, proof of correctness of the double-reversal method for BTAs.
Bj{\"{o}}rklund and Cleophas, in their taxonomy of minimization algorithms for tree automata~\cite{Bjorklund2016Taxonomy}, proposed a double-reversal method for DBTAs.
They observed that the reverse operation is embedded within the notion of top-down and bottom-up determinization although they did not include a proof of correctness of the algorithm.

Courcelle et al.~\cite{courcelle1991geometrical} also studied the problem of determinizing and minimizing word and tree automata by offering a geometrical and general view on these operations.
Roughly speaking, our framework is an instantiation of theirs using a concrete decomposition of their binary relations (namely, deterministic and co-deterministic decompositions).
However, they focus on the so-called \emph{lr-determinism} of BTAs, which is a relaxed version of our notion of \emph{co-determinism} for BTAs that is defined in terms of the downward languages of the states of the BTA instead of being a purely syntactic notion, as our definition.
As a consequence, the class of languages that are \emph{lr-determinizable}, i.e., the so-called \emph{homogeneous languages}, includes the path-closed languages.
While theirs is a more general setting, our framework is constructive, in the sense that, it is defined upon congruences that allows us to extract automata constructions.

\begin{figure}[!ht]
\vspace*{-2mm}
\centering
\begin{minipage}[l]{0.46\textwidth}
\begin{tikzcd}[column sep=3.5em, row sep=normal]
{\BTA} \ar[d, "R",leftrightarrow]  \ar[start anchor=70, end anchor=110, rr, bend left=20, "\cF{\mathsf{u}}"] \ar[r, "\cG{\mathsf{d}}" ] & \cG{\mathsf{d}}(\BTA) \ar[d, "R",leftrightarrow] \ar[r, "\cG{\mathsf{u}}"] & \cG{\mathsf{u}}(\cG{\mathsf{d}}(\BTA)) \ar[d, "R",leftrightarrow]\\
{\TTA} \ar[r, "\cJ{\mathsf{d}}"] \ar[start anchor=290, end anchor=250, rr, bend right=20, "\cK{\mathsf{u}}"] & \cJ{\mathsf{d}}(\TTA) \ar[r, "\cJ{\mathsf{u}}"] & \cJ{\mathsf{u}}(\cJ{\mathsf{d}}(\TTA))
\end{tikzcd}
\end{minipage}\\
\medskip
\begin{minipage}[r]{0.46\textwidth}
\begin{tikzcd}[column sep=3.5em, row sep=normal]
{\TTA} \ar[d, "R",leftrightarrow]  \ar[start anchor=70, end anchor=110, rr, bend left=20, "\cK{\mathsf{d}}"] \ar[r, "\cJ{\mathsf{u}}" ] & \cJ{\mathsf{u}}(\TTA) \ar[d, "R",leftrightarrow] \ar[r, "\cJ{\mathsf{d}}"] & \cJ{\mathsf{d}}(\cJ{\mathsf{u}}(\TTA)) \ar[d, "R",leftrightarrow]\\
{\BTA} \ar[r, "\cG{\mathsf{u}}"] \ar[start anchor=290, end anchor=250, rr, bend right=20, "\cF{\mathsf{d}}"] & \cG{\mathsf{u}}(\BTA) \ar[r, "\cG{\mathsf{d}}"] & \cG{\mathsf{d}}(\cG{\mathsf{u}}(\BTA))
\end{tikzcd}
\end{minipage}\vspace*{-2mm}
\caption{Relations between the automata constructions \(\cG{\mathsf{d}}, \cG{\mathsf{u}}\), \(\cJ{\mathsf{d}},\cJ{\mathsf{u}},\cF{\mathsf{d}},\cF{\mathsf{u}}, \cK{\mathsf{d}}\) and \(\cK{\mathsf{u}}\).
Note that constructions \(\cF{\mathsf{d}}, \cF{\mathsf{u}}, \cK{\mathsf{d}}\) and \(\cK{\mathsf{u}}\) are applied  to the language defined by the automaton in the origin of the labeled arrow, while the others are applied directly to the automaton.
The upper arcs of the diagrams follow from Theorem~\ref{theoremF}\ref{theoremF:DoubleReversalDBTA} and Theorem~\ref{corolK}\ref{corolK:DoubleReversal}.
The squares and the bottom arc follow from Corollary~\ref{corollary:HHR}, since \(\TTA = (\BTA)^R\).
Incidentally, the diagram shows a new relation which follows from the definition of \(\cH^{\mathsf{uR}}, \cH^{\mathsf{dR}}\) and the fact that \(\TTA = (\BTA)^R\): \(\cJ{\mathsf{d}}(\cJ{\mathsf{d}}(\TTA)) \equiv \cK{\mathsf{u}}(\TTA)\), the minimal co-deterministic TTA for \(\lang{\TTA}\).}
\label{Figure:diagramAutomata}\vspace*{-2mm}
\end{figure}

We also give a generalization of the double-reversal method in the same lines as the generalization of Brzozowski and Tamm~\cite{brzozowski2014theory} for the case of word automata, which further evidences the connection between congruences and determinization of automata.
Note that the double-reversal method only applies to path-closed languages since it requires a co-determinization step, which is only possible for that class of languages.
However, the generalized double-reversal method drops this restriction since, for every regular language, its minimal DBTA is already an automaton satisfying the condition of the generalized double-reversal method.

\medskip
Figure~\ref{Figure:diagramAutomata} summarizes the relations between these tree automata constructions.
Note that it includes the counterpart TTA congruence-based constructions (\(\cJ{\mathsf{d}}\), \(\cJ{\mathsf{u}}\), \(\cK{\mathsf{d}}\)  and \(\cK{\mathsf{u}}\)) whose definition we deferred to the Appendix.

\medskip
As a final note, we show how previous results on word automata follow from our results when we only consider monadic trees.
First, note that, in the monadic case:
\begin{enumerate}
\renewcommand\labelenumi{\theenumi}
\renewcommand{\theenumi}{(\roman{enumi})}
\itemsep=0.9pt
\item every tree language is path-closed;
\item every downward congruence is strongly downward w.r.t. a given tree language, and
\item the notion of root-to-pivot equivalence between contexts collapes to the standard notion of equality.
\end{enumerate}
A consequence of the latter is that our definition of \(\pre(\cdot)\) for tree automata coincides with the standard one for word automata.
As a result of these observations, in the monadic case, Corollary~\ref{thm:Brzozowski-method} (Brzozowski’s double-reversal method for BTAs) collapses to the counterpart result for word automata~\cite{brzozowski1962canonical}.
Finally, Theorem~\ref{theorem:minimalifreverseatomic} (generalization of the double-reversal method) also collapses to the counterpart generalization for word automata given by Brzozowski and Tamm~\cite{brzozowski2014theory}, later addressed in~\cite{ganty2019congruence}.\vspace*{-1mm}

\subsection{Future work}
We have not considered yet how to build an automaton that satisfies the condition of Theorem~\ref{theorem:minimalifreverseatomic}.
In particular, it is worth considering whether lr-deterministic automata satisfy this condition.

On the other hand, Ganty et al.~\cite{ganty2020residual} showed that quasiorders, i.e., reflexive and transitive binary relations, are related to residual word automata in the same way congruences are related to deterministic word automata.
We believe that relaxing the equivalences presented in this work to obtain quasiorders would allow us to offer a new perspective on residual tree automata, defined by Carme et al.~\cite{carme2003residual}.\vspace*{-1mm}

\subsection*{Acknowledgments}
The authors thank the reviewers for their insightful comments. Thank for the supporting projects through grants and scholarships:

Pierre Ganty was partially supported by the Madrid regional project S2018 / TCS-4339 BLOQUES, the
Spanish project PGC2018-102210-B-I00 BOSCO and the Ramón y Cajal Fellowship RYC-2016-20281.

Elena Gutierrez was partially supported by a grant BES-2016-077136 from the Spanish Ministry of Economy, Industry and Competitiveness.

\appendix
\section{Appendix}

\subsection{Top-down tree automata as reverse of bottom-up tree automata}
\label{appendix:TTA}
In the main part of the document, we focus   on bottom-up tree automata.
This decision is motivated by the fact that deterministic BTAs are strictly more expressive than their top-down counterparts.
On the other hand, TTAs can be interpreted as the reverse of BTAs, and thus we can easily extend the results obtained for bottom-up tree automata to their top-down counterparts.
This will be the goal of this section.

Our ultimate purpose is to give a complete perspective on Brzozowski's method to minimize BTAs as a technique that combines a reverse operation followed by determinization operation twice (see Figure~\ref{Figure:diagramAutomata}).
As a product, we obtain a counterpart method for the minimization of TTAs.
In the following, we use \(\BTA\) to denote a BTA and \(\TTA\) for a TTA.

\subsubsection{Top-down tree automata}
\begin{definition}[Top-down tree automaton]
A \emph{top-down tree automaton} (TTA for short) is a tuple \(\TTA = \tuple{Q, \Sigma, \delta, I}\) where \(Q\) is a finite set of \emph{states}; \(\Sigma\) is a ranked alphabet of rank \(n\); \(\delta: Q \to \wp(\bigcup_{i=0}^n \Sigma_i \times Q^i)\) is the \emph{transition function} and \(I \subseteq Q\) is the set of \emph{initial states}.
Given a TTA \(\TTA\), we define its set of \emph{final states} as \(\finals(\TTA) \udiffg  \{q \in Q \mid \exists a \in \Sigma_0 \colon a[\,] \in \delta(q)\}\).
\end{definition}

A TTA is \emph{deterministic} (DTTA for short) if{}f \(I\) is a singleton and for every state \(q \in Q\) and symbol \(f \in \Sigma_n\), with \(n \geq 1\), we have: %
if \(f[q_1,\ldots,q_{\rank{f}}]\in\delta(q)\) and \(f[q'_1,\ldots,q'_{\rank{f}}]\in\delta(q)\) then \(q_i=q'_i\) for each \(i=\range{\rank{f}}\).
Similarly, a TTA is \emph{co-deterministic} (co-DTTA for short) if{}f every set of states in the image of \(\delta^{-1}\) is a singleton or is empty.

We define the \emph{move} relation on a TTA, denoted by \(\xrightarrow{}_{\TTA} \in  \cT_{\Sigma \cup Q} \times \cT_{\Sigma \cup Q}\), as follows.
Let \(t = \subs{x}{q[t_1,\ldots,t_{\rank{f}}]}\) and \(t' = \subs{x}{f[t_1,\ldots,t_{\rank{f}}]}\) for some \(x \in \cC_{\Sigma\cup Q}, q \in Q, f \in \Sigma\) and \(t_1,\ldots,t_{\rank{f}} \in \cT_{Q}\).
Then \(t \ra_{\TTA} t'  \udiff f[t_1(\varepsilon),\ldots,t_{\rank{f}}(\varepsilon)]\in \delta(q)\).
We use \(\to^*_{\TTA}\) to denote the transitive closure of \(\xrightarrow{}_{\TTA}\).
The \emph{language} defined by \(\TTA\) is \(\lang{\TTA} \udiffg  \{t \in \cT_{\Sigma} \mid \exists t' \in \cT_{Q} \colon  t'(\varepsilon) \in I, t' \to^*_{\TTA} t\}\).
The language definition might seem counterintuitive because it starts with a tree \(t'\) in \(\cT_Q\) and ends up with a tree \(t\) in \(\cT_\Sigma\) which is the one ``recognized'' by the TTA.
Intuitively, a run of a TTA accepting a tree \(t\in\cT_\Sigma\) first ``guesses'' for each node \(v \in \dom(t)\) a state labelling \(v\) and then tries to reinstate the original labels in a top down fashion using the move relation.
We made this unusual choice to maintain coherence with the move relation for BTA.
In the example below we depicted such a sequence of moves from right to left because we think it is easier for the reader to parse the sequence of moves backwards.

\begin{example}\label{example:TTA}
Let \(\TTA = \tuple{Q,\Sigma_0 \cup \Sigma_2,\delta,I}\) be a TTA with \(Q = \{q_0,q_1\}\), \(\Sigma_0 =\{\true,\false\}\), \(\Sigma_2 = \{\AND,\OR\}\), \(I = \{q_1\}\), and
\begin{align*}
  \delta(q_0) &= \{\AND[q_0,q_1], \AND[q_0,q_0], \AND[q_1,q_0], \OR[q_0,q_0], \false[\,]\},\\
  \delta(q_1) &=  \{\AND[q_1,q_1], \OR[q_1,q_0],  \OR[q_0,q_1], \OR[q_1,q_1], \true[\,]\} \enspace .
\end{align*}
As in Example~\ref{example:BTA}, \(\lang{\TTA}\) is defined as the set of all trees of the form \(t \in \cT_{\Sigma_0 \cup \Sigma_2}\) which yield to propositional formulas, over the binary connectives \(\land\) and \(\lor\) and the constants \(\true\) and \(\false\), that evaluate to \(\true\).
For instance, the following is a sequence of moves accepting a tree \(t \in \cT_{\Sigma_0 \cup \Sigma_2}\) (listed last below) such that \(t\in \lang{\TA}\).

\noindent
\centering
\tree{15}{30}{
\Tree
[.{\(q_1\)}
  [.{\(q_1\)} {\(q_1\)} {\(q_0\)} ]
  [.{\(q_1\)} ]
]
}
\tmoveT{15}%
\tree{15}{30}{
\Tree
[.{\OR}
  [.{\(q_1\)} {\(q_1\)} {\(q_0\)} ]
  [.{\(q_1\)} ]
]
}
\tmoveT{15}%
\tree{15}{30}{
\Tree
[.{\OR}
  [.{\OR} {\(q_1\)} {\(q_0\)} ]
  [.{\(q_1\)} ]
]
}%
\tmoveT{15}%
\tikz{\node at (0,0) {};
\node at (0,15pt) {\(\ldots\)};}%
\tmoveT{15}%
\tree{15}{30}{
\Tree
[.{\OR}
  [.{\OR} {\(\true\)} {\(\false\)} ]
  [.{\(\true\)} ]
]
}
\tdot{15}

Note that, \(\subs{\OR[\cX,q_1]}{q_1[q_1,q_0]} \ra_{\TTA} \subs{\OR[\cX,q_1]}{\OR[q_1,q_0]}\) as \(\OR[q_1,q_0] \in \delta(q_1)\).
Observe that \(\TTA\) is co-deterministic, but not deterministic since \(\land[q_0, q_1], \land[q_1,q_0] \in \delta(q_0)\).
 \eox
\end{example}

For each \(q \in Q\) and \(S \subseteq Q\), define the \emph{upward} and \emph{downward language} of \(q\), respectively, as follows:
\begin{align*}
  \mathcal{L}^{\TTA}_{\ua}(q, S) &\udiffg  \{ c \in \cC_{\Sigma} \mid \exists t \in \cT_Q \colon t \to^*_{\TTA} \subs{c}{q}, t(\varepsilon) \in S \}\\
  \mathcal{L}^{\TTA}_{\da}(q, S) &\udiffg  \{ t \in \cT_{\Sigma} \mid \exists t' \in \cT_Q \colon  t' \to^*_{\TTA} t, t'(\varepsilon) = q, \ell(t') \subseteq S\}\enspace.
\end{align*}
We will simplify the notation and write \( \mathcal{L}^{\TTA}_{\ua}(q)\) when \(S = I\) and \( \mathcal{L}^{\TTA}_{\da}(q)\) when \(S = \finals(\TTA)\).
Also, we will drop the superscript \(\TTA\) when the TTA \(\TTA\) is clear from the context.

A state \(q \in Q\) of a TTA is \emph{unreachable} (resp.\ \emph{empty}) if{}f its upward (resp.\ downward) language is empty.
It is straightforward to check that for every TTA \(\TTA\) we have that \(\lang{\TTA} = \bigcup_{q\in I}\mathcal{L}_{\da}(q)\).

Given a TTA \(\TTA = \tuple{Q, \Sigma, \delta, I}\) without empty states, the \emph{top-down determinization}~\cite{cleophas2008tree} builds the DTTA \(\tuple{\wp(Q), \Sigma, \delta', \{I\}}\) where the transition function \(\delta'\) is defined as follows.
For each \(f \in \Sigma \setminus \Sigma_0\) and \(R \in \wp(Q)\), we have that \(f[R_1,\ldots,R_{\rank{f}}] \in \delta'(R) \), where
\(R_i \udiffg  \{q_i \!\in\! Q \mid \exists q \!\in\! R, q_1,\ldots, q_{i{-}1},q_{i{+}1},\ldots,q_{\rank{f}} \!\in\! Q \colon  f[q_1,\ldots,q_i, \ldots, q_{\rank{f}}] \in \delta(q)\}\).
On the other hand, for every \(a \in \Sigma_0\) and \(R \in \wp(Q)\)  such that \(\exists q\in R \colon  a \in \delta(q)\), we have that \(a[\,] \in \delta'(R)\).
We abuse notation (w.r.t. the notation introduced for BTAs) and denote by \((\TTA)^{\db}\) the result of applying the top-down determinization to \(\TTA\).

As shown by Cleophas~\cite{cleophas2008tree}, the automaton  \((\TTA)^{\db}\) is top-down deterministic and, whenever \(\lang{\TTA}\) is path-closed, \(\lang{\TTA} = \lang{(\TTA)^{\db}}\).

\subsubsection{Reverse tree automata constructions}
\label{sec:appendix-reverseTAs}
We define the \emph{reverse TTA} of a BTA \(\BTA = \tuple{Q,\Sigma,\delta,F}\) as \((\BTA)^R \udiffg  \tuple{Q,\Sigma,\delta_r,I}\), where $f[q_1,\ldots$, $q_{\rank{f}}] \in \delta_r(q) \udiff q \in \delta(f[q_1,\ldots,q_{\rank{f}}])\) and \(I \udiffg  F$.
Analogously, we define the \emph{reverse BTA} of a TTA \(\TTA = \tuple{Q,\Sigma,\delta,I}\) as \((\TTA)^R \udiffg  \tuple{Q,\Sigma,\delta_r,F}\), where \(q \in \delta_r(f[q_1,\ldots,q_{\rank{f}}]) \udiff f[q_1,\ldots,q_{\rank{f}}] \in \delta(q)\) and \(F \udiffg  I\).
Observe that the BTA and TTA shown in Examples~\ref{example:BTA} and~\ref{example:TTA} are the reverse of each other.

Let \(\BTA\) be a BTA and let \(\TTA\) be its reverse, i.e., \(\TTA \udiffg  (\BTA)^R\).
It is easy to check that \(\lang{\BTA} = \lang{\TTA}\) and \(\BTA\) is a co-DBTA (resp.\ DBTA) if{}f \(\TTA\) is a DTTA (resp.\ co-DTTA).
Moreover, we have that \((\BTA)^{\mathsf{cD}} \equiv (\TTA)^{\mathsf{D}}\) and a state of \(\BTA\) is unreachable (resp.\ empty) if{}f it is empty (resp.\ unreachable) in \(\TTA\).

Note that we can define the TTA-equivalent of the constructions  \(\cH^{\mathsf{u}}\) and \(\cH^{\mathsf{d}}\) simply by reversing their transition functions and setting their final states as the initial ones.
\begin{definition}%
Let \(L \subseteq \cT_\Sigma\), and \(\ru\) and \(\rd\) be an upward and downward congruence, respectively.
Define:
\begin{align*}
\cH^{\mathsf{uR}} \udiffg \cH^{\mathsf{u}}(\ru, L))^R \quad\quad\quad \cH^{\mathsf{dR}} \udiffg  (\cH^{\mathsf{d}}(\rd, L))^R
\end{align*}
\end{definition}
Clearly, \(\cH^{\mathsf{uR}}\) yields a co-DTTA if{}f \(\cH^{\mathsf{u}}\) yields a DBTA, and \(\cH^{\mathsf{dR}}\) yields a DTTA if{}f \(\cH^{\mathsf{d}}\) yields a co-DBTA.
Thus, the following result is a consequence of Lemma~\ref{HBPreservesL}, and the one after is a consequence of Lemmas~\ref{cor:HTPreservesL} and~\ref{downwardIsDet}.
\begin{corollary}
\label{cor:HTPreservesL}
Let \(L \subseteq \cT_{\Sigma}\) be a tree language and let \(\ru\) be an upward congruence such that \(t\ru r \Ra L\,t^{-1} = L\,r^{-1}\) for every \(t,r \in \cT_\Sigma \).
Then,  \(\cH^{\mathsf{uR}}(\ru,L)\) is a co-DTTA with \(\lang{\cH^{\mathsf{uR}}(\ru,L)} = L\).
\end{corollary}

\begin{corollary}
\label{cor:downward_is_det}
Let \(L \subseteq \cT_{\Sigma}\) be a path-closed language and let \(\rd\) be a strongly downward congruence w.r.t. \(L\) such that \(x \rd y \Ra x^{-1}L = y^{-1}L\) for every \(x,y \in \cC_{\Sigma}\).
Then, \(\cH^{\mathsf{dR}}(\rd,L)\) is a DTTA with \(\lang{\cH^{\mathsf{dR}}(\rd,L)} = L\).
\end{corollary}
Next, we define congruences based on the states of a given TTA.
These TTA-based congruences are finer than (or equal to) the corresponding language-based ones and are thus said to \emph{approximate} the language-based congruences.
To that end, we first define the \(\post(\cdot)\) and \(\pre(\cdot)\) operators for TTAs.

\begin{definition}
\label{def:postpre-TTA}
Let \(\TTA = \tuple{Q,\Sigma,\delta,I}\) be a TTA and let \(t \in \cT_\Sigma\), \(x \in \cC_\Sigma\) and \(S \subseteq Q\).
Define:
\begin{align*}
\post^{\TTA}_x(S) & \udiffg  \{q \in Q \mid x \in P_{\simp^S}(\mathcal{L}^{\TTA}_{\ua}(q,S))\} \\
\pre^{\TTA}_t(S) & \udiffg  \{q \in Q \mid t \in \mathcal{L}^{\TTA}_{\da}(q, S)\} \enspace .
\end{align*}
\end{definition}
Note that we will omit the superscript \(\TTA\) when it is clear from the context.
\begin{definition}\label{def:TTA-Equivalences}
Let \({\TTA} = \tuple{Q,\Sigma,\delta,I}\) be a TTA and let \(t, r \in \cT_{\Sigma}\) and \(x, y \in \cC_{\Sigma}\).
The \emph{upward and downward TTA-based equivalences} are respectively given by:
\begin{align*}
  t \ruAT r &\udiff \pre_{t}(\finals({\TTA})) = \pre_{r}(\finals({\TTA}))\\
  x \rdAT y &\udiff \post_{x}(I) = \post_{y}(I)\enspace.
\end{align*}
\end{definition}

\begin{lemma}
\label{congruencesCoincideTwo}
Let \({\BTA}\) be a BTA and \(\TTA \udiffg  (\BTA)^R\).
Then, the following hold:
\begin{enumerate}
\renewcommand\labelenumi{\theenumi}
\renewcommand{\theenumi}{(\alph{enumi})}
\item \(\mathord{\ru_{\lang{\BTA}}} = \mathord{\ru_{\lang{\TTA}}}\) and \(\mathord{\rd_{\lang{\BTA}}} = \mathord{\rd_{\lang{\TTA}}}\).
\item \(\mathord{\ruAB} = \mathord{\ru_{\TTA}}\) and \(\mathord{\rdAB} = \mathord{\rd_{\TTA}}\)\label{lemma:congruencesCoincide-b}.
\end{enumerate}
\end{lemma}
\begin{proof}
Let \({\BTA} = \tuple{Q,\Sigma,\delta,F}\) be a BTA with \(\TTA = (\BTA)^R\).
\begin{enumerate}
\renewcommand\labelenumi{\theenumi}
\renewcommand{\theenumi}{(\alph{enumi})}
\item \(\mathord{\ru_{\lang{\BTA}}} = \mathord{\ru_{\lang{\TTA}}}\) and \(\mathord{\rd_{\lang{\BTA}}} = \mathord{\rd_{\lang{\TTA}}}\).

Trivial, since \(\lang{\BTA} = \lang{\TTA}\).

\item \(\mathord{\ruAB} = \mathord{\ru_{\TTA}}\) and \(\mathord{\rdAB} = \mathord{\rd_{\TTA}}\).

Relying on Definitions~\ref{def:postpre} and~\ref{def:postpre-TTA}, it is easy to see that, for every \(t \in \cT_{\Sigma}, x\in \cC_{\Sigma}\) and \(S \subseteq Q \colon \post_{t}^{{\BTA}}(S)  = \pre_{t}^{\TTA}(S)\) and \(\post_{x}^{\TTA}(S)  = \pre_{x}^{{\BTA}}(S)\).
Thus, by Definitions~\ref{def:automataEquivalences} and~\ref{def:TTA-Equivalences}, \(\mathord{\ruAB} = \mathord{\ru_{\TTA}}\) and \(\mathord{\rdAB} = \mathord{\rd_{\TTA}}\).
\end{enumerate}

\vspace*{-6mm}
\end{proof}

As a consequence of Lemma~\ref{automataCongruences} and~\ref{congruencesCoincideTwo}\ref{lemma:congruencesCoincide-b}, we have the following corollary.

\begin{corollary}\label{corol:upwardDownward}
Let \({\TTA}\) be a TTA with \(L = \lang{{\TTA}}\).
Then,
\begin{enumerate}
\renewcommand\labelenumi{\theenumi}
\renewcommand{\theenumi}{(\alph{enumi})}
\item \(\ruAT\) is an upward congruence and  \(\mathord{\ruAT} \subseteq \mathord{\ru_L}\).\label{corol:automataCongruences:upAT}
\item If \(\TA\) has no empty states and \(L\) is path-closed then \(\rdAT\) is a strongly downward congruence w.r.t. \(L\) and \( \mathord{\rdAT} \subseteq \mathord{\rd_L} \). \label{corol:automataCongruences:downAT}
\end{enumerate}
\end{corollary}

\subsubsection{Determinization and minimization of TTAs using congruences}
We can now define the TTA-equivalents of the constructions from Definition~\ref{def:FG}.
\begin{definition}%
\label{def:JK}
Let \({\TTA}\) be a TTA  with \(L = \lang{{\TTA}}\).
Define:
\begin{align*}
\cJ{\mathsf{u}}({\TTA}) & \udiffg  \cH^{\mathsf{uR}}(\ruAT, {\TTA}) & \cK{\mathsf{u}}(L) & \udiffg  \cH^{\mathsf{uR}}(\ru_L, L) \\
\cJ{\mathsf{d}}({\TTA}) & \udiffg  \cH^{\mathsf{dR}}(\rdAT, {\TTA}) & \cK{\mathsf{d}}(L) & \udiffg  \cH^{\mathsf{dR}}(\rd_L, L) \enspace .
\end{align*}
\end{definition}

It follows from Lemma~\ref{congruencesCoincideTwo} that the constructions \(\cJ{\mathsf{u}}({\TTA})\) and \(\cJ{\mathsf{d}}({\TTA})\) are related to \(\cG{\mathsf{u}}(\BTA)\) and \(\cG{\mathsf{d}}(\BTA)\), respectively, through the reverse construction.
The same holds for the constructions \(\cK{\mathsf{u}}(L)\), \(\cK{\mathsf{d}}(L)\) and \(\cF{\mathsf{u}}(L)\), \(\cF{\mathsf{d}}(L)\), respectively.

\begin{corollary}\label{corollary:HHR}
Let \(\BTA\) be a BTA and \(\TTA \udiffg  (\BTA)^R\) with \(L = \lang{\TTA} = \lang{\BTA}\).
Then,
\begin{align*}
\cJ{\mathsf{u}}({\TTA}) &\equiv (\cG{\mathsf{u}}(\BTA))^R & \cK{\mathsf{u}}(L) &\equiv (\cF{\mathsf{u}}(L))^R \\
\cJ{\mathsf{d}}({\TTA}) &\equiv (\cG{\mathsf{d}}(\BTA))^R & \cK{\mathsf{d}}(L) &\equiv (\cF{\mathsf{d}}(L))^R \enspace .
\end{align*}
\end{corollary}

Finally, we obtain the following result as the TTA-equivalent of Theorem~\ref{theoremF}.

\begin{corollary}\label{corolK}
Let \({\TTA}\) be a TTA with \(L = \lang{\TTA}\).
Then the following properties hold:
\begin{enumerate}
\renewcommand\labelenumi{\theenumi}
\renewcommand{\theenumi}{(\alph{enumi})}
\item \(\lang{\cK{\mathsf{u}}(L)} = L  = \lang{\cJ{\mathsf{u}}({\TTA})}\).
\label{corolK:languageTTA}
\item If \(L\) is path-closed then \(\lang{\cK{\mathsf{d}}(L)} = L = \lang{\cJ{\mathsf{d}}({\TTA})}\).\label{corolK:languageTTAMin}
\item\(\cJ{\mathsf{u}}({\TTA}) \equiv (\TTA)^{\cdb}\). \label{corolK:DetismorphiccDTTA}
\item If \(L\) is path-closed and \({\TTA}\) has no empty states then \(\cJ{\mathsf{d}}({\TTA}) \equiv (\TTA)^{\db}\). \label{corolK:coDetismorphicDTTA}
\item \(\cK{\mathsf{u}}\) is isomorphic to the minimal co-DTTA for \(L\).
\item If \(L\) is path-closed then \(\cK{\mathsf{d}}(L)\) is isomorphic to the minimal DTTA for \(L\).\label{corolK:MinimalDTTA}
\item If \(L\) is path-closed then \(\cJ{\mathsf{u}}(\cJ{\mathsf{d}}(\TTA)) \equiv \cK{\mathsf{u}}(L)\).\label{corolK:doubleReversalcoDTTA}
\item If \(L\) is path-closed then \(\cJ{\mathsf{d}}(\cJ{\mathsf{u}}(\TTA)) \equiv \cK{\mathsf{d}}(L)\).\label{corolK:DoubleReversal}
\end{enumerate}
\end{corollary}

It follows from Corollary~\ref{corolK}\ref{corolK:DetismorphiccDTTA},~\ref{corolK:coDetismorphicDTTA} and~\ref{corolK:DoubleReversal} that, given a TTA \(\TA\), we have that \(((((\TTA)^R)^{\db})^R)^{\db}\) is the minimal DTTA for \(\lang{\TTA}\).

\end{document}